\documentclass[letterpaper,journal]{IEEEtran}
\usepackage{amsmath,amsfonts}
\usepackage{algorithmic}
\usepackage{algorithm}
\usepackage{array}
\usepackage[caption=false,font=normalsize,labelfont=sf,textfont=sf]{subfig}
\usepackage{textcomp}
\usepackage{stfloats}
\usepackage[noadjust]{cite}
\usepackage{booktabs}
\usepackage{multirow}
\usepackage{url}
\usepackage{amsthm}
\usepackage{xcolor,arydshln}
\usepackage{array}
\usepackage[table]{xcolor}
\usepackage{bm}
\usepackage{makecell}
\usepackage{mathtools}
\usepackage{colortbl,makecell}
\usepackage{verbatim}
\usepackage{graphicx}
\usepackage{enumitem}
\usepackage{xcolor}
\usepackage{cite}
\usepackage{amsmath,amssymb}
\usepackage{tikz}
\usetikzlibrary{arrows.meta,shapes.geometric,calc}
\usepackage{float}

\newtheorem{theorem}{Theorem}

\newtheorem{lemma}{Lemma}

\newtheorem{assumption}{Assumption}

\definecolor{cenc}{HTML}{0072B2}
\definecolor{cphys}{HTML}{00784D}
\definecolor{closs}{HTML}{D55E00}
\definecolor{chead}{HTML}{CC79A7}
\definecolor{cgry}{HTML}{6E6E6E}

\usepackage{titlesec}

\titlespacing*{\section}{0pt}{1ex}{0.5ex}

\newcommand{\Wfig}{177.5}
\newcommand{\Hfig}{76}

\newcommand{\Wfigcol}{88.9}
\newcommand{\Hfigcol}{85}

\newcommand{\bd}{\mathbf}

\allowdisplaybreaks

\newenvironment{proofsketch}
  {\begin{proof}[Proof sketch]}
  {\end{proof}}

\makeatletter
\renewcommand{\@IEEEprocessthesectionargument}[1]{%
  \@IEEEappendixsavesection*{%
    \raggedright
    \appendixname~\thesection: #1}%
  \addcontentsline{toc}{section}{\appendixname~\thesection: #1}%
}
\makeatother

\begin{document}

\title{GridSFM: A Foundation Model  for Solving \\ AC Optimal Power Flow}
\author{
Luke Bhan\textsuperscript{1, 3 $\ast$},
Weiwei Yang\textsuperscript{2, 3 $\ast$},
Margaret Capetz\textsuperscript{4},
and Baosen Zhang\textsuperscript{3, 4}%
\thanks{$^\ast$ Equal contribution, listed alphabetically. The authors are with  \textsuperscript{1}University of California, San Diego; \textsuperscript{2}Asemic.ai; \textsuperscript{3}Microsoft Research; \textsuperscript{4}University of Washington. Emails: lbhan@ucsd.edu, weiwei@asemic.io, mcapetz@uw.edu, zhangbao@uw.edu }}




\maketitle

\begin{abstract}
We introduce GridSFM, a framework that combines a pretrained foundation model across grid topologies with physics-informed fine-tuning for solving AC Optimal Power Flow (AC-OPF) at scale. It is a $15$ million parameter physics-inspired graph neural network pretrained across $54$ topologies of $500$ to $4{,}000$ buses. Our model attains a $2.45\%$ zero-shot generation-cost error on a $10{,}000$ bus case held-out operating conditions with no degradation as system size grows. Building on this, we pair the pretrained backbone with a physics-informed fine-tuning design based on Newton's method for power flow. With only $100$ solved instances, GridSFM adapts to unseen grids up to $10{,}000$ buses. We show it out performs single topology, dedicated neural network models that are trained more data, both in terms of cost and solver iterations when deployed as warm starting points. 

In designing this foundation model, we overcome the fact that the feasible set for AC-OPF can be disconnected. This is an obstruction that prevents any continuous neural network from approximating the solution map. To do so, we lift the problem and relax its constraints with logarithmically penalized slacks. We prove that the resulting elastic feasible set is contractible, that the AC-OPF minimizers remain minimizers of the elastic problem above an explicit penalty threshold, and that projecting an approximate solution back onto the AC-OPF feasible set is well posed.  We release all models, data, and code so that the community can build on a shared starting point for AC-OPF.
\end{abstract}


\section{Introduction} 
AC Optimal Power Flow (AC-OPF) is a fundamental problem in power system operations and planning. It asks for the least cost generation solution that would satisfy the loads in the system under various network and engineering constraints. This problem is nonconvex and largescale, and the search for practical and efficient algorithms has spurred a large body of work both in academia and industry~\cite{khaloie2025review}. At present, there are a number of nonlinear programming solvers that are used by system operators~\cite{5491276}.

Despite advances in numerical algorithms, changing operating conditions can make AC-OPF substantially harder to solve. Renewables and new large loads are now integrated into the grid, and consequently, they can push the grid to operating points that operators have not seen before, and dramatically influence the performance of AC-OPF solvers. For example, for the \texttt{Texas2k} grid~\cite{7725528}, AC-OPF can be solved in $24$ seconds ($103$ solver iterations) for a nominal load using Ipopt.
But it would take $83$ seconds ($365$ solver iterations) when the load increases by less than $10$\%. Thus, accelerating the solution process remains an active research area.

In the last several years, machine learning (ML) or AI based surrogates for AC-OPF have gained significant attention. Starting by viewing the AC-OPF problem as a mapping from the load to the optimal solutions, a surrogate model (often some type of neural network) can be used to replace the solver. Because feedforward function evaluations are much faster than iterative solvers, significant speedup can be achieved. For a comprehensive list of different neural network architectures and training methods, we refer the reader to \cite{pan_chen_ml_opf_wiki}. 

To be practically useful, an AI-based surrogate should have several features. First, it needs to generalize to a wide range of topologies and operating conditions. Second, the architecture should scale, as speedup becomes more important for larger systems. Third, the speedup obtained should not degrade the quality of the solutions. And finally, it should not require extensive computing resources.

In this paper, we present a foundation model framework, called GridSFM (Grid Small Foundation Model)\footnote{A preliminary version of GridSFM was described in an earlier, non-peer-reviewed technical white paper~\cite{Yan00}. The present paper substantially extends that work with a new backbone, fine-tuning design, and additional experiments.} that achieves all of these goals. It is a fast $15$ million parameter foundation model pretrained across more than $50$ topologies ranging from $500$ to $4,000$ buses. We note GridSFM is quite small compared to some other AI models (hence the `S` in the name), and we show, without significant computational resources, it can be easily adopted in practice. We pair this pretrained backbone with a fine-tuning method that yields both physically consistent outputs and a feasibility projection layer with explicit guarantees. We show that this framework can scale between grids with $500$ to $10,000$ buses and its performance does not degrade as the grid size increases. Our code and models are publicly available at \cite{bhan_gridsfm}. 

Our design is a careful combination of neural network architecture and fine-tuning methodology. A key obstruction to training a foundation model for AC-OPF across multiple grids is that the geometry of the feasible sets can vary drastically and behave poorly. It is well known that the feasible set can become disconnected~\cite{hiskens2001exploring} as the binding inequality constraints change with load variation. Given that neural networks are continuous, this inhibits \emph{any} neural network surrogate that approximates the load to solution mapping~\cite{cybenko1989approximation}. We therefore train against a slacked reformulation of AC-OPF
and prove that it removes this obstruction (Lemma~\ref{lem:elastic-connected}). Moreover, above an explicit penalty threshold, we show the slacked formulation recovers the optimal solution (Lemma~\ref{lem:exact}), and that it admits a unique Lipschitz feasibility projection within a radius computable from the network data and the distance to the power-flow singularity (Theorem~\ref{thm:proj}). 

The main architecture of GridSFM is built from a graph neural network (GNN), consisting of a linear attention layer and a heterogeneous message passing layer \cite{NIPS2017_5dd9db5e} that produce just the generator and voltage setpoints. This choice is deliberate, enabling a model that can be trained with gradient stability across many different topologies while producing batchable downstream outputs in milliseconds at inference time. 

From the pretrained backbone, we then present a fine-tuning framework that transforms  predicted
control setpoints into full, physically consistent solutions. To do so, we 
place a Newton power flow completion inside the fine-tuning loop, with a
slacked reformulation of AC-OPF as the target loss, and back-propagate by combining
the standard adjoint method~\cite{donti2021dc}. We include a residual term that remains differentiable even when the power flow fails to converge. As a result, fine-tuning stays stable even from initially poor setpoints and allows GridSFM to adapt to inputs that are far from the pretraining distribution.

The GridSFM pretrained backbone yields a single model that attains a $2.45\%$
zero-shot cost error on held-out operating conditions across all $54$
pretraining grids. Moreover, combinng with the proposed fine-tuning framework adapts to the unseen $6{,}470$-bus
\texttt{case6470\_rte} \cite{pglib} and $10{,}000$-bus \texttt{ACTIVSg10k} \cite{7725528} systems from
only $100$ solved instances in under $80$ minutes of fine-tuning on a single GPU. Operators can use GridSFM in several modes. For example, when exact feasibility is not required, GridSFM can be deployed for quick screening or planning problems~\cite{4956966} where its output can be used directly. When exact feasibility is important, it can serve as a warm start to a solver. For example, on the \texttt{Texas2k}~\cite{7725528} grid, it reduces the iteration needed by a solver by approximately $7.5\times$. When optimality is not required, a projection algorithm based on Theorem~\ref{thm:proj} delivers feasible operating points at
cost gaps of $0.57\%$ on $10{,}000$-bus systems at even faster speeds.

This paper follows the explosive growth in the last decade of using fast surrogate tools to approximate optimal power flow. 
Originally, work such as DeepOPF \cite{zam19, fio19, pan20} and DC3 \cite{donti2021dc} focused on single-grid designs, learning mappings from loads to the generator controls, and then completing the decomposition with a traditional solver. Variants of these methods have since been developed including unsupervised approaches~\cite{Hua24b,Kim25c} as well as using GNNs~\cite{Don20,Fal21}. 

From this initial set of work, two challenges remained. The first revolved around feasibility, which has been studied in multiple different works~\cite{Zha20,Hua24b,JMLR:v25:23-1577,Kim25c,Ngu25}. Broadly speaking, training losses are designed to penalize infeasible solutions and post-processing is used to repair infeasible solutions. However, guarantees are difficult to obtain, even for a single fixed system.

The second has focused on generalization. For most approaches, whenever grid conditions change out of training distributions or transmission outages modify topologies, neural networks require retraining. Hence, there has been a growing interest in GNN designs that adapt to grid
reconfigurations~\cite{Liu22f} and multiple topologies~\cite{Zho23,Pil24}. This has culminated in a few recent efforts, including ours, to pretrain single backbones across grids, and deploy them for various downstream tasks. For example, \cite{Pue26} proposed unifying power flow, OPF, and state estimation in one architecture. Pretrained models such in \cite{Par26} and \cite{Li26c} present preliminary results in scaling such models across multiple topologies. 

However, these approaches have not been shown to succeed across scale,
with pretraining corpora drawn from at most ten base networks and
without the feasibility certificates operators need at more than $10{,}000$ buses.
GridSFM closes this gap with a foundation model and a paired fine-tuning
framework that yields physically consistent outputs and a feasibility
projection with explicit guarantees.

The paper is organized as follows. In Section \ref{sec:formulation}, we formulate AC-OPF along with its slacked relaxation.
In Section \ref{sec:foundation-model-design}, we present the architectural design of GridSFM.
In Section \ref{sec:finetuning-full}, we develop a fine-tuning design that allows one to explicitly adapt the GridSFM backbone to out-of-distribution grid topologies. In Section \ref{sec:theoretical}, we show that the GridSFM framework is well-posed, proving that feasibility restoration from any GridSFM control output is Lipschitz and unique provided it is close to the AC-OPF minimizer. Lastly, Section \ref{sec:numerical} provides multiple numerical experiments highlighting the pretraining and fine-tuned performance of GridSFM. 

\section{Problem Formulation}\label{sec:formulation}
We work with the standard bus-injection formulation of AC-OPF~\cite{Low2026}.
Let bus $0$ denote the slack bus, $\mathcal N:=\{1,\dots,n_b\}$ be the set of the rest of the buses, and define $\mathcal N_0:=\{0\}\cup\mathcal N$. Let
$\mathcal G\subseteq\mathcal N_0$ and
$\mathcal L\subseteq\mathcal N_0\times\mathcal N_0$ denote the sets of
generator buses and branches, respectively, with
$n_g:=|\mathcal G|$ and $n_\ell:=|\mathcal L|$.
Let $\mathcal L_i$ denote the branches connected to bus $i$. We use 
 $v_i, \theta_i$ to denote the voltage magnitude and angle, $p_i^{\mathrm g}$ and $q_i^{\mathrm g}$ are the
active and reactive generation, and let $p_i^{\mathrm d}$ and $q_i^{\mathrm d}$ be the
fixed demands at bus $i$. 
The AC-OPF problem is given by 
\begin{subequations}\label{eq:AC-OPF}
\begin{align}
\min_{\mathbf v,\boldsymbol\theta,\mathbf p^{\mathrm g},\mathbf q^{\mathrm g}} \quad
& \sum_{i\in\mathcal G}c_i(p_i^{\mathrm g})                    \label{eq:AC-OPF-a}\\
\mathrm{s.t.}\quad
& p_i^{\mathrm g}-p_i^{\mathrm d}=p_i(\mathbf v,\boldsymbol\theta),
&& i\in\mathcal N_0,                                   \label{eq:AC-OPF-b}\\
& q_i^{\mathrm g}-q_i^{\mathrm d}=q_i(\mathbf v,\boldsymbol\theta),
&& i\in\mathcal N_0,                                   \label{eq:AC-OPF-c}\\
& \underline p_i^{\mathrm g}\le p_i^{\mathrm g}\le\overline p_i^{\mathrm g},
&& i\in\mathcal G,                                     \label{eq:AC-OPF-d}\\
& \underline q_i^{\mathrm g}\le q_i^{\mathrm g}\le\overline q_i^{\mathrm g},
&& i\in\mathcal G,                                     \label{eq:AC-OPF-e}\\
& \underline v_i\le v_i\le\overline v_i,
&& i\in\mathcal N_0,   \label{eq:AC-OPF-f}\\
& p_\ell(\mathbf v, \boldsymbol{\theta})^2 + q_\ell(\mathbf v, \boldsymbol{\theta})^2 \le\overline S_\ell^2,
&& \ell\in\mathcal L,                                  \label{eq:AC-OPF-g}
\end{align}
\end{subequations}
where for each branch $\ell$, the branch flows are
\begin{subequations}\label{eq:branch-flows}
\begin{align}
p_\ell(\mathbf v, \boldsymbol{\theta})&=v_i v_k (G_{ik}\cos\theta_{ik}
      +B_{ik}\sin\theta_{ik}), \label{eq:branch-flows-a}\\
q_\ell(\mathbf v, \boldsymbol{\theta})&=v_i v_k (G_{ik}\sin\theta_{ik}
      -B_{ik}\cos\theta_{ik}), \label{eq:branch-flows-b}
\end{align}
\end{subequations}
and  $\theta_{ik} = \theta_i - \theta_k$. The active and reactive power injections are the sum of the outgoing branch flows
\begin{equation}\label{eq:nodal-injections}
p_i=\sum_{\ell \in \mathcal{L}_i} p_\ell, \quad q_i=\sum_{\ell \in \mathcal{L}_i} q_\ell.        
\end{equation}
The quantities $G_{ik},B_{ik}$ in \eqref{eq:nodal-injections} are the real and
imaginary parts of the bus admittance matrix
$\mathbf Y=\mathbf G+\mathrm{j}\mathbf B$.

For subsequent analysis, we collect the equality and inequality
constraints into vector-valued mappings. In particular, the equality constraints are denoted by $\mathbf{g}$ and are the power balance equations in~\eqref{eq:AC-OPF-b} and~\eqref{eq:AC-OPF-c}. The rest of the constraints from~\eqref{eq:AC-OPF-d} to~\eqref{eq:AC-OPF-g} are inequalities and are denoted by $\mathbf{h}$. 
Accordingly, \eqref{eq:AC-OPF} can be written more compactly as
\begin{subequations}\label{eq:AC-OPF-compact}
\begin{align}
\min_{\mathbf v,\boldsymbol\theta,\mathbf p^{\mathrm g},\mathbf q^{\mathrm g}}\quad
& \sum_{i\in\mathcal G}c_i(p_i^{\mathrm g}),                  \label{eq:compact-a}\\
\mathrm{s.t.}\quad
& \mathbf g(\mathbf v,\boldsymbol\theta,\mathbf p^{\mathrm g},\mathbf q^{\mathrm g})=\mathbf 0,                              \label{eq:compact-b}\\
& \mathbf h(\mathbf v,\boldsymbol\theta,\mathbf p^{\mathrm g},\mathbf q^{\mathrm g})\leq\mathbf 0.                           \label{eq:compact-c}
\end{align}
\end{subequations}
We often collect all variables together and refer to them as $\mathbf x$, and \eqref{eq:compact-b} and \eqref{eq:compact-c} becomes $\bd{g}(\bd x)=0$ and $\bd h (\bd x) \leq 0$.

The challenge of learning solutions for \eqref{eq:AC-OPF} is well-understood given the problem is a nonlinear, nonconvex optimization whose feasible set can consist of multiple disconnected components \cite{hiskens2001exploring,zhang2012geometry}. The existence of disconnected components is, in fact, a major barrier in using a neural network for AC-OPF. Neural networks are continuous functions so they map connected sets to connected sets. As a consequence, we cannot in general approximate the solutions of AC-OPF using neural networks.

Therefore, we propose reformulating \eqref{eq:AC-OPF} by using a slacked relaxation, which, as proven in Lemma \ref{lem:elastic-connected}, enables a solution set that is fully connected and directly motivates the physics-informed loss in Section \ref{sec:finetuning-full}. More precisely, we introduce slack variables associated with each of the constraints that leads to an elastic relaxation of the original problem: 
\begin{align}
\min_{\mathbf x,\mathbf s^{ g},\mathbf s^{ h}}
& \sum_{i\in\mathcal G}c_i(p_i^{\mathrm g})
+\rho_{ g}\mathbf 1^\top\log(\mathbf 1+\mathbf s^{ g})
+\rho_{ h}\mathbf 1^\top\log(\mathbf 1+\mathbf s^{ h}) \nonumber\\
\mathrm{s.t.}\quad
& -\mathbf s^{ g}\leq{\mathbf g}(\mathbf x)\leq\mathbf s^{ g}, \; {\mathbf h}(\mathbf x)\leq\mathbf s^{ h}, \;  \mathbf s^{ g}\geq\mathbf 0,\; \mathbf s^{ h}\geq\mathbf 0,                            \label{eq:AC-OPF-elastic}
\end{align}
where we introduce nonnegative slack variables \(\mathbf s^{ g}\) and \(\mathbf s^{ h}\) for the equality and inequality residuals, respectively. If \(\mathbf s^{ g}\) and \(\mathbf s^{ h}\) are zero, we recover the solution to the original problem~\eqref{eq:AC-OPF-compact}.

Here, the logarithm is applied element-wise. Note that other penalty functions could be used, and our choice is motivated by the fact 
 $\log(1+s)$ remains with unit slope at the origin enabling a local minimizer of the original AC-OPF problem \eqref{eq:AC-OPF-compact} to be a minimizer of the slacked problem \eqref{eq:AC-OPF-elastic} as proven in Lemma \ref{lem:exact}. Additionally, its bounded, decaying gradient $\rho/(1+s)$ keeps diverged power-flow residuals from dominating the stochastic-gradient loss during training.

Since a solution $\hat{\mathbf x}$ of \eqref{eq:AC-OPF-elastic} is not necessarily a feasible solution of \eqref{eq:AC-OPF-compact}, we introduce the following projection problem which, using any candidate solution, $\hat{\mathbf x}$ of \eqref{eq:AC-OPF-elastic}, restores feasibility by solving
\begin{equation}\label{eq:AC-OPF-projection}
\min_{\mathbf x}\quad
\frac{1}{2}\lVert \mathbf x-\hat{\mathbf x}\rVert_2, \text{ s.t. }
 {\mathbf g}(\mathbf x)=\mathbf 0, \; {\mathbf h}(\mathbf x)\leq\mathbf 0.                            
\end{equation}

Together, \eqref{eq:AC-OPF-elastic} and \eqref{eq:AC-OPF-projection} form the
optimization scaffolding for the remainder of this work where their theoretical properties will be discussed in Section~\ref{sec:theoretical}. The  next section introduces the key design choices of GridSFM to solve \eqref{eq:AC-OPF-compact}.

\section{Foundation Model Design}\label{sec:foundation-model-design}
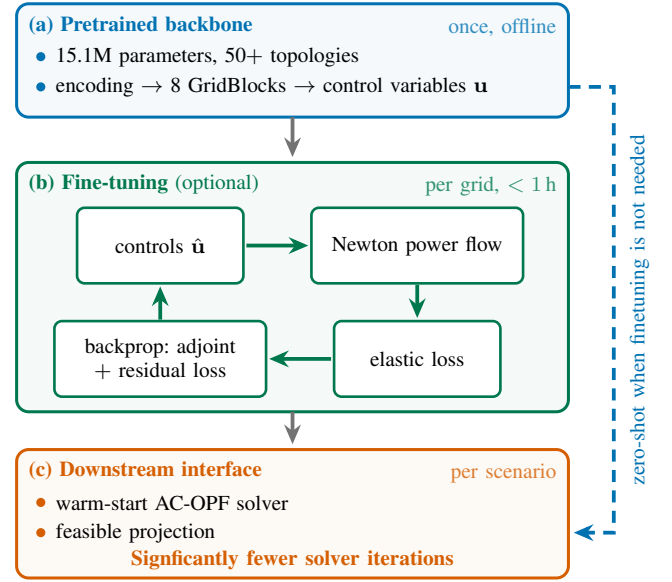
\begin{figure}[ht]
    \centering
\begin{tikzpicture}[x=1mm,y=1mm,
  every node/.style={inner sep=0pt,outer sep=0pt,font=\footnotesize},
  bx/.style={rounded corners=1mm,line width=0.35mm,align=center,
             inner xsep=1.4mm,inner ysep=1.1mm},
  ttl/.style={anchor=north west,align=left,font=\footnotesize\bfseries},
  cad/.style={anchor=north east,align=right,font=\footnotesize},
  li/.style={anchor=west,align=left,font=\footnotesize},
  hl/.style={anchor=center,align=center,font=\footnotesize\bfseries},
  ar/.style={-{Stealth[length=2.6mm,width=2.2mm]},line width=0.5mm,draw=black!55},
  cyc/.style={-{Stealth[length=2.6mm,width=2.2mm]},line width=0.5mm,draw=cphys},
  gar/.style={-{Stealth[length=2.6mm,width=2.2mm]},line width=0.5mm,draw=cenc,
              dash pattern=on 1.6mm off 1.1mm},
]

\useasboundingbox (0,0) rectangle (\Wfigcol,-\Hfigcol);

\foreach \T/\B/\C/\N/\W in {6/21/cenc/{(a)~Pretrained backbone}/{once, offline},
                            27/60/cphys/{(b)~Fine-tuning \normalfont(optional)}/{per grid, $<1$\,h},
                            65/82/closs/{(c)~Downstream interface}/{per scenario}}{
  \draw[rounded corners=1.6mm,draw=\C,line width=0.4mm,fill=\C!4]
        (3,-\T) rectangle (76,-\B);
  \node[ttl,text=\C,text width=50mm] at (4.6,-\T-1.6) {\N};
  \node[cad,text=\C!80] at (74.4,-\T-1.9) {\W};}

\foreach \y/\T in {-13/{15.1M parameters, 50$+$ topologies},
                   -17/{encoding $\rightarrow$ 8 GridBlocks $\rightarrow$
                        control variables $\mathbf{u}$}}{
  \fill[cenc] (6.2,\y) circle (0.6);
  \node[li,text width=66mm] at (8.2,\y) {\T};}

\node[bx,draw=cphys,fill=white,minimum width=22mm,minimum height=10mm,text width=19mm]
  at (22,-38) {controls $\hat{\mathbf{u}}$};
\node[bx,draw=cphys,fill=white,minimum width=28mm,minimum height=10mm,text width=25mm]
  at (56,-38) {Newton power flow};
\node[bx,draw=cphys,fill=white,minimum width=22mm,minimum height=10mm,text width=19mm]
  at (56,-53) {elastic loss};
\node[bx,draw=cphys,fill=white,minimum width=28mm,minimum height=10mm,text width=25mm]
  at (22,-53) {backprop: adjoint\\ $+$ residual loss};
\draw[cyc] (33.4,-38) -- (41.6,-38);
\draw[cyc] (56,-43.2) -- (56,-47.6);
\draw[cyc] (44.6,-53) -- (36.4,-53);
\draw[cyc] (22,-47.8) -- (22,-43.4);

\foreach \y/\T in {-72/{warm-start AC-OPF solver},
                   -76/{feasible projection}}{
  \fill[closs] (6.2,\y) circle (0.6);
  \node[li,text width=66mm] at (8.2,\y) {\T};}
\node[hl,text=closs] at (39.5,-79.6) {Signficantly fewer solver iterations};

\draw[ar] (39.5,-21) -- (39.5,-26.3);
\draw[ar] (39.5,-60) -- (39.5,-64.3);
\draw[gar] (76,-17) -- (82,-17) -- (82,-76) -- (76.7,-76);
\node[rotate=90,anchor=center,text=cenc]
  at (85.4,-46.5) {zero-shot when finetuning is not needed};

\end{tikzpicture}
\caption{The GridSFM framework. (a)~A $15$M-parameter backbone, pretrained offline on $50{+}$ topologies, maps grids and operating conditions to controls $\mathbf u$. (b) For specific grids, the model is then fine-tuned using a physics informed, power flow based design. (c) From here, GridSFM predictions can be used for downstream tasks.}
\label{fig:flow_diag}
\vspace{-1em}
\end{figure}

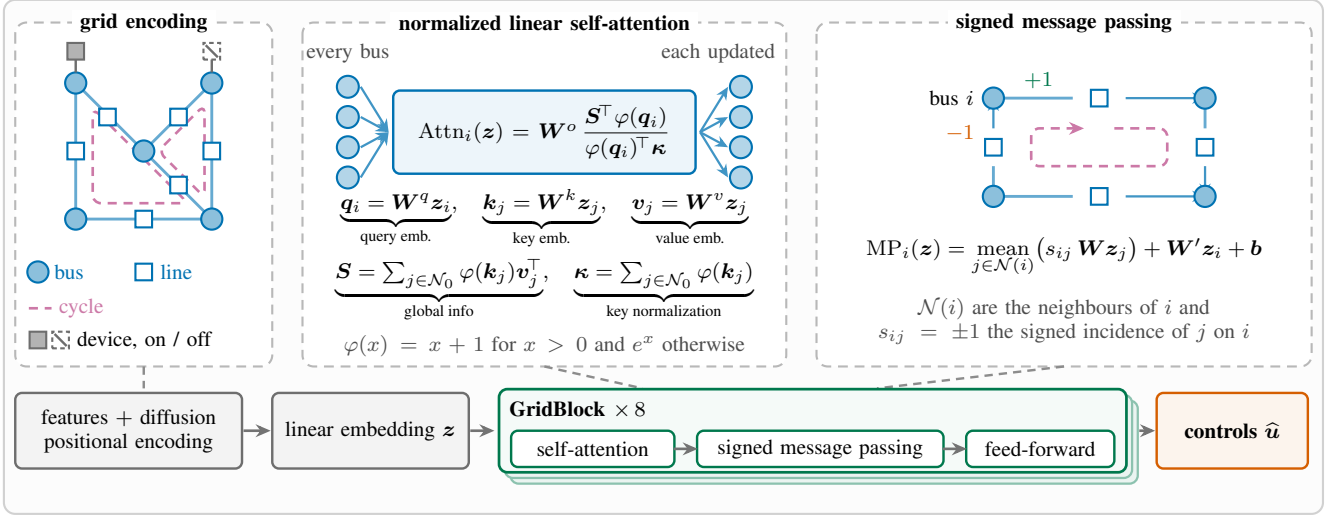
\begin{figure*}[t]
\begin{tikzpicture}[x=1mm,y=1mm,
  every node/.style={inner sep=0pt,outer sep=0pt,font=\footnotesize},
  bx/.style={rounded corners=0.9mm,line width=0.32mm,align=center,
             inner xsep=1.3mm,inner ysep=1mm},
  enc/.style={bx,draw=cenc,fill=cenc!7},
  phys/.style={bx,draw=cphys,fill=cphys!7},
  pre/.style={bx,draw=cgry,fill=cgry!8},
  lss/.style={bx,draw=closs,fill=closs!7},
  ttl/.style={anchor=north west,align=left,font=\small\bfseries},
  sub/.style={anchor=north west,align=left,font=\footnotesize\bfseries},
  note/.style={anchor=north,align=center,font=\footnotesize,text=black!70,
               execute at begin node={\hyphenpenalty=10000\exhyphenpenalty=10000}},
  ar/.style={-{Stealth[length=2mm,width=1.7mm]},line width=0.35mm,draw=black!55},
  zoom/.style={draw=black!45,line width=0.32mm,dash pattern=on 1.1mm off 0.8mm},
  thin ar/.style={-{Stealth[length=1.7mm,width=1.5mm]},line width=0.3mm,draw=cenc!75},
  cyclc/.style={draw=chead,line width=0.35mm,dash pattern=on 1.2mm off 0.9mm,
                rounded corners=1.1mm},
  busn/.style={circle,draw=cenc,fill=cenc!45,line width=0.28mm,minimum size=2.8mm},
  brn/.style={rectangle,draw=cenc,fill=white,line width=0.28mm,minimum size=2.2mm},
  onn/.style={rectangle,draw=cgry,fill=black!30,line width=0.28mm,minimum size=2.2mm},
  offn/.style={rectangle,draw=cgry,fill=white,line width=0.28mm,minimum size=2.2mm,
               dash pattern=on 0.6mm off 0.45mm},
  stub/.style={draw=cgry,line width=0.32mm},
  wire/.style={draw=cenc!60,line width=0.38mm},
  tok/.style={circle,draw=cenc,fill=cenc!30,line width=0.28mm,minimum size=2.9mm},
  sgn/.style={inner sep=0.2mm,fill=white},
]

\useasboundingbox (0,0) rectangle (\Wfig,-\Hfig);
\draw[rounded corners=1.5mm,draw=black!18,line width=0.28mm,fill=black!1]
      (1.5,-1.5) rectangle (176,-69.5);

\begin{scope}[yshift=4.5mm]
\foreach \L/\R/\N in {3/37/{grid encoding}, 41/105/{normalized linear self-attention},
                      109/174.5/{signed message passing}}{
  \draw[rounded corners=1.5mm,draw=black!28,line width=0.28mm,fill=white,
        dash pattern=on 1.3mm off 0.9mm] (\L,-9) rectangle (\R,-54.5);
  \node[fill=white,inner xsep=1.1mm,inner ysep=0.4mm,font=\footnotesize\bfseries]
        at ({\L+0.5*(\R-\L)},-9) {\N};}

\coordinate (c1) at (11,-17);  \coordinate (c2) at (29,-17);
\coordinate (c3) at (20,-26);
\coordinate (c4) at (11,-35);  \coordinate (c5) at (29,-35);
\coordinate (fA) at (18.5,-28); \coordinate (fB) at (26,-26);
\draw[cyclc]
      ($(c1)!0.3!(fA)$)--($(c3)!0.3!(fA)$)--($(c5)!0.3!(fA)$)--($(c4)!0.3!(fA)$)--cycle;
\draw[cyclc] ($(c3)!0.32!(fB)$)--($(c2)!0.32!(fB)$)--($(c5)!0.32!(fB)$)--cycle;
\draw[wire] (c1)--(c3) (c2)--(c3) (c1)--(c4) (c2)--(c5) (c3)--(c5) (c4)--(c5);
\foreach \p in {(15.2,-21.5),(24.6,-21.5),(11,-26),(29,-26),(20,-35),(24.6,-30.5)}
  {\node[brn] at \p {};}
\foreach \p in {(c1),(c2),(c3),(c4),(c5)}{\node[busn] at \p {};}
\draw[stub] (11,-15.6)--(11,-13.8);  \node[onn]  at (11,-12.8) {};
\draw[stub] (29,-15.6)--(29,-13.8);  \node[offn] at (29,-12.8) {};
\draw[cgry,line width=0.32mm] (28.2,-12)--(29.8,-13.6);

\node[busn] at (6,-42) {};  \node[anchor=west,text=cenc] at (8.2,-42) {bus};
\node[brn]  at (20,-42) {}; \node[anchor=west,text=cenc] at (22.2,-42) {line};
\draw[cyclc] (4.8,-46.8) -- (7.8,-46.8);
\node[anchor=west,text=chead] at (8.8,-46.8) {cycle};
\node[onn] at (6,-51) {}; \node[offn] at (9,-51) {};
\draw[cgry,line width=0.32mm] (8.2,-50.2)--(9.8,-51.8);
\node[anchor=west] at (11.2,-51) {device, on / off};

\node[anchor=center,text=black!70] at (47,-13) {every bus};
\node[anchor=center,text=black!70] at (96,-13) {each updated};
\foreach \y in {-17.5,-21.5,-25.5,-29.5}{
  \node[tok] at (47,\y) {};
  \node[tok] at (99,\y) {};
  \draw[thin ar] (48.7,\y) -- (52.4,-23.5);
  \draw[thin ar] (93.6,-23.5) -- (97.3,\y);}
%
\node[enc,minimum width=40mm,minimum height=11mm,text width=37mm]
  at (73,-23.5) {$\operatorname{Attn}_i(\bm z)=\bm W^o\,
                  \dfrac{\bm S^{\!\top}\varphi(\bm q_i)}
                        {\varphi(\bm q_i)^{\!\top}\bm\kappa}$};
\node[anchor=center] at (73,-35)
  {$\underbrace{\bm q_i = \bm W^q\bm z_i}_{\text{query emb.}},\quad
    \underbrace{\bm k_j = \bm W^k\bm z_j}_{\text{key emb.}},\quad
    \underbrace{\bm v_j = \bm W^v\bm z_j}_{\text{value emb.}}$};
\node[anchor=center] at (73,-44.5)
  {$\underbrace{\bm S
      =\textstyle\sum_{j\in\mathcal N_0}\varphi(\bm k_j)\bm v_j^{\!\top}
      }_{\text{global info}},
    \quad
    \underbrace{\bm\kappa
      =\textstyle\sum_{j\in\mathcal N_0}\varphi(\bm k_j)
      }_{\text{key normalization}}$};
\node[note,text width=62mm] at (73,-50.3)
  {$\varphi(x)=x+1$ for $x>0$ and $e^x$ otherwise};

%
\coordinate (B1) at (132.4,-19); \coordinate (B2) at (160.4,-19);
\coordinate (B3) at (160.4,-32); \coordinate (B4) at (132.4,-32);
\draw[cyclc,-{Stealth[length=2mm,width=1.7mm]}]
      (148.4,-23) -- (155.4,-23) -- (155.4,-28) -- (137.4,-28)
      -- (137.4,-23) -- (144.4,-23);
\draw[wire] (B1) -- (142.9,-19);  \draw[thin ar] (149.9,-19) -- (B2);
\draw[wire] (B2) -- (160.4,-23);  \draw[thin ar] (160.4,-28) -- (B3);
\draw[wire] (B4) -- (142.9,-32);  \draw[thin ar] (149.9,-32) -- (B3);
\draw[wire] (B4) -- (132.4,-28);  \draw[thin ar] (132.4,-23) -- (B1);
\foreach \p in {(146.4,-19),(160.4,-25.5),(146.4,-32),(132.4,-25.5)}{\node[brn] at \p {};}
\foreach \p in {(B1),(B2),(B3),(B4)}{\node[busn] at \p {};}
\node[anchor=east] at (129.8,-19) {bus $i$};
\node[sgn,text=cphys] at (138.4,-16.6) {$+1$};
\node[sgn,text=closs] at (128,-23.6) {$-1$};
\node[anchor=center] at (141.75,-40)
  {$\operatorname{MP}_i(\bm z) =
    \mathop{\mathrm{mean}}\limits_{j\in\mathcal N(i)}
      \bigl(s_{ij}\,\bm W\bm z_j\bigr)
    + \bm W'\bm z_i + \bm b$};
\node[note,text width=63mm] at (141.75,-45.4)
  {$\mathcal N(i)$ are the neighbours of $i$ and $s_{ij}=\pm1$ the signed
   incidence of $j$ on $i$};

\draw[zoom] (20,-54.5) -- (20,-57.4);
\draw[zoom] (73,-54.5) -- (86,-57.4);
\draw[zoom] (141.75,-54.5) -- (117,-57.4);
\node[pre,minimum width=30mm,minimum height=10mm,text width=27mm]
  (s1) at (18,-63) {features $+$ diffusion\\ positional encoding};
\node[pre,minimum width=26mm,minimum height=10mm,text width=23mm]
  (s2) at (50,-63) {linear embedding $\bm z$};
\foreach \d in {1.5,0.75}{
  \draw[rounded corners=1.1mm,draw=cphys!45,line width=0.3mm,fill=cphys!8]
        (67+\d,-57.5-\d) rectangle (150+\d,-68.5-\d);}
\draw[rounded corners=1.1mm,draw=cphys,line width=0.36mm,fill=cphys!4]
      (67,-57.5) rectangle (150,-68.5);
\node[sub] at (68.4,-59.1) {GridBlock $\times\,8$};
\node[phys,fill=white,minimum width=21mm,minimum height=4.8mm,text width=19mm]
  at (79.4,-65.4) {self-attention};
\node[phys,fill=white,minimum width=32mm,minimum height=4.8mm,text width=30mm]
  at (109.5,-65.4) {signed message passing};
\node[phys,fill=white,minimum width=19mm,minimum height=4.8mm,text width=17mm]
  at (138.6,-65.4) {feed-forward};
\draw[ar] (90.1,-65.4) -- (93.3,-65.4);
\draw[ar] (125.7,-65.4) -- (128.9,-65.4);
\node[lss,minimum width=20mm,minimum height=10mm,text width=17mm]
  (s4) at (164,-63) {\textbf{controls} $\widehat{\bm u}$};
\draw[ar] (33.2,-63) -- (36.8,-63);
\draw[ar] (63.2,-63) -- (66.4,-63);
\draw[ar] (151.7,-63) -- (153.8,-63);

\end{scope}
\end{tikzpicture}
\caption{GridSFM architecture. The grid is first encoded into a graph structure to a common hidden dimension where the model is then passed through $8$ layers consisting of both a global linear attention and signed message passing component resulting in the fine dispatch control estimates for the AC-OPF solve.}
\label{fig:arch-diagram}
\end{figure*}

In this section, we describe a framework that aligns with how a foundation model could be used by an operator, combining the computational speed of neural networks with the power-flow equations to produce physically consistent solutions. It is composed of three stages (see figure~\ref{fig:flow_diag}):

\begin{enumerate}
\item[(i)] \textbf{A topology-agnostic pretrained model} that maps the grid topology, line parameters, and loads to the generator active-power outputs and voltage magnitude. We call these the control variables and denote them by $\bm{u}$. 
\item[(ii)] \textbf{A physically consistent fine-tuning framework} that steers the pretrained backbone toward a single grid of interest via a physics-informed
loss, and remains effective even when the predicted setpoints are far from their targets on tasks outside the pretraining distribution.
\item[(iii)] \textbf{An output interface} that connects the approximate solutions to downstream tasks, including warm-starting a conventional solver
and projection via \eqref{eq:AC-OPF-projection}.
\end{enumerate}

This section provides intuition and high level structure of the architecture design, with the fine details omitted due to length constraints. They can be found in Appendix \ref{appendix:architecture}.

\subsection{Graph Encoding}
\label{sec:graph-encoding}

We encode the input grid as a heterogeneous graph whose nodes
take one of seven types, $\Sigma := \{\mathrm{bus}, \mathrm{gen}, \mathrm{load}, \mathrm{shunt}, \mathrm{line}, \mathrm{transformer}, \mathrm{cycle}\}.$ The first six are self explanatory, and the last type, cycle, accounts for the fact that Kirchhoff's voltage law imposes a constraint on the power flow in cycles~\cite{6756976,zhang2014network}. The edges of the graph record incidence. For example, generators, loads, and shunts attach to their bus by an unsigned edge. Each line and transformer attaches to its two endpoint buses by a signed edge, $+1$ at the from-bus and $-1$ at the to-bus, and to every cycle containing it by a second signed edge giving the direction in which that cycle traverses it. Throughout, we write $\sigma(i) \in \Sigma$ for the type of node $i$.



When node $i$'s type is a bus, branch (AC lines or transformers), or a cycle, it is associated with a positional encoding that captures its position relative to the other nodes. It is computed through a diffusion operation on a Laplacian defined from the topology of the input grid. The construction essentially allows a node's encoding to summarize the information in its neighborhood, and structurally different parts of the network are separated by their node encodings. The rates of diffusion are learned in the training process, and let each channel decide the contribution of their adjacent neighbors allowing more important network nodes to dominate. This style of encoding is not new \cite{dwivedi2022graph}, but to our knowledge, has not been applied in the context of power systems.

\subsection{Neural Network Design}
\label{sec:nn-design}

After encoding, GridSFM consists of eight layers combining linear self-attention, signed
message passing, and classic multi-layer perceptron networks (see
Figure~\ref{fig:arch-diagram}). Each block applies the three in sequence. We begin by detailing the normalized linear self-attention design \cite{katharopoulos20a}. The layer acts within a node type: node $i$ attends to nodes of its own type $\sigma(i)$ and to no others, with parameters shared by all nodes of that type. Explicitly, for each type, using $\tilde{\bm{z}}_i$ as the previous input into layer, we define the embeddings
\begin{align}
    \bm{q}_i^{\rm emb} = \bm{W}_{\sigma(i)}^{q}\tilde{\bm{z}}_i, \,
    \bm{k}_i^{\rm emb} = \bm{W}_{\sigma(i)}^{k}\tilde{\bm{z}}_i, \,
    \bm{v}_i^{\rm emb} = \bm{W}_{\sigma(i)}^{v}\tilde{\bm{z}}_i,
\end{align}
which yield the linearized attention output
\begin{align}
    \bm{z}_i^{\rm attn}
    &= \bm{W}_{\sigma(i)}^{\rm o}\frac{
        \left(\sum_{j \in \mathcal{V}_{\sigma(i)}}
        \bm{v}_j^{\rm emb} \varphi(\bm{k}_j^{\rm emb})^\top\right)
        \varphi(\bm{q}_i^{\rm emb})
    }{
        \varphi(\bm{q}_i^{\rm emb})^\top
        \sum_{j \in \mathcal{V}_{\sigma(i)}}
        \varphi(\bm{k}_j^{\rm emb})
    }, \label{eq:attn}
\end{align}
where $\mathcal{V}_\tau$ is the set of nodes of type $\tau$, and
$\varphi(x)$ is the Exponential linear unit function (see Figure \ref{fig:arch-diagram}) acting element-wise. In \eqref{eq:attn}, $\bm{z}_i^{\rm attn}$ is a single attention head and in practice each attention layer has four of these outputs that are concatenated (See Appendix \ref{appendix:architecture}). The matrices $\bm{W}_{\sigma(i)}^q, \bm{W}_{\sigma(i)}^k, \bm{W}_{\sigma(i)}^v, \bm{W}_{\sigma(i)}^o$ are learned independently per node-type and in the multi-head case, once per head. 
This layer aggregates global information across each node type by computing a shared global state from the local hidden states, creating long range dependencies across the network. 

We next turn to signed message passing. 
Let $\mathcal{N}(i)$ denote the nodes
adjacent to $i$, and let $s_{ij}$ be the sign of the edge between $i$ and $j$,
equal to $\pm 1$ on the bus--branch and branch--loop edges and to $+1$ on the unsigned edges attaching generators, loads, and shunts. Each node aggregates its neighbors of each type separately,
\begin{equation}
\label{eq:mp-update}
\begin{split}
    \bm{z}^{\mathrm{mp}}_i = \bm{z}^{\mathrm{attn}}_i + \sum_{\tau} \Bigl[
      &\underbrace{\operatorname*{mean}_{j \in \mathcal{N}(i) \cap \mathcal{V}_\tau}
        \bigl(s_{ij}\, \bm{W}_{\tau,\,  \sigma(i)}\,\tilde{\bm{z}}_j\bigr)}
        _{\text{neighbor contributions}} \\
      &+ \underbrace{\bm{W}'_{\tau,\,  \sigma(i)}\tilde{\bm{z}}_i}
        _{\text{node contribution}}
      + \bm{b}_{\tau,\,  \sigma(i)} \Bigr] \,,
\end{split}
\end{equation}
where the sum runs over the node types $\tau$ adjacent to $\sigma(i)$.

Finally, each block closes with a position-wise feed-forward network, $\operatorname{FFN}(\tilde{\bm{h}}_i) = \bm{W}'_{\sigma(i)}
\operatorname{GELU}\bigl(\bm{W}_{\sigma(i)} \tilde{\bm{z}}_i+ \bm{b}_{\sigma(i)}\bigr) + \bm{b}'_{\sigma(i)}$, where each node type carries its own pair in each block, shared by every
node of that type.

\subsection{GridSFM Output Heads}
\label{sec:output}

To complete GridSFM, the output heads are composed of two phases: a fusion that
gathers each node's neighborhood and the grid as a whole, and a per-quantity head. This part of the architecture is standard (see Appendix \ref{appendix:architecture}). An output layer aggregates all the hidden state information, and then output $\tilde{\mathbf {u}} = (\tilde{\mathbf{p}}^g, \tilde{\mathbf{v}}_{\mathcal{N}_v})$. These are projected on to there feasible ranges by thresholding against the voltage (and generator power) upper and lower bounds, giving the final output $\hat{\bd{u}}$. 

\subsection{GridSFM Pretraining Loss}
GridSFM is pretrained by supervised regression using solved AC-OPF instances. The loss function is
\begin{equation}
\label{eq:pretrain-loss}
    \Phi_{\rm pt} = w_p \mathop{\mathcal{M}}_{i \in \mathcal{G}}
      \left(\biggl| \frac{\widehat p^{\mathrm g}_i - p^{\mathrm g}_i}
                   {\varepsilon_p} \biggr|\right)
    + w_v\mathop{\mathcal{M}}_{i \in \mathcal{N}_v}
      \left(\biggl| \frac{\widehat v_i - v_i}{\varepsilon_v} \biggr|\right) \,, 
\end{equation}
where  $\mathop{\mathcal{M}}$ is a function on a set of errors $\{e_i\}_{i \in \mathcal{S}}$, defined as
\begin{align}
    \mathop{\mathcal{M}}_{i \in \mathcal{S}} e_i := w_{\mu} \operatorname*{mean}_{i \in \mathcal{S}} e_i
      + w_{\tau} \operatorname*{mean}_{i \in \text{topk}(\mathcal{S})} e_i
      + w_{\infty} \max_{i \in \mathcal{S}} e_i \,.
\end{align}
where $\text{topk}(\mathcal{S})$ collects the top $k$ errors in $\mathcal{S}$. The tolerances $\varepsilon_p$, $\varepsilon_v$ and weights $w_p$,$w_v$,$w_\mu$, $w_\tau$, $w_\infty$ are set by the user (see Table \ref{tab:hyperparams-finetuning} in Appendix \ref{appendix:hyperparams}).

\section{Physics-informed Fine-tuning Design}\label{sec:finetuning-full}

\begin{algorithm}[t]
\caption{One fine-tuning step on a batch of scenarios}
\label{alg:finetune-step}
\begin{algorithmic}[1]
\REQUIRE GridSFM pretrained backbone, batched scenarios $\mathcal{B}$
\STATE $\widehat{\bm{u}} \leftarrow
       (\widehat{\mathbf p}^{\mathrm g}, \widehat{\mathbf v}_{\mathcal N_v})
       = \mathrm{GridSFM}(\mathcal B)$
\FORALL{scenarios $s \in \mathcal B$}
  \STATE $\bm{y}^{(0)} \leftarrow$ predicted
         magnitudes with DC angles, $k= 0$
  \WHILE{$k < 40$ \textbf{and}
         $\|\bm{r}(\bm{y}; \widehat{\bm{u}})\|_{\infty} \geq 10^{-5}$}
    \STATE update $\bm{y}$ using \eqref{eq:newton} and $k \leftarrow k + 1$
  \ENDWHILE
  \STATE mark as converged if
         $\|\bm{r}(\bm{y}; \widehat{\bm{u}})\|_{\infty} < 10^{-5}$
  \STATE retain $\bm{J}$ and
         $\bm{r}(\bm{y}; \widehat{\bm{u}})$ at the final iterate
  \STATE recover the injections and the branch flows
\ENDFOR
\STATE evaluate $\Phi_{\mathrm{ft}}$, given in \eqref{eq:finetune-loss}, at the restored
       point
\STATE $\partial \Phi_{\mathrm{ft}} / \partial \widehat{\bm{u}},\
       \bm{\gamma} \leftarrow$ autograd through the loss and the model
\FORALL{scenarios $s \in \mathcal B$}
  \STATE \textbf{if} converged \textbf{then} solve
         $\bm{J}^{\!\top} \bm{\mu} = \bm{\gamma}$ and set \\
         $\nabla_{\widehat{\bm{u}}} \Phi^{(s)}_{\mathrm{ft}} \leftarrow
          \partial \Phi_{\mathrm{ft}} / \partial \widehat{\bm{u}}
          - (\partial \bm{r} / \partial \widehat{\bm{u}})^{\!\top} \bm{\mu}$
  \STATE \textbf{else} $\nabla_{\widehat{\bm{u}}} \Phi^{(s)}_{\mathrm{ft}}
         \leftarrow \partial \Phi_{\mathrm{ft}} / \partial \widehat{\bm{u}}$
\ENDFOR
\STATE update GridSFM by an Adam step
\end{algorithmic}
\end{algorithm}

Our pretraining step minimizes regression loss with respect to the generator setpoints and does not directly provide $\bd v$ and $\bm \theta$ solutions. This section describes how we close this gap for specific grids of interest by solving power flow  and fine-tuning based on the physically compatible solution. The process is summarized in Algorithm \ref{alg:finetune-step}.
 \subsection{Power-flow design}\label{subsec:powerflow}

Fix an input and let $\widehat{\bm{u}} = (\widehat{\mathbf p}^{\mathrm g},
\widehat{\mathbf v}_{\mathcal N_v})$ be output of the neural network. Write $\bm y =(\{\theta_i\}_{i \in \mathcal{N}}, \{v_i\}_{i \in \mathcal{N}_0\backslash \mathcal{N}_v})$ as the power flow solutions given $\widehat{\bm{u}}$ where
$\mathcal N_{\mathrm v}\subseteq\mathcal N_0$ denote the voltage-controlled
buses: namely the reference bus together with every bus hosting at least one
in-service generator. Define $\bm r$ to be the vector of active and reactive residuals: 
\begin{align}
\label{eq:pf-residual}
    \bm{r}(\bm{y}; \widehat{\bm{u}}) := (
      \widehat{\bm{p}}^{\mathrm g}- \bm{p}^{\mathrm d}
        - \bm{p}(\bm v, \bm\theta) ,
      - \bm{q}^{\mathrm d}
        - \bm{q}(\bm v, \bm\theta)).
\end{align}
We solve $\bm{r} = \bm 0$ by Newton's method  with a step-halving
line search~\cite{4073219},
\begin{equation}
\label{eq:newton}
    \bm{y}^{(k+1)} = \bm{y}^{(k)} - \beta_k\,
      \bm{J}\bigl(\bm{y}^{(k)}\bigr)^{-1}
      \bm{r}\bigl(\bm{y}^{(k)}; \widehat{\bm{u}}\bigr), \,
    \bm{J} :=\frac{\partial \bm{r}}{\partial \bm{y}},
\end{equation}
where $\beta_k$ is the largest step in $\{1, \tfrac12, \tfrac14, \dots\}$ that
decreases $\lVert \bm{r} \rVert_\infty$. We cap the
iteration at $40$ steps and declare convergence at $\lVert \bm{r} \rVert_\infty \le 10^{-5}$ p.u. We write  $\widehat{\bm v}$ and $\widehat{\bm \theta}$ as the power flow solutions. With the other other quantities that are computed from them,  we call this the restored operating point. 


\subsection{Fine-tuning loss}\label{subsec:finetuning-loss}
The controls $\widehat{\bm u}$ satisfy the generation and voltage
bounds, but the restored point computed from the power flow solutions may not. In addition, the restored branch flows may violate the capacity constraints. Recall that $\bm h$ is used collect all the inequality constraints~\eqref{eq:compact-c}, and let $\widehat{\bm h}$ denote the violations in the AC-OPF inequalities. The fine-tuning objective is then
\begin{align}
\label{eq:finetune-loss}
    \Phi_{\mathrm{ft}} =&\, w_c \sum_{i \in \mathcal G}
      c_i\bigl(\widehat p^{\mathrm g}_i\bigr)
      && \text{generation cost} \nonumber \\
    &+ w_g \sum_{m=1}^{n_r} \log\bigl(1 + |r_m|\bigr)
      && \text{equality violation} \nonumber \\
    &+ w_h \sum_{m=1}^{n_h} \log\bigl(1 + \widehat h_m\bigr)
      && \text{inequality violation} \nonumber \\
    &+ w_r \lVert \bm{r} \rVert_2^2
      && \text{power flow failure} \nonumber \\
    &+ w_u\, \Phi_{\mathrm{pt}}
      && \text{control loss}
\end{align}
where we minimize $\Phi_{\mathrm{ft}}$ by backpropagating through the layers of the neural network.
Next, we explain how to differentiate through the power flow solver. 


Within one sample, the
chain rule separates into two routes by which $\Phi_{\mathrm{ft}}$ depends on the
controls,
\begin{equation}
\label{eq:grad-routes}
    \nabla_{\widehat{\bm{u}}} \Phi_{\mathrm{ft}}
    = \underbrace{\partial \Phi_{\mathrm{ft}} /
      \partial \widehat{\bm{u}}}_{\text{explicit}}
    + \underbrace{\bigl( \partial \bm{y} / \partial \widehat{\bm{u}}
      \bigr)^{\!\top} \bm{\gamma}}_{\text{restoration}} \,, \quad
    \bm{\gamma} := \partial \Phi_{\mathrm{ft}} / \partial \bm{y} \,,
\end{equation}
where the first term can be explicitly computed. 
When the power flow solve converges, the second term can be computed via  the implicit
function theorem, since  $\partial \bm{y} / \partial \widehat{\bm{u}} =
-\bm{J}^{-1} \partial \bm{r} / \partial \widehat{\bm{u}}$, so the second term
equals $-(\partial \bm{r} / \partial \widehat{\bm{u}})^{\!\top} \bm{\mu}$ with
$\bm{J}^{\!\top} \bm{\mu} = \bm{\gamma}$ and thus, can be solved quickly since this reuses
the factorization already formed for the final Newton step. 

When the power flow equations fail to converge or the Jacobian is singular, the situation becomes more interesting. Existing AC-OPF surrogate approaches either assume this would not happen~\cite{donti2021dc} or declare failure~\cite{Par26}. But this scenario cannot be ignored for a foundation model, where all the samples in a batch could be non-convergent before fine-tuning is done. We overcome this by first dropping the $\partial \bm{y} / \partial \widehat{\bm{u}}$ term in \eqref{eq:grad-routes}. We then rely on the fourth term of \eqref{eq:finetune-loss}, $w_r \lVert \bm{r} \rVert_2^2$, to provide a descent direction. This is the classical least-squares measure of power-flow insolvability \cite{336130,373951} and is computed from its own physics rather than from proximity to a label.

Differentiating it through $\bm{y}(\widehat{\bm{u}})$ gives
\begin{equation}
\label{eq:merit-grad}
    \nabla_{\widehat{\bm{u}}} \lVert \bm{r} \rVert_2^2
    = \underbrace{
      \Bigl( \frac{\partial \bm{r}}{\partial \widehat{\bm{u}}} \Bigr)^{\!\top}
      \bm{r}}_{\text{retained}}
    + \underbrace{\Bigl( \frac{\partial \bm{y}}{\partial \widehat{\bm{u}}}
      \Bigr)^{\!\top} \bm{J}^{\!\top} \bm{r}}_{\text{neglected}},
\end{equation}
 whose neglected term carries the factor $\bm{J}^{\top}\bm{r} = \nabla_{\bm{y}} \lVert \bm{r} \rVert_2^2$, the gradient of the residual function with respect to the state, and so this vanishes wherever the closure has come to rest at a stationary point of the residual norm. A stalled solve leaves us at such a point: were it not so, the Newton direction
$\bm{d} = -\bm{J}^{-1} \bm{r}$ would satisfy
$\langle \bm{J}^{\!\top} \bm{r}, \bm{d} \rangle =
-\lVert \bm{r} \rVert_2^2 < 0$, and a sufficiently short step along it
would reduce the residual, so the line search would not have exhausted. The
retained term is therefore exact at a stationary point. 

\subsection{Downstream tasks: Feasible Projection and Warm Start}
\label{sec:downstream}
After fine-tuning, GridSFM can be used to find a solution that is nearly optimal and nearly feasible for an input load. However, there is no guarantee it is \emph{exactly} optimal or \emph{exactly} feasible. There are two natural ways to overcome this and both align with how operator may deploy surrogate models in practice. They are: (1) using it as a warm start for a conventional solve or (2) using it to find a close feasible point. 

Both of these methods are evaluated in detail in Section~\ref{sec:numerical}. The warm start method can be integrated seamlessly into existing AC-OPF algorithms, and often leads to faster solver speeds compared to flat start or DC-OPF based initialization. The projection method is even faster and works when there are no good warm starting points for the solver.

 Beyond these, we envision GridSFM will be powerful for downstream tasks that enumerate topologies such as in $N-1$ contingency screening and optimal transmission switching. Moreover, it can be used for AC feasibility checks for a candidate 
commitment, where the outer loop enumerates various loads or operating conditions. 

\section{Well-posedness of GridSFM training}\label{sec:theoretical}

A reader may reasonably ask why Section~\ref{sec:formulation} reformulates AC-OPF at all, whether relaxing changes the solution, and, what is the projection of an infeasible point onto a nonconvex feasible set.   
This section answers the three in turn. We first show that the feasible set of AC-OPF, when disconnected,
obstructs \emph{any} continuous approximator, whereas the elastic feasible set of \eqref{eq:AC-OPF-elastic} is simply connected (contractible). We then exhibit a finite threshold on the inequality parameters above which the elastic and the AC-OPF minimizers coincide. Finally, we show that the projection \eqref{eq:AC-OPF-projection} is uniquely defined and Lipschitz within an explicit radius of a
feasible point.

We begin by first isolating the only non-convex inequality constraint of \eqref{eq:AC-OPF}, the line flows. The reason for this will be clear in Assumption \ref{assumption:main} and Theorem \ref{thm:proj}. 
Hence, we introduce the auxiliary decision variables
$\tilde p_\ell$ and $\tilde q_\ell$ for the active and reactive power flow on each
branch $\ell\in\mathcal L$, stacked as
$\tilde{\mathbf p}_{\mathcal L}:=\operatorname{col}\{\tilde p_\ell\}_{\ell\in\mathcal L}\in\mathbb R^{n_\ell}$
and
$\tilde{\mathbf q}_{\mathcal L}:=\operatorname{col}\{\tilde q_\ell\}_{\ell\in\mathcal L}\in\mathbb R^{n_\ell}$,
and conveniently redefine the shorthand solution vector as
\begin{equation}
\mathbf x:=\operatorname{col}
\left(\mathbf v,\boldsymbol\theta,\mathbf p^{\mathrm g},\mathbf q^{\mathrm g},\tilde{\mathbf p}_{\mathcal L},\tilde{\mathbf q}_{\mathcal L}\right)
\in\mathbb R^{2n_b+1+2n_g+2n_\ell}.
\end{equation}
Then, the exact lifted problem becomes
\begin{equation}\label{eq:AC-OPF-lifted}
\min_{\mathbf x}\quad
 \sum_{i\in\mathcal G}c_i(p_i^{\mathrm g}), 
\text{ s.t. } 
\tilde{\mathbf g}(\mathbf x)=\mathbf 0, \tilde{\mathbf h}(\mathbf x)\leq\mathbf 0,
\end{equation}
where the equality constraints $\tilde{\mathbf g}$ formed by the original constraints in $\mathbf{g}(\mathbf{x})$ and line flow constraints $\{\tilde p_\ell-p_\ell(\mathbf v,\boldsymbol\theta)\}_{\ell\in\mathcal L}, \{\tilde q_\ell-q_\ell(\mathbf v,\boldsymbol\theta)\}_{\ell\in\mathcal L}$. 
The inequality constraints $\tilde{\mathbf h}(\mathbf x)$ collect the remaining constraints. 
Its elastic relaxation becomes
\begin{align}
& \min_{\mathbf x,\mathbf s^{\tilde g},\mathbf s^{\tilde h}}\quad
\sum_{i\in\mathcal G}c_i(p_i^{\mathrm g})
+\rho_{\tilde g}\mathbf 1^\top\log(\mathbf 1+\mathbf s^{\tilde g})
+\rho_{\tilde h}\mathbf 1^\top\log(\mathbf 1+\mathbf s^{\tilde h}) \nonumber
\\
& \text{s.t. }
-\mathbf s^{\tilde g}\leq\tilde{\mathbf g}(\mathbf x)\leq\mathbf s^{\tilde g}, \, \tilde{\mathbf h}(\mathbf x)\leq\mathbf s^{\tilde h}, \, \mathbf s^{\tilde g}\geq\mathbf 0,\, \mathbf s^{\tilde h}\geq\mathbf 0.                 \label{eq:AC-OPF-elastic-lifted}
\end{align}
where we introduce nonnegative elastic variables \(\mathbf s^{\tilde g}\) and \(\mathbf s^{\tilde h}\) for the equality and inequality residuals, respectively. For the remainder of this section and the proofs of Appendix~\ref{app:proofs}, the
projection \eqref{eq:AC-OPF-projection} is likewise understood over the lifted
variables: $\mathbf x\in\mathbb R^{N}$ with $\tilde{\mathbf g},\tilde{\mathbf h}$
in place of $\mathbf g,\mathbf h$, so that \eqref{eq:AC-OPF-projection} is the
Euclidean projection onto
$\mathcal F:=\{\mathbf x\in\mathbb R^{N}:\tilde{\mathbf g}(\mathbf x)=\mathbf 0,\
\tilde{\mathbf h}(\mathbf x)\leq\mathbf 0\}$. 

First, note that every component of $\tilde{\mathbf h}$ is convex and hence, all nonconvexity of AC-OPF is confined to $\tilde{\mathbf g}$. Next, note that the second derivatives are bounded branch-by-branch by admittance magnitudes and the voltage caps. This will allow us to bound the curvature inequalities in an easy manner.

The relaxation formulation avoids the possibility of disconnected feasible set as stated below
\begin{lemma}[The elastic feasible set is contractible]\label{lem:elastic-connected}
Define the feasible set of the lifted
elastic problem \eqref{eq:AC-OPF-elastic-lifted},
\begin{align}
  \Omega := \Bigl\{ (\mathbf{x},\mathbf{s}^{\tilde g},\mathbf{s}^{\tilde h})
     : \bigl|\tilde{\mathbf g}(\mathbf{x})\bigr| \leq \mathbf{s}^{\tilde g},\
     \bigl[\tilde{\mathbf h}(\mathbf{x})\bigr]_{+} \leq \mathbf{s}^{\tilde h} \Bigr\},
  \label{eq:Omega}
\end{align}
where $|\cdot|$, $[\,\cdot\,]_{+}:=\max\{\,\cdot\,,0\}$, and the inequalities apply componentwise. The feasible set $\Omega$ is set $\Omega$ is contractible, hence path-connected and simply connected.
\end{lemma}

We note that the contractibility of $\Omega$ does not imply that approximating the elastic solution map is easy. The elastic optimization remains nonconvex and its solution map need not be continuous. Instead, in the two results that follow, we provide checkable conditions for an operator to verify their foundation models performance such that the resulting slacked AC-OPF is both learnable, and any feasible projection is well-posed. To do so, we require the following assumption:

\begin{assumption}[Linear Inequality Constraint Qualification (LICQ) at the target point]\label{assumption:main}
Let $\mathcal{F}\subseteq\mathbb R^{N}$ be the feasible set of the lifted problem \eqref{eq:AC-OPF-lifted}
and let $\mathbf x^\ast\in\mathcal{F}$ be a local minimizer. 
For $\mathbf x\in\mathcal{F}$, let
$\mathcal{A}(\mathbf x):=\{\tilde h_k(\mathbf x)=0\}$ denote the active inequalities and
\begin{align}
    \mathbf J(\mathbf x):=\begin{bmatrix}
        \nabla \tilde{\mathbf{g}}(\mathbf x)^\top \\[2pt]
        \nabla \tilde h_k(\mathbf x)^\top, & k\in\mathcal{A}(\mathbf x)
    \end{bmatrix},
    \label{eq:active-jacobian}
\end{align}
the active-constraint Jacobian. We assume there exist a radius $R_0>0$ and a constant
$\underline{\sigma}>0$ such that
\begin{equation}\label{eq:licq}
    \sigma_{\min}\!\big(\mathbf J(\mathbf x)\big)\;\ge\;\underline{\sigma}
    \qquad\text{for all } \mathbf x\in\mathcal{F}\cap\overline{B}(\mathbf x^\ast,R_0),
\end{equation}
and such that $R_0\ge\delta$, where
\begin{equation}\label{eq:Mdelta}
    \delta:=\frac{\underline{\sigma}}{M},
    \qquad
    M:=\kappa\left(\max_{i\in\mathcal N_0}\ \sum_{\ell=(i,k)\in\mathcal{L}_i}|Y_{ik}|\right),
\end{equation}
with $\kappa:=10\,(1+\overline{v})$, $\overline{v}:=\max_{i\in\mathcal{N}_0}\overline{v}_i$ the largest
upper voltage limit in \eqref{eq:AC-OPF-f}, and $Y_{ik}=G_{ik}+\mathrm{j}B_{ik}$ the off-diagonal entry
of the bus admittance matrix appearing in \eqref{eq:branch-flows}.
\end{assumption}

This assumption requires that the power flow and active inequality constraint Jacobian is non-singular in a small radius around any minimizer. This is the standard LICQ assumption which, for most practical grids, is satisfied (see \cite{8391733}). Second, we assume this radius is larger than a constant $\delta$, which is always satisfiable given LICQ as $\underline{\sigma}$ can always be lowered.

In Theorem \ref{thm:proj}, we will see that $\delta/2$ will be the reachable radius such that the projection is well-posed. It is distinctly made of two components, $\underline{\sigma}$ which controls the feasibility distance of the original AC-OPF problem, and $M$ depending on network parameters. We first show that, for the radius $R_0$, any minimizer of the original AC-OPF problem is a minimizer of the elastic problem given large enough penalty on the slack:

\begin{lemma}[Local exactness at $\mathbf x^\ast$]\label{lem:exact}
Let Assumption~\ref{assumption:main} hold, and write
$c(\mathbf x):=\sum_{i\in\mathcal G}c_i(p_i^{\mathrm g})$ for the cost of
\eqref{eq:AC-OPF-lifted}. Additionally, assume that $c$ is continuously
differentiable on $\overline{B}(\mathbf x^\ast,R_0)$ and let $L_c>0$ be a
constant with $\lVert\nabla c(\mathbf x)\rVert\le L_c$ for all
$\mathbf x\in\overline{B}(\mathbf x^\ast,R_0)$. Set $\rho_{\min}:=\min\{\rho_{\tilde g},\rho_{\tilde h}\}$. If
\begin{equation}\label{eq:rhocond}
    \rho_{\min}\;>\;L_c/\underline{\sigma},
\end{equation}
then $(\mathbf x^\ast,\mathbf 0,\mathbf 0)$ is a local minimizer of the
elastic problem \eqref{eq:AC-OPF-elastic-lifted}.
\end{lemma}

\begin{proofsketch}
The elastic problem reduces to minimizing
$f_\rho(\mathbf x)=c(\mathbf x)+\rho_{\tilde g}\mathbf 1^\top\log(\mathbf 1+|\tilde{\mathbf g}(\mathbf x)|)
+\rho_{\tilde h}\mathbf 1^\top\log(\mathbf 1+[\tilde{\mathbf h}(\mathbf x)]_+)$, and for $\mathbf x$ near $\mathbf x^\ast$. We show that any improvement in cost must at least incur an equivalent or greater  feasibility penalty. For the cost, using the Lipschitz bound on $c$ gives $c(\mathbf x)\ge c(\mathbf x^\ast)-L_c d$ where $d := \mathrm{dist}(\mathbf{x}, \mathcal{F})$. Analogously, it can be shown using the KKT conditions onto $\mathcal{F}$ with the LICQ assumption and $\rho_{\min}$ in \eqref{eq:rhocond} that the infeasibility penalty is at least $\rho_{\min}(\underline\sigma-O(d))\,d\ge L_c d$ for small $d$. Hence, $(\mathbf{x}^\ast, \bm 0, \bm 0)$ is a minimizer. 
The full proof is given in Appendix \ref{appendix:lemma2proof}.
\end{proofsketch} 

Lemma~\ref{lem:exact} says that once $\rho_{\min}$
clears $L_c/\underline{\sigma}$, the slacks vanish at $x^\ast$ and training against the
elastic objective targets exactly results in the same minimizers of the lifted AC-OPF. However, $\underline{\sigma}$ is not necessarily known by operators prior since, $x^\ast$ is not known. Hence, computing $\rho_{\rm min}$ is not necessarily practical.

Therefore, the resulting output from GridSFM should be expected to have some slack, whose valuable is directly checkable. We are interested in the result of the projection. Theorem~\ref{thm:proj} shows that this projection map is well posed
in a neighborhood of $x^\ast$, and that the neighborhood is governed by the constraint Jacobian. Explicitly, a larger $\underline{\sigma}$ yields a wider  radius for $\delta=\underline{\sigma}/M$ on which the projection can be trusted.

\begin{figure}[ht]
    \centering
    \includegraphics{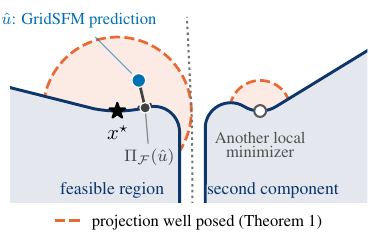}
    \caption{Illustrates geometrically the idea of the projection of Theorem \ref{thm:proj}}
    \label{fig:reach}
\end{figure}

\begin{theorem}[Local projection radius]\label{thm:proj}
Let Assumption~1 hold and let $\mathcal{B}_{\mathbf x^\ast}$ denote the connected component of $\mathcal{F}$
containing $\mathbf x^\ast$. Then for every $\hat{\mathbf x}\in B(\mathbf x^\ast,\delta/2)$ the projection \eqref{eq:AC-OPF-projection} has exactly
one stationary point $\boldsymbol\xi\in\mathcal{F}\cap\overline{B}(\mathbf x^\ast,R_0)$ with $\|\hat{\mathbf x}-\boldsymbol\xi\|<\delta$. Denoting
it $\Pi^{\mathrm{loc}}_{\mathcal{F}}(\hat{\mathbf x})$, we have
\begin{enumerate}[leftmargin=*]
\item[(i)] \emph{(uniqueness)} $\Pi^{\mathrm{loc}}_{\mathcal{F}}(\hat{\mathbf x})$ coincides with the global projection
$\Pi_{\mathcal{F}}(\hat{\mathbf x})$
\item[(ii)] \emph{(correct component)} $\Pi^{\mathrm{loc}}_{\mathcal{F}}(\hat{\mathbf x})\in\mathcal{B}_{\mathbf x^\ast}$
\item[(iii)] \emph{(Lipschitz)} for $\hat{\mathbf x}_1,\hat{\mathbf x}_2\in B(\mathbf x^\ast,\delta/2)$ with
$s_0:=\max_i\mathrm{dist}(\hat{\mathbf x}_i,\mathcal{F})<\delta/2$, we have $\|\Pi^{\mathrm{loc}}_{\mathcal{F}}(\hat{\mathbf x}_1)-\Pi^{\mathrm{loc}}_{\mathcal{F}}(\hat{\mathbf x}_2)\| \leq  \frac{1}{1-s_0/\delta}\,\|\hat{\mathbf x}_1-\hat{\mathbf x}_2\| < 2\,\|\hat{\mathbf x}_1-\hat{\mathbf x}_2\|.$
\end{enumerate}
\end{theorem}

Theorem \ref{thm:proj} provides a concrete radius such that the projection is well-posed in the sense that it is unique, Lipschitz, and on the same connected component as the nearby local minimizer of interest. 
In Figure \ref{fig:reach}, we see that for two disconnected components, the reach radius of $\delta/2$ exactly coincides with the radius for which the projection will align with the correct component. The value of $\delta/2$ is a direct function of the curvature of the feasible set which is controlled by by the constraint Jacobian $(\underline{\sigma})$ and the topology of the network $(M)$. Together, they provide a certificate for the user to check well-posedness. Lastly, from an practical perspective, we find the projection is always well-posed as it converges for all finetuned models in Section \ref{sec:numerical} when using GridSFM.

\begin{proofsketch}
We start with the following lemma
\begin{lemma}\label{lem:reach}
Let Assumption~\ref{assumption:main} hold and let $\boldsymbol\xi\in\mathcal F\cap\overline B(\mathbf x^\ast,R_0)$ be a
stationary point of the projection \eqref{eq:AC-OPF-projection} of $\hat{\mathbf x}$. Then for every
$\boldsymbol\zeta\in\mathcal F\cap\overline B(\mathbf x^\ast,R_0)$,
\begin{align}
    \|\hat{\mathbf x}-\boldsymbol\zeta\|^2\ \ge\ \|\hat{\mathbf x}-\boldsymbol\xi\|^2
    +\Big(1-\tfrac{\|\hat{\mathbf x}-\boldsymbol\xi\|}{\delta}\Big)\|\boldsymbol\zeta-\boldsymbol\xi\|^2 .
    \label{eq:reach-growth}
\end{align}
\end{lemma}
This lemma is proved in Appendix \ref{appendix:bowl-proof}, and says that moving along $\mathcal F$ away from $\boldsymbol\xi$ increases the distance to
$\hat{\mathbf x}$ at the quadratic rate $1-\|\hat{\mathbf x}-\boldsymbol\xi\|/\delta$.
Since $\hat{\mathbf x}-\boldsymbol\xi$ is normal to $\mathcal F$ at $\boldsymbol\xi$, a competitor
could only come closer if $\mathcal F$ bent back toward $\hat{\mathbf x}$ at second order, and
$\delta=\underline\sigma/M$ is the radius that caps this bending. The three claims follow.
\begin{enumerate}
\item[(i)] (unique, point-valued) Since $\mathcal F$ is closed, a global projection $\bar{\boldsymbol\xi}$ exists, and
$\|\hat{\mathbf x}-\bar{\boldsymbol\xi}\|\le\|\hat{\mathbf x}-\mathbf x^\ast\|<\delta/2$ places it in
$\overline B(\mathbf x^\ast,\delta)\subseteq\overline B(\mathbf x^\ast,R_0)$. For any other stationary point
$\boldsymbol\xi$ with $\|\hat{\mathbf x}-\boldsymbol\xi\|<\delta$, applying \eqref{eq:reach-growth} at $\boldsymbol\xi$ with
$\boldsymbol\zeta=\bar{\boldsymbol\xi}$ and using $\|\hat{\mathbf x}-\bar{\boldsymbol\xi}\|\le\|\hat{\mathbf x}-\boldsymbol\xi\|$ gives
$0\ge(1-\|\hat{\mathbf x}-\boldsymbol\xi\|/\delta)\|\bar{\boldsymbol\xi}-\boldsymbol\xi\|^2$, forcing $\boldsymbol\xi=\bar{\boldsymbol\xi}$.
\item[(ii)] (Lipschitz) For $\hat{\mathbf x}_1,\hat{\mathbf x}_2$ with projections $\boldsymbol\xi_1,\boldsymbol\xi_2$, applying the
cross-term bound underlying \eqref{eq:reach-growth} to each pair and adding yields
$(\hat{\mathbf x}_1-\hat{\mathbf x}_2)^\top(\boldsymbol\xi_1-\boldsymbol\xi_2)\ge(1-s_0/\delta)\|\boldsymbol\xi_1-\boldsymbol\xi_2\|^2$, and
Cauchy--Schwarz gives the Lipschitz constant $(1-s_0/\delta)^{-1}<2$.
\item[(iii)] (correct component) By (i) and (ii), $\Pi^{\rm loc}_{\mathcal F}$ is continuous on the connected ball
$B(\mathbf x^\ast,\delta/2)$, its image is a connected subset of $\mathcal F$ containing
$\mathbf x^\ast$ and hence lies in $\mathcal B_{\mathbf x^\ast}$.
\end{enumerate}
\end{proofsketch}

\vspace{-0.8cm}
\section{Numerical Experiments}\label{sec:numerical}
In this section, we conduct two primary experiments to highlight the performance of GridSFM. All models and code are released in \cite{bhan_gridsfm}. We evaluate GridSFM as a zero-shot pretrained model on various grids. Then we show that with the fine-tuning design of Section \ref{sec:finetuning-full}, its performance extrapolates to out-of-distribution grids with only $100$ new labels. 

\vspace{-0.3cm}
\subsection{Dataset Generation}\label{subsec:data-gen}
We begin by discussing the pretraining dataset design. It is built by considering base topologies from OPFData \cite{lovett2024opfdatalargescaledatasetsac}, PGLib \cite{pglib}, Texas A\&M Synthetic Grids \cite{7725528} and MSR synthetic topologies \cite{britto2026buildingpowergridmodels} and then applying perturbations. We consider load rescaling, permutations of generation costs, reductions in thermal line limits, reductions in voltage limits, and generator outages. 
We emphasize that, in particular, permuting generation costs makes the problem challenging as it forces the model to adapt across changing cost rankings instead of memorizing dispatch preferences. 
The full dataset design and perturbations are detailed in Appendix \ref{appendix:dataset-design}.

\begin{figure}[t]
    \centering
    \includegraphics{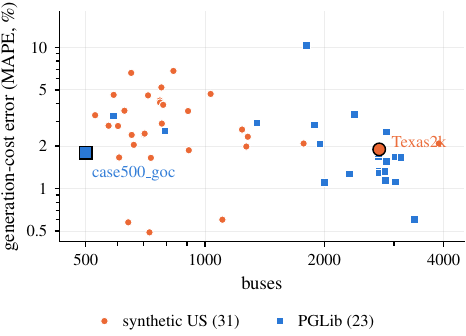}
    \caption{Performance of pretrained model (zero-shot) across $54$ different grids showing the generation cost percentage error vs the number of buses. 
    \vspace{-0.5cm}}
    \label{fig:pretrained}
\end{figure}

\vspace{-0.3cm}
\subsection{Performance as a model across topologies}
We first evaluate the pretrained backbone on $10{,}000$ held-out AC-OPF scenarios from the $54$ grid topologies used during pretraining. It took approximately one week to pretrain the model on a NVIDIA DGX B200 system with eight Blackwell GPUs. Figure~\ref{fig:pretrained} reports the resulting zero-shot cost error (the accuracy of control variables follow similar trends, omitted here due to space constraints). 
They remain relatively small (between $1\%$ and $10\%$) across most grids, and moreover, are approximately constant as the number of buses grows. This suggests that the pretrained model's predictive accuracy does not degrade with system size.

\begin{figure*}
    \centering
    \includegraphics[width=\linewidth]{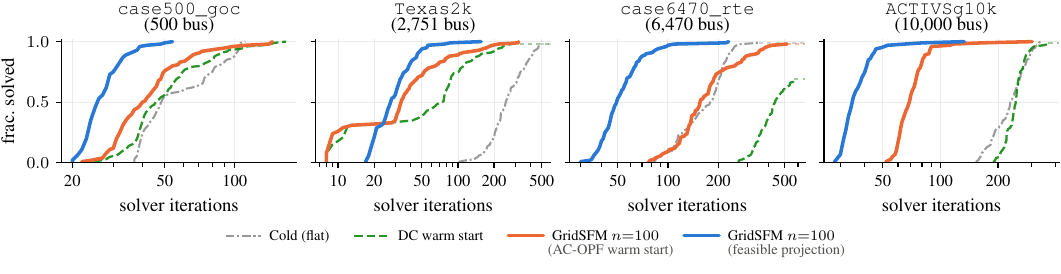}
    \caption{Performance comparison of Ipopt solver with various warm-start AC-OPF points on $100$ eval cases, showing the percentage of cases solved as a function of the number of iterations. Cold, DC, and GridSFM (orange line) represent warm-starts for solving the AC-OPF problem fully while the blue GridSFM line presents the iterations needed to obtain a feasible projection. \vspace{-0.5cm}}
    \label{fig:iter-cdf-cpu}
\end{figure*}

\vspace{-0.3cm}
\subsection{Fine-tuning performance}\label{subsec:finetune}

We illustrate the fine-tuning design of Section~\ref{sec:finetuning-full} on four
grids: \texttt{case500\_goc} and \texttt{Texas2k}, whose topologies are in
the pretraining set, and two significantly larger out-of-distribution grids,
\texttt{case6470\_rte} and \texttt{ACTIVSg10k}. For each grid we fine-tune on
$100$ solved scenarios generated with the perturbations of
Section~\ref{subsec:data-gen}, 
For comparison we consider four baselines: an MLP
mapping the scenario parameters directly to the controls in the style of
DeepOPF~\cite{pan20}, a GNN without signed message passing, the Linux
Foundation grid model~\cite{Pue26}, and GridSFM without pretraining. The MLP,
GNN, and unpretrained GridSFM are each trained on $1{,}000$ solved instances
per grid, while the Linux Foundation model is used directly through the
GridFM data kit (see Appendix~\ref{sec:hyperparams-baselines} for all
hyperparameters and training details).

We evaluate two downstream tasks: warm-starting AC-OPF in Ipopt
and the feasible projection \eqref{eq:AC-OPF-projection}. Throughout, we
report solver iterations rather than wall-clock time so that the results are
hardware independent (timings are given in
Table~\ref{tab:compute-cost} of Appendix~\ref{app:experiments}). As a
reference, one Ipopt iteration on \texttt{ACTIVSg10k} takes roughly $0.5$\,s
on our hardware.

Figure~\ref{fig:iter-cdf-cpu} compares the fine-tuned GridSFM against flat and
DC warm starts on $100$ unseen evaluation scenarios per grid. The top row
shows that the fine-tuned model outperforms both on every grid except
\texttt{case6470\_rte}. For that one, DC warmstarts significantly degrades solver
performance, while GridSFM matches the flat start. 
In all cases the projection converges in far fewer iterations than the full AC-OPF solve, with just 0.5\% optimality loss. In total, GridSFM reduces the iteration counts by $4$--$8\times$.

We emphasize that the speedup given may appear less then other literature (\cite{fio19,pan20,Pil24,Pue26}). These works measure the batched inference wall-clock times compared to solver wall-clock times. These are not directly comparable given that much of the saving depends on exact implementation and how infeasible output of the ML surrogate is counted. Because a solver needs some iterations to check the quality of the solution, even if we start with the actual minimizer, a median of 26 iterations are needed for convergence for the grids in Fig.~\ref{fig:iter-cdf-cpu}. Using GridSFM leads to $69$ iterations compared to $232.5$ for a cold start. Hence, GridSFM recovers $\approx 80$\% of the speedup as a perfect initial point. For a wall-clock perspective, please see Table \ref{tab:compute-cost} of Appendix \ref{sec:finetuning-appendix}.

Table~\ref{tab:warmstart} reports the results quantitatively. On
\texttt{Texas2k}, the fine-tuned warm start lowers the iterations needed to
solve $90\%$ of cases from $422$ to $103$ ($7.6\times$ median speedup), and
on \texttt{ACTIVSg10k} from $295$ to $87$ ($3.3\times$) at $100\%$
convergence, whereas the per-grid MLP and GNN converge on only $3$ of $100$. For the projection, 
using MLP on \texttt{case6470\_rte}
is $3.6\times$ faster but $32.9\%$ from the optimum, while GridSFM achieves
$3.2\times$ at a $2.0\%$ gap, and on \texttt{ACTIVSg10k} it projects every
case within $45$ iterations at a $0.57\%$ gap.  Appendix~\ref{app:experiments}
repeats the experiment with the GPU solver MadNLP~\cite{shin2024accelerating}
and finds the same qualitative behavior.

Moreover, we additionally provide results for the scaling fine-tuned performance with the number of data during the fine-tuning stage in Figure \ref{fig:activsg10k-scaling}. We see that both control performance and AC-OPF warm start iterations increases as more scenarios are used during training. 
Appendix \ref{appendix:ablation} has two ablation studies highlighting both the value of the power flow residual in the fine-tuning loss and the advantage of the elastic design over finetuning on just the controls.

\begin{table*}[!t]
\centering
\caption{\vspace{1em} \parbox{\textwidth}{\normalfont\footnotesize Performance of various warm-starting points on both AC-OPF \eqref{eq:AC-OPF} and the feasible projection \eqref{eq:AC-OPF-projection} (Ipopt) for $100$ evaluation problems across various grid sizes. $90\%$ done indicates the number of iterations to complete $90/100$ problems, the speed-up is the median speedup relative to a cold start computed on only convergent problems, convergence $\%$ indicates the number of problems that converged ($600$-iteration limit) and the cost gap indicates the projection's cost gap compared to the Ipopt optimum. The MLP, GNN and GridSFM without pretraining are retrained on $1$k solved instances per grid. 
The Linux Foundation model \cite{Pue26} is finetuned on the \texttt{case500\_goc} grid and should be considered zero-shot for the others. Best performing metric per column is indicated in bold and $\dagger$ indicates metrics are computed on $<50$ convergent cases.}}
\label{tab:warmstart}
\setlength{\tabcolsep}{2.4pt}
\renewcommand{\arraystretch}{1.15}
\footnotesize
\setlength{\arrayrulewidth}{0.6pt}
\arrayrulecolor{gray!60}
\begin{tabular}{@{}l;{1.5pt/1.2pt}cccc;{1.5pt/1.2pt}cccc;{1.5pt/1.2pt}cccc;{1.5pt/1.2pt}cccc@{}}
\arrayrulecolor{black}
\toprule
Warm start & \multicolumn{4}{c}{\texttt{case500\_goc}} & \multicolumn{4}{c}{\texttt{Texas2k}} & \multicolumn{4}{c}{\texttt{case6470\_rte}} & \multicolumn{4}{c}{\texttt{ACTIVSg10k}} \\
\cmidrule(lr){2-5}\cmidrule(lr){6-9}\cmidrule(lr){10-13}\cmidrule(lr){14-17}
 & \makecell{P90\\(iters)} & \makecell{Speed-\\up} & \makecell{Conv.\\(\%)} & \makecell{Cost gap\\(\%)} & \makecell{P90\\(iters)} & \makecell{Speed-\\up} & \makecell{Conv.\\(\%)} & \makecell{Cost gap\\(\%)} & \makecell{P90\\(iters)} & \makecell{Speed-\\up} & \makecell{Conv.\\(\%)} & \makecell{Cost gap\\(\%)} & \makecell{P90\\(iters)} & \makecell{Speed-\\up} & \makecell{Conv.\\(\%)} & \makecell{Cost gap\\(\%)} \\
\midrule
\multicolumn{17}{@{}l}{\itshape AC-OPF warm start}\\[1pt]
Cold start (flat) & 101 & 1.00 & 100 & 0.00 & 422 & 1.00 & 98 & 0.00 & \textbf{233} & 1.00 & 99 & 0.00 & 295 & 1.00 & 100 & 0.00 \\
DC warm start & 87 & 1.06 & 100 & 0.00 & 157 & 4.06 & 98 & 0.00 & -- & 0.44 & 69 & 0.00 & 282 & 0.93 & 99 & 0.00 \\
\noalign{\vskip 0.8pt}
\arrayrulecolor{gray!50}\cdashline{1-17}[2.4pt/2.2pt]\arrayrulecolor{black}
\noalign{\vskip 0.8pt}
MLP (DeepOPF \cite{pan20}) & 90 & 0.81 & 100 & 0.00 & 140 & 4.00 & 99 & 0.00 & 586 & 0.60 & 90 & 0.00 & -- & \textcolor{gray!60}{1.54\,$^{\dagger}$} & 3 & 0.00 \\
GNN & 72 & 1.18 & 100 & 0.00 & 140 & 4.09 & 100 & 0.00 & 369 & 0.89 & 95 & 0.00 & -- & \textcolor{gray!60}{1.18\,$^{\dagger}$} & 3 & 0.00 \\
Linux Foundation \cite{Pue26} & 72 & 1.06 & 100 & 0.00 & 175 & 2.88 & 100 & 0.00 & -- & 0.52 & 83 & 0.00 & 356 & 0.84 & 97 & 0.00 \\
\noalign{\vskip 0.8pt}
\arrayrulecolor{gray!50}\cdashline{1-17}[2.4pt/2.2pt]\arrayrulecolor{black}
\noalign{\vskip 0.8pt}
GridSFM (no pretrain) & 69 & \textbf{1.43} & 100 & 0.00 & 112 & 5.75 & 99 & 0.00 & -- & 0.49 & 82 & 0.00 & -- & \textcolor{gray!60}{0.56\,$^{\dagger}$} & 19 & 0.00 \\
GridSFM (zero-shot) & 73 & 1.22 & 100 & 0.00 & \textbf{92} & 5.99 & 99 & 0.00 & -- & 0.50 & 81 & 0.00 & -- & -- & 0 & -- \\
GridSFM (fine-tuned) & \textbf{68} & 1.32 & 100 & 0.00 & 103 & \textbf{7.56} & 100 & 0.00 & 340 & \textbf{1.02} & 98 & 0.00 & \textbf{87} & \textbf{3.28} & 100 & 0.00 \\
\midrule
\multicolumn{17}{@{}l}{\itshape AC-OPF feasible projection}\\[1pt]
MLP (DeepOPF \cite{pan20}) & \textbf{32} & 1.95 & 100 & 9.41 & 53 & 7.88 & 100 & 2.56 & \textbf{67} & \textbf{3.59} & 100 & 32.88 & -- & \textcolor{gray!60}{1.36\,$^{\dagger}$} & 8 & 2.59 \\
GNN & 34 & 1.78 & 100 & 0.86 & 58 & 7.01 & 100 & 1.97 & 87 & 2.94 & 100 & 3.98 & -- & \textcolor{gray!60}{1.83\,$^{\dagger}$} & 4 & 0.36 \\
Linux Foundation \cite{Pue26} & 36 & \textbf{2.05} & 100 & 1.69 & 76 & 5.44 & 100 & 65.81 & 205 & 1.18 & 100 & 36.54 & 276 & 1.11 & 100 & 3.22 \\
\noalign{\vskip 0.8pt}
\arrayrulecolor{gray!50}\cdashline{1-17}[2.4pt/2.2pt]\arrayrulecolor{black}
\noalign{\vskip 0.8pt}
GridSFM (no pretrain) & 33 & 1.75 & 100 & 0.31 & 108 & 6.85 & 100 & 0.29 & 396 & 0.53 & 100 & 1.83 & -- & 1.04 & 76 & 0.07 \\
GridSFM (zero-shot) & 36 & 1.84 & 100 & 0.60 & \textbf{46} & 8.48 & 100 & 0.88 & 310 & 0.72 & 100 & 2.65 & -- & -- & 0 & -- \\
GridSFM (fine-tuned) & 35 & 1.85 & 100 & 0.55 & \textbf{46} & \textbf{9.10} & 100 & 0.67 & 74 & 3.20 & 100 & 2.04 & \textbf{45} & \textbf{6.36} & 100 & 0.57 \\
\bottomrule
\end{tabular}
\setlength{\arrayrulewidth}{0.4pt}
\begin{minipage}{\textwidth}\vspace{3pt}\footnotesize

\end{minipage}
\vspace{-0.5cm}
\end{table*}

\begin{figure*}
    \centering
    \includegraphics{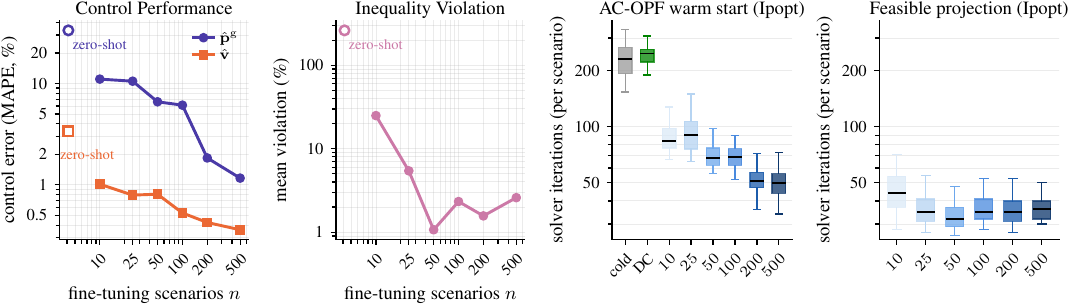}
    \caption{Scaling performance for fine-tuning GridSFM on \texttt{ACTIVSg10k} against the
number of fine-tuning scenarios $n$, on the same $100$ held-out test cases as Figure \ref{fig:iter-cdf-cpu}.
Inequality violation is the equal-weight mean of generator
active and reactive power, bus voltage magnitude, and branch apparent flow each
normalized by the bounds and averaged over every element of each class over each problem.
\vspace{-0.5cm}}
\label{fig:activsg10k-scaling}
\end{figure*}

\section{Conclusion}
In this paper, we presented GridSFM, a $15$ million parameter foundation model pretrained across $54$ transmission topologies ranging from $500$ to $4{,}000$ buses for solving AC-OPF.  Our approach is grounded in a slacked reformulation of AC-OPF and we 
show that this removes one obstruction to learning AC-OPF with disconnected feasible sets.
From this slacked design, we develop a graph transformer architecture that encodes the physics of power-flow through signed message passing while maintaining a global state of the system through linearized attention. With this architecture, the  pretrained GridSFM backbone attains a zero-shot $2.45\%$ generation-cost error with no degradation in accuracy as system size grows. 

 In consort with the pretrained model, we introduce a physics-informed fine-tuning design.
 With $100$ solved instance, this design extrapolates GridSFM to the $6{,}470$-bus \texttt{case6470\_rte} and $10{,}000$-bus \texttt{ACTIVSg10k} systems. On downstream tasks, we see GridSFM warm-starts improve Ipopt  up to $7.6\times$ fewer iterations on \texttt{Texas2k}. Using a simple projection algorithm, GridSFM points can be projected in $45$ ipopt iterations on a $10{,}000$-bus at a $0.57\%$ cost gap, while other neural network baselines typically fail to converge.

\bibliographystyle{IEEEtran}
\bibliography{references}

\begin{thebibliography}{10}
\providecommand{\url}[1]{#1}
\csname url@samestyle\endcsname
\providecommand{\newblock}{\relax}
\providecommand{\bibinfo}[2]{#2}
\providecommand{\BIBentrySTDinterwordspacing}{\spaceskip=0pt\relax}
\providecommand{\BIBentryALTinterwordstretchfactor}{4}
\providecommand{\BIBentryALTinterwordspacing}{\spaceskip=\fontdimen2\font plus
\BIBentryALTinterwordstretchfactor\fontdimen3\font minus \fontdimen4\font\relax}
\providecommand{\BIBforeignlanguage}[2]{{%
\expandafter\ifx\csname l@#1\endcsname\relax
\typeout{** WARNING: IEEEtran.bst: No hyphenation pattern has been}%
\typeout{** loaded for the language `#1'. Using the pattern for}%
\typeout{** the default language instead.}%
\else
\language=\csname l@#1\endcsname
\fi
#2}}
\providecommand{\BIBdecl}{\relax}
\BIBdecl

\bibitem{khaloie2025review}
H.~Khaloie, M.~Dolányi, J.-F. Toubeau, and F.~Vallée, ``Review of machine learning techniques for optimal power flow,'' \emph{Applied Energy}, vol. 388, p. 125637, 2025.

\bibitem{5491276}
R.~D. Zimmerman, C.~E. Murillo-Sánchez, and R.~J. Thomas, ``Matpower: Steady-state operations, planning, and analysis tools for power systems research and education,'' \emph{IEEE Transactions on Power Systems}, vol.~26, no.~1, pp. 12--19, 2011.

\bibitem{7725528}
A.~B. Birchfield, T.~Xu, K.~M. Gegner, K.~S. Shetye, and T.~J. Overbye, ``{Grid Structural Characteristics as Validation Criteria for Synthetic Networks},'' \emph{IEEE Transactions on Power Systems}, vol.~32, no.~4, pp. 3258--3265, 2017.

\bibitem{pan_chen_ml_opf_wiki}
\BIBentryALTinterwordspacing
X.~Pan and M.~Chen. (2021) {Machine Learning for Solving Optimal Power Flow Problems}. [Online]. Available: \url{https://personal.cityu.edu.hk/mchen88/projects/ML_OPF_wiki.html}
\BIBentrySTDinterwordspacing

\bibitem{Yan00}
\BIBentryALTinterwordspacing
W.~Yang, A.~Britto, T.~Spina, S.~Fowers, B.~Zhang, and C.~M. White, ``{GridSFM: A Foundation Model for AC Optimal Power Flow},'' \emph{Microsoft Research Technical White Paper}, 2026. [Online]. Available: \url{https://www.microsoft.com/en-us/research/publication/gridsfm-a-foundation-model-for-ac-optimal-power-flow/}
\BIBentrySTDinterwordspacing

\bibitem{bhan_gridsfm}
\BIBentryALTinterwordspacing
L.~Bhan, W.~Yang, M.~Capetz, and B.~Zhang, ``{GridSFM Code Repository}.'' [Online]. Available: \url{https://github.com/lukebhan/gridsfm}
\BIBentrySTDinterwordspacing

\bibitem{hiskens2001exploring}
I.~A. Hiskens and R.~J. Davy, ``Exploring the power flow solution space boundary,'' \emph{IEEE transactions on power systems}, vol.~16, no.~3, pp. 389--395, 2001.

\bibitem{cybenko1989approximation}
G.~Cybenko, ``Approximation by superpositions of a sigmoidal function,'' \emph{Mathematics of Control, Signals and Systems}, vol.~2, no.~4, pp. 303--314, Dec. 1989.

\bibitem{NIPS2017_5dd9db5e}
W.~Hamilton, Z.~Ying, and J.~Leskovec, ``{Inductive Representation Learning on Large Graphs},'' in \emph{Advances in Neural Information Processing Systems}, I.~Guyon, U.~V. Luxburg, S.~Bengio, H.~Wallach, R.~Fergus, S.~Vishwanathan, and R.~Garnett, Eds., vol.~30.\hskip 1em plus 0.5em minus 0.4em\relax Curran Associates, Inc., 2017.

\bibitem{donti2021dc}
P.~L. Donti, D.~Rolnick, and J.~Z. Kolter, ``{DC3: A learning method for optimization with hard constraints},'' in \emph{International Conference on Learning Representations}, 2021.

\bibitem{pglib}
{IEEE PES Task Force on Benchmarks for Validation of Emerging Power System Algorithms}, ``{The Power Grid Library for Benchmarking AC Optimal Power Flow Algorithms},'' \emph{arXiv preprint arXiv:1908.02788}, 2021, version 23.07.

\bibitem{4956966}
B.~Stott, J.~Jardim, and O.~Alsac, ``{DC Power Flow Revisited},'' \emph{IEEE Transactions on Power Systems}, vol.~24, no.~3, pp. 1290--1300, 2009.

\bibitem{zam19}
A.~Zamzam and K.~Baker, ``{Learning Optimal Solutions for Extremely Fast AC Optimal Power Flow},'' \emph{IEEE International Conference on Communications, Control, and Computing Technologies for Smart Grids}, pp. 1--6, 2020.

\bibitem{fio19}
F.~Fioretto, T.~W. Mak, and P.~V. Hentenryck, ``{Predicting AC Optimal Power Flows: Combining Deep Learning and Lagrangian Dual Methods},'' \emph{Proceedings of the AAAI Conference on Artificial Intelligence}, vol.~34, no.~1, pp. 630--637, 2020.

\bibitem{pan20}
X.~Pan, M.~Chen, T.~Zhao, and S.~Low, ``{DeepOPF: A Feasibility-Optimized Deep Neural Network Approach for AC Optimal Power Flow Problems},'' \emph{IEEE Systems Journal}, vol.~17, no.~1, pp. 673--683, 2023.

\bibitem{Hua24b}
W.~Huang, M.~Chen, and S.~H. Low, ``{Unsupervised Learning for Solving AC Optimal Power Flows: Design, Analysis, and Experiment},'' \emph{IEEE Transactions on Power Systems}, vol.~39, no.~6, pp. 7102--7114, 2024.

\bibitem{Kim25c}
M.~Kim and H.~Kim, ``{Unsupervised Deep Lagrange Dual With Equation Embedding for AC Optimal Power Flow},'' \emph{IEEE Transactions on Power Systems}, vol.~40, no.~1, pp. 1078--1090, 2025.

\bibitem{Don20}
B.~Donon, R.~Clément, B.~Donnot, A.~Marot, I.~M. Guyon, and M.~Schoenauer, ``{Neural networks for power flow: Graph neural solver},'' \emph{Electric Power Systems Research}, vol. 189, p. 106547, 2020.

\bibitem{Fal21}
T.~Falconer and L.~Mones, ``{Leveraging Power Grid Topology in Machine Learning Assisted Optimal Power Flow},'' \emph{IEEE Transactions on Power Systems}, vol.~38, no.~3, pp. 2234--2246, 2023.

\bibitem{Zha20}
T.~Zhao, X.~Pan, M.~Chen, A.~Venzke, and S.~Low, ``{DeepOPF+: A Deep Neural Network Approach for DC Optimal Power Flow for Ensuring Feasibility},'' \emph{2020 IEEE International Conference on Communications, Control, and Computing Technologies for Smart Grids (SmartGridComm)}, 2020.

\bibitem{JMLR:v25:23-1577}
E.~Liang, M.~Chen, and S.~H. Low, ``{Homeomorphic Projection to Ensure Neural-Network Solution Feasibility for Constrained Optimization},'' \emph{Journal of Machine Learning Research}, vol.~25, no. 329, pp. 1--55, 2024.

\bibitem{Ngu25}
H.~Nguyen and P.~Donti, ``{FSNet: Feasibility-Seeking Neural Network for Constrained Optimization with Guarantees},'' \emph{Advances in Neural Information Processing Systems}, pp. 39\,670--39\,708, 2025.

\bibitem{Liu22f}
S.~Liu, C.~Wu, and H.~Zhu, ``{Topology-Aware Graph Neural Networks for Learning Feasible and Adaptive AC-OPF Solutions},'' \emph{IEEE Transactions on Power Systems}, vol.~38, no.~6, pp. 5660--5670, 2023.

\bibitem{Zho23}
M.~Zhou, M.~Chen, and S.~Low, ``{DeepOPF-FT: One Deep Neural Network for Multiple AC-OPF Problems With Flexible Topology},'' \emph{IEEE Transactions on Power Systems}, vol.~38, no.~1, pp. 964--967, 2023.

\bibitem{Pil24}
L.~Piloto, S.~Liguori, S.~Madjiheurem, M.~Zgubič, S.~Lovett, H.~Tomlinson, S.~Elster, C.~Apps, and S.~Witherspoon, ``{CANOS: A Fast and Scalable Neural AC-OPF Solver Robust To N-1 Perturbations},'' \emph{ArXiv}, 2024.

\bibitem{Pue26}
A.~Puech \emph{et~al.}, ``{GENCO: A Unified Neural Solver Embedded in a Development Framework for Steady-State Grid Analysis},'' \emph{ArXiv}, 2026.

\bibitem{Par26}
K.~Park, K.~Song, Y.~Lim, S.~Park, K.~Kim, and H.~Kim, ``{UNION: A Unified AC-OPF Framework for Topology-Varying Real-Time Grid Operation},'' \emph{ArXiv}, 2026.

\bibitem{Li26c}
Y.~Li, Z.~Memon, H.~Jin, S.~Fenu, K.~Song, S.~Sharma, P.~Gasana, H.~Kim, L.~Zhao, and K.~Kim, ``{LUMINA: Foundation Models for Topology Transferable ACOPF},'' \emph{ArXiv}, 2026.

\bibitem{Low2026}
S.~Low, \emph{Power System Analysis - Analytical Tools and Structural Properties}.\hskip 1em plus 0.5em minus 0.4em\relax Cambridge University Press, 2026.

\bibitem{zhang2012geometry}
B.~Zhang and D.~Tse, ``Geometry of injection regions of power networks,'' \emph{IEEE Transactions on Power Systems}, vol.~28, no.~2, pp. 788--797, 2013.

\bibitem{6756976}
S.~H. Low, ``{Convex Relaxation of Optimal Power Flow—Part I: Formulations and Equivalence},'' \emph{IEEE Transactions on Control of Network Systems}, vol.~1, no.~1, pp. 15--27, 2014.

\bibitem{zhang2014network}
B.~Zhang, R.~Rajagopal, and D.~Tse, ``Network risk limiting dispatch: Optimal control and price of uncertainty,'' \emph{IEEE Transactions on Automatic Control}, vol.~59, no.~9, pp. 2442--2456, 2014.

\bibitem{dwivedi2022graph}
V.~P. Dwivedi, A.~T. Luu, T.~Laurent, Y.~Bengio, and X.~Bresson, ``{Graph Neural Networks with Learnable Structural and Positional Representations},'' in \emph{International Conference on Learning Representations}, 2022.

\bibitem{katharopoulos20a}
A.~Katharopoulos, A.~Vyas, N.~Pappas, and F.~Fleuret, ``{Transformers are RNNs: Fast Autoregressive Transformers with Linear Attention},'' in \emph{ICML}, 2020, pp. 5156--5165.

\bibitem{4073219}
W.~F. Tinney and C.~E. Hart, ``{Power Flow Solution by Newton's Method},'' \emph{IEEE Transactions on Power Apparatus and Systems}, vol. PAS-86, no.~11, pp. 1449--1460, 1967.

\bibitem{336130}
T.~Overbye, ``{A power flow measure for unsolvable cases},'' \emph{IEEE Transactions on Power Systems}, vol.~9, no.~3, pp. 1359--1365, 1994.

\bibitem{373951}
------, ``{Computation of a practical method to restore power flow solvability},'' \emph{IEEE Transactions on Power Systems}, vol.~10, no.~1, pp. 280--287, 1995.

\bibitem{8391733}
A.~Hauswirth, S.~Bolognani, G.~Hug, and F.~Dörfler, ``Generic existence of unique lagrange multipliers in {AC} optimal power flow,'' \emph{IEEE Control Systems Letters}, vol.~2, no.~4, pp. 791--796, 2018.

\bibitem{lovett2024opfdatalargescaledatasetsac}
\BIBentryALTinterwordspacing
S.~Lovett, M.~Zgubic, S.~Liguori, S.~Madjiheurem, H.~Tomlinson, S.~Elster, C.~Apps, S.~Witherspoon, and L.~Piloto, ``Opfdata: Large-scale datasets for ac optimal power flow with topological perturbations,'' 2024. [Online]. Available: \url{https://arxiv.org/abs/2406.07234}
\BIBentrySTDinterwordspacing

\bibitem{britto2026buildingpowergridmodels}
\BIBentryALTinterwordspacing
A.~Britto, T.~Spina, W.~Yang, S.~Fowers, B.~Zhang, and C.~White, ``Building power grid models from open data: A complete pipeline from openstreetmap to optimal power flow,'' 2026. [Online]. Available: \url{https://arxiv.org/abs/2605.04289}
\BIBentrySTDinterwordspacing

\bibitem{shin2024accelerating}
S.~Shin, M.~Anitescu, and F.~Pacaud, ``Accelerating optimal power flow with {GPU}s: {SIMD} abstraction of nonlinear programs and condensed-space interior-point methods,'' \emph{Electric Power Systems Research}, vol. 236, p. 110651, 2024.

\bibitem{hu2020strategies}
W.~Hu, B.~Liu, J.~Gomes, M.~Zitnik, P.~Liang, V.~Pande, and J.~Leskovec, ``Strategies for pre-training graph neural networks,'' in \emph{International Conference on Learning Representations}, 2020.

\end{thebibliography}

\clearpage
\appendices

\renewcommand{\thesubsection}{\thesection.\arabic{subsection}}
\renewcommand{\thesubsectiondis}{\thesubsection}

\renewcommand{\thesubsubsection}{\thesubsection.\arabic{subsubsection}}
\renewcommand{\thesubsubsectiondis}{\thesubsubsection}
\numberwithin{equation}{section}
\renewcommand{\theequation}{\thesection\arabic{equation}}
\allowdisplaybreaks
\raggedbottom

\section{Proofs}\label{app:proofs} 

\subsection{Proof of Lemma \ref{lem:elastic-connected}}
\begin{proof}
Define the mapping $\Phi:\mathbb R^N\times\mathbb R^{n_e}\times\mathbb R^{n_i}\to\mathbb R^N\times\mathbb R^{n_e}\times\mathbb R^{n_i}$ by
\begin{align}
    \Phi(\mathbf x,\mathbf s^{\tilde g},\mathbf s^{\tilde h})
    :=\big(\mathbf x,\ \mathbf s^{\tilde g}-|\tilde{\mathbf g}(\mathbf x)|,\ \mathbf s^{\tilde h}-[\tilde{\mathbf h}(\mathbf x)]_+\big).
    \label{eq:Phi-def}
\end{align}
Since $\tilde{\mathbf g}$, $\tilde{\mathbf h}$, $|\cdot|$ and $[\,\cdot\,]_+$ are continuous, $\Phi$ is continuous, and it is a bijection with continuous inverse
\begin{align}
    \Phi^{-1}(\mathbf x,\mathbf u^{\tilde g},\mathbf u^{\tilde h})
    =\big(\mathbf x,\ \mathbf u^{\tilde g}+|\tilde{\mathbf g}(\mathbf x)|,\ \mathbf u^{\tilde h}+[\tilde{\mathbf h}(\mathbf x)]_+\big).
    \label{eq:Phi-inv}
\end{align}
Hence $\Phi$ is a homeomorphism of $\mathbb R^N\times\mathbb R^{n_e}\times\mathbb R^{n_i}$ onto itself. 
It remains to show that $\Phi$ maps $\Omega$ onto $\mathcal{C}=\mathbb{R}^N\times\mathbb{R}_{\geq 0}^{n_e}\times\mathbb{R}^{n_i}_{\geq 0}$. Notice that $(\mathbf x,\mathbf s^{\tilde g},\mathbf s^{\tilde h})\in\Omega$ if and only if
$\mathbf s^{\tilde g}-|\tilde{\mathbf g}(\mathbf x)|\ge\mathbf 0$ and $\mathbf s^{\tilde h}-[\tilde{\mathbf h}(\mathbf x)]_+\ge\mathbf 0$, which occurs precisely when $\Phi(\mathbf x,\mathbf s^{\tilde g},\mathbf s^{\tilde h})\in \mathcal{C}$. Hence, $\Phi(\Omega)=\mathcal{C}$. Since $\Phi$ is a homeomorphism of the ambient space, its restriction $\Phi|_{\Omega}:\Omega\to \mathcal{C}$ is a homeomorphism onto $\mathcal{C}$. Finally, $\mathcal{C}$ is a product of convex sets, hence convex, and therefore contractible.  Thus, $\Omega$ is contractible, and correspondingly simply connected.
\end{proof}

\subsection{Proof of Lemma \ref{lem:exact}} \label{appendix:lemma2proof}
\begin{proof}
 Let $f(\mathbf x,\mathbf s^{\tilde g},\mathbf s^{\tilde h})$ denote the objective in \eqref{eq:AC-OPF-elastic-lifted}.
Since $\log(1 + s)$ is increasing with $s$, for every $(\mathbf x,\mathbf s^{\tilde g},\mathbf s^{\tilde h})\in\Omega$ we have
$f(\mathbf x,\mathbf s^{\tilde g},\mathbf s^{\tilde h})\ge f_\rho(\mathbf x)$, where
\begin{align}
    f_\rho(\mathbf x):=&\,c(\mathbf x)+\rho_{\tilde g}\sum_{j\in\mathcal E}\log\!\big(1+|\tilde g_j(\mathbf x)|\big) \nonumber \\
    &\,+\rho_{\tilde h}\sum_{k\in\mathcal M}\log\!\big(1+[\tilde h_k(\mathbf x)]_+\big),
\end{align}
with equality when $\mathbf s^{\tilde g}=|\tilde{\mathbf g}(\mathbf x)|$ and $\mathbf s^{\tilde h}=[\tilde{\mathbf h}(\mathbf x)]_+$.
In particular $f(\mathbf x^\ast,\mathbf 0,\mathbf 0)=f_\rho(\mathbf x^\ast)=c(\mathbf x^\ast)$ since $\mathbf x^\ast\in\mathcal F$.
Hence $(\mathbf x^\ast,\mathbf 0,\mathbf 0)$ is a local minimizer of \eqref{eq:AC-OPF-elastic-lifted} as soon as
$\mathbf x^\ast$ is a local minimizer of $f_\rho$. Hence it suffices to show that $\mathbf x^\ast$ is a local minimizer of $f_\rho$, i.e.\ that
$f_\rho(\mathbf x)\ge c(\mathbf x^\ast)$ for all $\mathbf x$ near $\mathbf x^\ast$. We do so by showing that the cost at
$\mathbf x$ can fall below $c(\mathbf x^\ast)$ by at most $L_c\,\mathrm{dist}(\mathbf x,\mathcal F)$, while the penalty
at $\mathbf x$ is at least $L_c\,\mathrm{dist}(\mathbf x,\mathcal F)$.

We begin with the cost. Since $\mathbf x^\ast$ is a local minimizer of \eqref{eq:AC-OPF-lifted}, there is
$r_0\in(0,R_0]$ such that $c(\boldsymbol\xi)\ge c(\mathbf x^\ast)$ for all
$\boldsymbol\xi\in\mathcal F\cap\overline B(\mathbf x^\ast,r_0)$. Let $\mathbf x$ satisfy $\|\mathbf x-\mathbf x^\ast\|\le r_0/2$ and let
$\bar{\mathbf x}\in\mathcal F$ be a nearest feasible point. Define $\Delta\mathbf x:=\mathbf x-\bar{\mathbf x}$, so that
$\|\Delta\mathbf x\|=\mathrm{dist}(\mathbf x,\mathcal F)$. Since $\mathbf x^\ast\in\mathcal F$, we have
$\|\Delta\mathbf x\|\le\|\mathbf x-\mathbf x^\ast\|$ and hence $\|\bar{\mathbf x}-\mathbf x^\ast\|\le2\|\mathbf x-\mathbf x^\ast\|\le r_0$,
so $c(\bar{\mathbf x})\ge c(\mathbf x^\ast)$. Using $\|\nabla c\|\le L_c$ on the set
$\overline B(\mathbf x^\ast,R_0)$ (which contains the segment $[\bar{\mathbf x},\mathbf x]$), we obtain the bound
\begin{align}
    c(\mathbf x)\ \ge\ c(\bar{\mathbf x})-L_c\|\Delta\mathbf x\|\ \ge\ c(\mathbf x^\ast)-L_c\,\|\Delta\mathbf x\| .
    \label{eq:costlb}
\end{align}
Hence, the cost is at best, $L_c\|\Delta \mathbf{x}\|$ better then the minimizer $\mathbf{x}^\ast$. 

We now turn to the penalty, and first reduce it to a bound on the total violation. Define the total violation as 
\begin{align}
    \epsilon(\mathbf x):=\sum_{j\in\mathcal E}|\tilde g_j(\mathbf x)|+\sum_{k\in\mathcal M}[\tilde h_k(\mathbf x)]_+ .
\end{align}
On a sufficiently small neighborhood, we will show the inequality penalty is bounded by this $\epsilon$.

First, since \eqref{eq:rhocond} is strict we may fix $\sigma\in(0,\underline\sigma)$ with
$\rho_{\min}\sigma>L_c$. Define the constant
\begin{align}
    \eta:=\frac{\rho_{\min}\sigma}{L_c}-1\ >\ 0 .\label{eq:a5}
\end{align}
Now, by concavity of log, we have
 $\log(1+s) \geq s/(1+\eta)$ for $0 \leq s \leq \eta$. If we then shrink the neighborhood so that $|\tilde g_j(\mathbf x)|\le\eta$ and
$[\tilde h_k(\mathbf x)]_+\le\eta$ for all $j,k$ (which is possible since the violations are continuous
and vanish at $\mathbf x^\ast$), we obtain
\begin{align}
    \rho_{\tilde g}\!\sum_{j\in\mathcal E}\log\!\big(1+|\tilde g_j(\mathbf x)|\big)
    +\rho_{\tilde h}\!\sum_{k\in\mathcal M}\log\!\big(1+&[\tilde h_k(\mathbf x)]_+\big)\nonumber 
    \\  \ge&\,  \frac{\rho_{\min}}{1+\eta}\,\epsilon(\mathbf x)\, \nonumber \\
    =&\, \frac{L_c}{\sigma}\epsilon(\mathbf{x}) \label{eq:penlb}
\end{align}
where we used $1+\eta=\rho_{\min}\sigma/L_c$ from \eqref{eq:a5} for the last equality. It remains to  lower bound the violation rate $\epsilon(\mathbf{x})$ by $\Delta \mathbf{x}$.

To lower bound the violation rate, we will use Taylor's theorem. Let $\Lambda$ bound the second derivative of the constraints
$\sum_{j\in\mathcal E}\|\nabla^2\tilde g_j\|+\sum_{k\in\mathcal M}\|\nabla^2\tilde h_k\|$ on
$\overline B(\mathbf x^\ast,R_0)$, which is finite because the residuals are polynomial in
$(\mathbf v,\mathbf p^{\mathrm g},\mathbf q^{\mathrm g},\mathbf p_{\mathcal L},\mathbf q_{\mathcal L})$ and trigonometric in
$\boldsymbol\theta$. Since $\bar{\mathbf x}$ is feasible, $\tilde g_j(\bar{\mathbf x})=0$ for all $j$ and
$\tilde h_k(\bar{\mathbf x})=0$ for $k\in\mathcal A(\bar{\mathbf x})$, so for some $\mathbf z_j,\mathbf z_k$ on the segment
$[\bar{\mathbf x},\mathbf x]\subseteq\overline B(\mathbf x^\ast,R_0)$, we have via Taylor's theorem:
\begin{align}
    \tilde g_j(\mathbf x)&=\nabla\tilde g_j(\bar{\mathbf x})^\top\Delta\mathbf x+\tfrac12\Delta\mathbf x^\top\nabla^2\tilde g_j(\mathbf z_j)\Delta\mathbf x,
    && j\in\mathcal E, \label{eq:appendix-a7}\\
    \tilde h_k(\mathbf x)&=\nabla\tilde h_k(\bar{\mathbf x})^\top\Delta\mathbf x+\tfrac12\Delta\mathbf x^\top\nabla^2\tilde h_k(\mathbf z_k)\Delta\mathbf x,
    && k\in\mathcal A(\bar{\mathbf x}).\label{eq:appendix-a8}
\end{align}
Taking absolute values in \eqref{eq:appendix-a7} and the  positive parts in \eqref{eq:appendix-a8} (using $[a+b]_+\ge[a]_+-|b|$), with the definition of $\epsilon(\mathbf x)$ gives
\begin{align}
    \epsilon(\mathbf x)\ \ge&\, \epsilon^{\rm \nabla }(\mathbf x)-\tfrac{\Lambda}{2}\|\Delta\mathbf x\|^2,
  \\ 
    \epsilon^{\rm \nabla }(\mathbf x):=&\, \sum_{j\in\mathcal E}\big|\nabla\tilde g_j(\bar{\mathbf x})^\top \Delta\mathbf x\big|
    +\!\!\sum_{k\in\mathcal A(\bar{\mathbf x})}\!\!\big[\nabla\tilde h_k(\bar{\mathbf x})^\top \Delta\mathbf x\big]_+ ,
    \label{eq:taylor-viol}
\end{align}
and
therefore suffices to bound $\epsilon^{\rm \nabla}(\mathbf x)$. 

To do so, we will invoke LICQ on the projection. Explicitly, note the point $\bar{\mathbf x}$ minimizes $\tfrac12\|\mathbf y-\mathbf x\|^2$
over $\mathcal F$, i.e.\ it solves the projection \eqref{eq:AC-OPF-projection} of $\mathbf x$, whose Lagrangian is
\begin{align}
    \mathcal L_{\rm proj}(\mathbf y;\boldsymbol\lambda,\boldsymbol\mu)
    :=\tfrac12\|\mathbf y-\mathbf x\|^2+\sum_{j\in\mathcal E}\lambda_j\tilde g_j(\mathbf y)+\sum_{k\in\mathcal M}\mu_k\tilde h_k(\mathbf y).
\end{align}
Since LICQ holds at $\bar{\mathbf x}$ by \eqref{eq:licq}, the KKT conditions hold there. Thus, there exist
multipliers $\boldsymbol\lambda\in\mathbb R^{n_e}$ and $\boldsymbol\mu\in\mathbb R^{n_i}_{\ge0}$ with $\mu_k=0$ for
$k\notin\mathcal A(\bar{\mathbf x})$ such that $\nabla_{\mathbf y}\mathcal L_{\rm proj}(\bar{\mathbf x};\boldsymbol\lambda,\boldsymbol\mu)=\mathbf 0$. Thus, we have
\begin{align}
   \mathbf x-\bar{\mathbf x}
    =\sum_{j\in\mathcal E}\lambda_j\nabla\tilde g_j(\bar{\mathbf x})+\sum_{k\in\mathcal A(\bar{\mathbf x})}\mu_k\nabla\tilde h_k(\bar{\mathbf x})
    =\mathbf J(\bar{\mathbf x})^\top\mathbf w,\label{eq:a12}
\end{align}
where $\mathbf J(\bar{\mathbf x})$ is the active-constraint Jacobian \eqref{eq:active-jacobian} and
$\mathbf w:=\big(\boldsymbol\lambda,(\mu_k)_{k\in\mathcal A(\bar{\mathbf x})}\big)$ stacks the corresponding multipliers.
As $\mathbf J(\bar{\mathbf x})$ has full row rank with $\sigma_{\min}(\mathbf J(\bar{\mathbf x}))\ge\underline\sigma$, we
have
\begin{align}
   \|\mathbf x-\bar{\mathbf x}\| =  \|\Delta\mathbf x\|=\|\mathbf J(\bar{\mathbf x})^\top\mathbf w\|\ge\underline\sigma\|\mathbf w\|\label{eq:a13}\,. 
\end{align}
Now, taking the inner product of \eqref{eq:a12} with $\Delta\mathbf x$ and bounding each term using
$|\lambda_j|\le\|\mathbf w\|$, $0\le\mu_k\le\|\mathbf w\|$ and $u\le[u]_+$, we obtain
\begin{align}
    \|\Delta\mathbf x\|^2 =\ & \mathbf w^\top \mathbf J(\bar{\mathbf x})\Delta\mathbf x  \nonumber \\
    \le\ & \|\mathbf w\|\Big(\sum_{j\in\mathcal E}\big|\nabla\tilde g_j(\bar{\mathbf x})^\top \Delta\mathbf x\big|
    +\!\!\sum_{k\in\mathcal A(\bar{\mathbf x})}\!\!\big[\nabla\tilde h_k(\bar{\mathbf x})^\top \Delta\mathbf x\big]_+\Big), \label{eq:a14}
\end{align}
Substituting the multiplier bound $\|\mathbf w\|\le\|\Delta\mathbf x\|/\underline\sigma$ from \eqref{eq:a13} into
\eqref{eq:a14} gives $\|\Delta\mathbf x\|^2\le\|\Delta\mathbf x\|\,\epsilon^{\nabla}(\mathbf x)/\underline\sigma$, and
dividing by $\|\Delta\mathbf x\|$ (the case $\Delta\mathbf x=\mathbf 0$ being trivial),
\begin{align}
    \epsilon^{\nabla}(\mathbf x)\ \ge\ \underline\sigma\,\|\Delta\mathbf x\| .
    \label{eq:linviol}
\end{align}
Inserting \eqref{eq:linviol} into \eqref{eq:taylor-viol} yields
\begin{align}
    \epsilon(\mathbf x)\ \ge\ \Big(\underline\sigma-\tfrac{\Lambda}{2}\|\Delta\mathbf x\|\Big)\|\Delta\mathbf x\| .
    \label{eq:errbound}
\end{align}
Finally, shrink the neighborhood so that $\|\mathbf x-\mathbf x^\ast\|\le2(\underline\sigma-\sigma)/\Lambda$.
Since $\|\Delta\mathbf x\|\le\|\mathbf x-\mathbf x^\ast\|$, this gives $\tfrac{\Lambda}{2}\|\Delta\mathbf x\|\le\underline\sigma-\sigma$,
and \eqref{eq:errbound} becomes
\begin{align}
    \epsilon(\mathbf x)\ \ge\ \sigma\,\|\Delta\mathbf x\| .
    \label{eq:errbound-goal}
\end{align}
Substituting \eqref{eq:errbound-goal} into the penalty bound \eqref{eq:penlb}, the penalty at $\mathbf x$
satisfies
\begin{align}
    \rho_{\tilde g}\!\sum_{j\in\mathcal E}\log\!\big(1+|\tilde g_j(\mathbf x)|\big)
    +\rho_{\tilde h}\!\sum_{k\in\mathcal M}\log\!\big(1+&[\tilde h_k(\mathbf x)]_+\big) \nonumber \\ 
    \ \ge &\, \frac{L_c}{\sigma}\,\sigma\,\|\Delta\mathbf x\|\  \nonumber \\ =&\ L_c\,\|\Delta\mathbf x\| \,. 
    \label{eq:penfinal}
\end{align}

Combining \eqref{eq:costlb} and \eqref{eq:penfinal}, there exists a positive radius such that
\begin{align}
    f_\rho(\mathbf x)\ \ge\ c(\mathbf x^\ast)-L_c\|\Delta\mathbf x\|+L_c\|\Delta\mathbf x\|\ =\ c(\mathbf x^\ast)\ =\ f_\rho(\mathbf x^\ast).
\end{align}
Thus $\mathbf x^\ast$ minimizes $f_\rho$ locally, and $(\mathbf x^\ast,\mathbf 0,\mathbf 0)$ is a local minimizer of
\eqref{eq:AC-OPF-elastic-lifted}.
\end{proof}

\subsection{Proof of Lemma \ref{lem:reach}} \label{appendix:bowl-proof}
\begin{proof}
The Lagrangian of the projection \eqref{eq:AC-OPF-projection} is given by:
\begin{align}
    \mathcal{L}_{\rm proj}(\mathbf y; \boldsymbol\lambda, \boldsymbol\mu) := \frac{1}{2} \|\mathbf y- \hat{\mathbf x}\|^2 + \sum_{j\in\mathcal E} \lambda_j  \tilde g_j(\mathbf y) + \sum_{m\in\mathcal M} \mu_m \tilde h_m(\mathbf y)\,,
\end{align}
for $\mu_m \geq 0$, $\lambda_j \in \mathbb{R}$. A stationary point $\boldsymbol\xi$ satisfies $\nabla_{\mathbf y} \mathcal{L}_{\rm proj}(\boldsymbol\xi) = \mathbf 0$, hence we have
\begin{align}
    \hat{\mathbf x} - \boldsymbol\xi = \sum_{j\in\mathcal E} \lambda_j  \nabla \tilde g_j(\boldsymbol\xi) + \sum_{m\in\mathcal M} \mu_m\nabla  \tilde h_m(\boldsymbol\xi)\,.
\end{align}
Now, using the quadratic expansion
\begin{align}
    \|\hat{\mathbf x}-\boldsymbol\zeta\|^2 = \|\hat{\mathbf x}-\boldsymbol\xi\|^2 - 2 (\hat{\mathbf x} - \boldsymbol\xi)^\top (\boldsymbol\zeta - \boldsymbol\xi)  + \|\boldsymbol\zeta - \boldsymbol\xi\|^2\,,
\end{align}
it remains to bound the cross term. By direct calculation, we have
\begin{align}
  \nonumber  (\hat{\mathbf x} - \boldsymbol\xi)^\top (\boldsymbol\zeta - \boldsymbol\xi) =&\, \sum_{j\in\mathcal E} \lambda_j\, \nabla \tilde g_j(\boldsymbol\xi)^\top (\boldsymbol\zeta-\boldsymbol\xi) \\ &\, + \sum_{m\in\mathcal M} \mu_m\, \nabla \tilde h_m(\boldsymbol\xi)^\top (\boldsymbol\zeta - \boldsymbol\xi)
\end{align}
Notice, that for each $m$ with $\mu_m>0$ complementary slackness gives $\tilde h_m(\boldsymbol\xi) = 0$, so since each $\tilde h_m$ is convex and $\tilde h_m(\boldsymbol\zeta) \leq 0$, we have
\begin{align}
    \nabla  \tilde h_m(\boldsymbol\xi)^\top (\boldsymbol\zeta - \boldsymbol\xi)  \leq \tilde h_m(\boldsymbol\zeta) - \tilde h_m(\boldsymbol\xi) \leq 0\,, \quad \forall m\,,
\end{align}
and hence since $\mu_m \geq 0$, this term contributes nothing to the upper bound. For the first term, note that each $\tilde g_j$ is twice continuously differentiable. Thus, we can use Taylor's theorem, which says there exists some $\mathbf z_j \in [\boldsymbol\xi, \boldsymbol\zeta]$ such that
\begin{align}
  \nonumber  \tilde g_j(\boldsymbol\zeta) =&\, \tilde g_j(\boldsymbol\xi) + \nabla \tilde g_j(\boldsymbol\xi)^\top (\boldsymbol\zeta-\boldsymbol\xi)  \\ &\, + \frac{1}{2} (\boldsymbol\zeta - \boldsymbol\xi)^\top \nabla^2 \tilde g_j(\mathbf z_j) (\boldsymbol\zeta - \boldsymbol\xi)\,, \quad \forall j
\end{align}
Using the fact that $\boldsymbol\zeta, \boldsymbol\xi$ are both feasible, we have $\tilde g_j(\boldsymbol\zeta) = \tilde g_j(\boldsymbol\xi) = 0$.

Rearranging and weighting by the multipliers, the equality terms collect into the single matrix $\mathbf A:=\sum_{j\in\mathcal E} \lambda_j \nabla^2 \tilde g_j(\mathbf z_j)$ and we obtain
\begin{align}
      (\hat{\mathbf x} - \boldsymbol\xi)^\top (\boldsymbol\zeta - \boldsymbol\xi)  \leq \frac{1}{2}\left|(\boldsymbol\zeta-\boldsymbol\xi)^\top \mathbf A (\boldsymbol\zeta-\boldsymbol\xi)\right| \leq \frac{1}{2} \|\mathbf A\| \|\boldsymbol\zeta - \boldsymbol\xi\|^2\,.
      \label{eq:crossA}
\end{align}

It remains to bound $\|\mathbf A\|$ by a constant depending only on the network. We first identify
the nonlinear terms of $\tilde{\mathbf g}$, then bound the Hessian of each, and finally assemble the
pieces.

The equality residuals are the flow definitions
$\tilde p_\ell-p_\ell(\mathbf v,\boldsymbol\theta)$ and
$\tilde q_\ell-q_\ell(\mathbf v,\boldsymbol\theta)$ for $\ell\in\mathcal L$,  and the
power balances $p_i^{\mathrm g}-p_i^{\mathrm d}-p_i(\mathbf v,\boldsymbol\theta)$ and
$q_i^{\mathrm g}-q_i^{\mathrm d}-q_i(\mathbf v,\boldsymbol\theta)$ for $i\in\mathcal N_0$. All dependence on
$(\mathbf p^{\mathrm g},\mathbf q^{\mathrm g},\tilde{\mathbf p}_{\mathcal L},\tilde{\mathbf q}_{\mathcal L})$ is linear, so the Hessians are
those of $p_i,q_i,p_\ell,q_\ell$ with respect to $(\mathbf v,\boldsymbol\theta)$. Write each entry of the bus
admittance matrix in polar form, $G_{ik}+\mathrm jB_{ik}=Y_{ik}=|Y_{ik}|e^{\mathrm j\psi_{ik}}$. The identities
$G_{ik}\cos\phi+B_{ik}\sin\phi=|Y_{ik}|\cos(\phi-\psi_{ik})$ and
$G_{ik}\sin\phi-B_{ik}\cos\phi=|Y_{ik}|\sin(\phi-\psi_{ik})$ turn \eqref{eq:nodal-injections} and
\eqref{eq:branch-flows} into
\begin{subequations}
\begin{align}
    p_i(\mathbf v,\boldsymbol\theta)&=G_{ii}\,v_i^2+\sum_{k\ne i}|Y_{ik}|\,v_iv_k\cos(\theta_{ik}-\psi_{ik}),\\
    q_i(\mathbf v,\boldsymbol\theta)&=-B_{ii}\,v_i^2+\sum_{k\ne i}|Y_{ik}|\,v_iv_k\sin(\theta_{ik}-\psi_{ik}),\\
    p_\ell(\mathbf v,\boldsymbol\theta)&=|Y_{ik}|\,v_iv_k\cos(\theta_{ik}-\psi_{ik}),\\
    q_\ell(\mathbf v,\boldsymbol\theta)&=|Y_{ik}|\,v_iv_k\sin(\theta_{ik}-\psi_{ik}),
\end{align}
\end{subequations}

for a branch $\ell=(i,k)$, where the sums run over the buses $k$ adjacent to $i$, since $Y_{ik}=0$
otherwise. Hence every nonlinear term in $\tilde{\mathbf g}$ is of one of two kinds: a \emph{cross term}
$|Y_{ik}|\,v_iv_k\cos(\theta_{ik}-\psi_{ik})$ or $|Y_{ik}|\,v_iv_k\sin(\theta_{ik}-\psi_{ik})$ attached to a
branch $\ell=(i,k)\in\mathcal L$, or a \emph{self term} $G_{ii}v_i^2$ or $-B_{ii}v_i^2$ attached to a bus
$i\in\mathcal N_0$. In particular, each $\tilde g_j$ is a linear function plus a
sum of cross and self terms, so $\nabla^2\tilde g_j(\mathbf z_j)$ is the sum of the Hessians of those
terms, which we now bound.

Since $\boldsymbol\xi,\boldsymbol\zeta$ are feasible and the voltage limits \eqref{eq:AC-OPF-f} form a convex box,
every point of the segment $[\boldsymbol\xi,\boldsymbol\zeta]$, in particular every $\mathbf z_j$, satisfies
$v_i\le\overline v$ for all $i\in\mathcal N_0$. Now, consider a cross term of branch $\ell=(i,k)$, say
$f_\ell:=|Y_{ik}|\,v_iv_k\cos(\theta_{ik}-\psi_{ik})$. It depends only on the four variables
$v_i,v_k,\theta_i,\theta_k$, so $\nabla^2 f_\ell$ is zero outside the $4\times4$ block on those
coordinates. A direct computation shows that every nonzero entry of this block is, up to sign, one of
\begin{align*}
    |Y_{ik}|\cos(\theta_{ik}-\psi_{ik}),\qquad
    |Y_{ik}|\,v_i\sin(\theta_{ik}-\psi_{ik}),\\
    |Y_{ik}|\,v_k\sin(\theta_{ik}-\psi_{ik}),\qquad
    |Y_{ik}|\,v_iv_k\cos(\theta_{ik}-\psi_{ik}),
\end{align*}
and summing the squares of all entries gives
\begin{align}
 \nonumber  \|\nabla^2 f_\ell\|_F^2
  =&\ |Y_{ik}|^2\Big[2\cos^2(\theta_{ik}-\psi_{ik}) \\\nonumber
  &\quad +4\sin^2(\theta_{ik}-\psi_{ik})\big(v_i^2+v_k^2\big) \\\nonumber
  &\quad +4v_i^2v_k^2\cos^2(\theta_{ik}-\psi_{ik})\Big] \\
  \le&\ 4|Y_{ik}|^2\big(1+\overline v^2\big)^2,
\end{align}
using $\cos^2,\sin^2\le1$ and $v_i,v_k\le\overline v$. As the operator norm is bounded by the
Frobenius norm, $\|\nabla^2 f_\ell\|\le2(1+\overline v^2)|Y_{ik}|$. The sine cross terms obey the same exact result. Consequently, for
any vector $\mathbf u$, a cross term of branch $\ell=(i,k)$ satisfies
\begin{align}
    |\mathbf u^\top\nabla^2 f_\ell\,\mathbf u|
    \ \le\ 2(1+\overline v^2)|Y_{ik}|\,\big(u_{v_i}^2+u_{\theta_i}^2+u_{v_k}^2+u_{\theta_k}^2\big),
    \label{eq:cross-quad}
\end{align}
since only the four coordinates of $\mathbf u$ on the block enter the quadratic form. A self term of
bus $i$, $G_{ii}v_i^2$ or $-B_{ii}v_i^2$, depends only on $v_i$, so its Hessian has a single nonzero
entry, $2G_{ii}$ or $-2B_{ii}$, in the $(v_i,v_i)$ position, and for any $\mathbf u$,
\begin{align}
    |\mathbf u^\top\nabla^2(G_{ii}v_i^2)\,\mathbf u|\le&2\, |Y_{ii}|\,u_{v_i}^2,    \label{eq:self-quad} \\ 
    |\mathbf u^\top\nabla^2(B_{ii}v_i^2)\,\mathbf u|\le&\, 2|Y_{ii}|\,u_{v_i}^2 .
    \label{eq:self-quad2}
\end{align}

Since $\mathbf A$ is symmetric, $\|\mathbf A\|=\max_{\|\mathbf u\|=1}|\mathbf u^\top\mathbf A\mathbf u|$. Fix a unit vector
$\mathbf u$. Expanding $\mathbf u^\top\mathbf A\mathbf u=\sum_j\lambda_j\,\mathbf u^\top\nabla^2\tilde g_j(\mathbf z_j)\,\mathbf u$ into
cross and self terms and bounding each by \eqref{eq:cross-quad}, \eqref{eq:self-quad} or \eqref{eq:self-quad2}, it remains
to count how often each term occurs. A branch appears in six residuals (the active and reactive
balances at its two ends and its two flow definitions) and a bus in two (its active and reactive
balances), so by Cauchy--Schwarz the multipliers attached to any one branch or bus have total
absolute value at most $\sqrt6\,\|\boldsymbol\lambda\|_2$. Hence
\begin{align}
   |\mathbf u^\top \mathbf A \mathbf u|
   \le&\ 2\sqrt6\,(1+\overline v^2)\,\|\boldsymbol\lambda\|_2 \nonumber \\ & \qquad \times \Big[\sum_{\ell=(i,k)\in\mathcal L}|Y_{ik}|\big(u_{v_i}^2+u_{\theta_i}^2+u_{v_k}^2+u_{\theta_k}^2\big) \nonumber\\
   &\qquad\qquad +\sum_{i\in\mathcal N_0}|Y_{ii}|\,u_{v_i}^2\Big] \nonumber\\
   \le&\ 2\sqrt6\,(1+\overline v^2)\,\|\boldsymbol\lambda\|_2
   \sum_{i\in\mathcal N_0}\Big(\sum_{k\in\mathcal N_0}|Y_{ik}|\Big)\big(u_{v_i}^2+u_{\theta_i}^2\big) \nonumber \\
   \le&\ 2\sqrt6\,(1+\overline v^2)\Big(\max_{i\in\mathcal N_0}\sum_{k\in\mathcal N_0}|Y_{ik}|\Big)\|\boldsymbol\lambda\|_2 ,
\end{align}
where the second inequality regroups by bus: the coordinates of bus $i$ appear once for each
incident branch $\ell=(i,k)\in\mathcal L_i$ with weight $|Y_{ik}|$, and once in the self term with weight
$|Y_{ii}|$. Therefore $\|\mathbf A\|\le M\|\boldsymbol\lambda\|_2$ with $M$ as in \eqref{eq:Mdelta}, using
$2\sqrt6\,(1+\overline v^2)\le10(1+\overline v)$ for $\overline v\le2$.

Finally, LICQ bounds $\|\boldsymbol\lambda\|_2$: writing the stationarity relation as $\hat{\mathbf x}-\boldsymbol\xi=\mathbf J(\boldsymbol\xi)^\top\mathbf w$ with $\mathbf w:=\big(\boldsymbol\lambda,(\mu_m)_{m\in\mathcal A(\boldsymbol\xi)}\big)$, since $\mu_m=0$ off the active set, \eqref{eq:licq} gives
$\|\boldsymbol\lambda\|_2\le\|\mathbf w\|_2\le\|\hat{\mathbf x}-\boldsymbol\xi\|/\underline{\sigma}$. Substituting this into
\eqref{eq:crossA} gives the \emph{cross-term bound}
\begin{align}
   \nonumber  (\hat{\mathbf x}-\boldsymbol\xi)^\top(\boldsymbol\zeta-\boldsymbol\xi)\ \le&\, \frac{\|\hat{\mathbf x}-\boldsymbol\xi\|\,M}{2\underline{\sigma}}\,\|\boldsymbol\zeta-\boldsymbol\xi\|^2
    \\ =&\, \frac{\|\hat{\mathbf x}-\boldsymbol\xi\|}{2\delta}\,\|\boldsymbol\zeta-\boldsymbol\xi\|^2 ,
    \label{eq:cross}
\end{align}
using $\delta=\underline{\sigma}/M$. Inserting \eqref{eq:cross} into the quadratic expansion
$\|\hat{\mathbf x}-\boldsymbol\zeta\|^2=\|\hat{\mathbf x}-\boldsymbol\xi\|^2-2(\hat{\mathbf x}-\boldsymbol\xi)^\top(\boldsymbol\zeta-\boldsymbol\xi)+\|\boldsymbol\zeta-\boldsymbol\xi\|^2$ yields
\begin{align}
    \|\hat{\mathbf x}-\boldsymbol\zeta\|^2\ \ge\ \|\hat{\mathbf x}-\boldsymbol\xi\|^2+\Big(1-\tfrac{\|\hat{\mathbf x}-\boldsymbol\xi\|}{\delta}\Big)\|\boldsymbol\zeta-\boldsymbol\xi\|^2,
\end{align}
which is \eqref{eq:reach-growth}.
\end{proof}

\section{Additional details on GridSFM architecture}\label{appendix:architecture}
In this section, we provide a full detailed description of the GridSFM design.
The section is broken into three components: the graph encoding, the core layer
design, and the output heads.

\subsection{Graph encoding} \label{appendix:graph-encoding}
As introduced in Section~\ref{sec:graph-encoding}, the input grid is encoded as
a heterogeneous graph whose nodes take one of seven types,
$\Sigma := \{\mathrm{bus}, \mathrm{gen}, \mathrm{load}, \mathrm{shunt},
\mathrm{line}, \mathrm{transformer}, \mathrm{cycle}\}$, and whose edges record
incidence: generators, loads, and shunts attach to their bus by an unsigned
edge, while each line and transformer attaches to its two endpoint buses by a
signed edge, $+1$ at the from-bus and $-1$ at the to-bus, and to every cycle
containing it by a second signed edge giving the direction in which that cycle
traverses it. Throughout, $\sigma(i) \in \Sigma$ is the type of node $i$, and in
subscripts we abbreviate the line and transformer types as $\mathrm{ac}$ and
$\mathrm{tr}$, respectively.

Let $\bm{f}_i$ denote the raw feature vector of node $i$. Since each node type
may carry features of different widths, we lift each type to a common hidden
dimension by a per-type affine map,
\begin{equation}
\label{eq:lift}
    \bm{z}^{(0)}_i = \bm{W}^{\mathrm{in}}_{\sigma(i)}
      \operatorname{LayerNorm}_{\sigma(i)}\!\bigl([\,\bm{f}_i ; \bm{e}_i^{\mathrm pos}\,]\bigr)
      + \bm{b}^{\mathrm{in}}_{\sigma(i)},
\end{equation}
where $[\,\cdot\,;\cdot\,]$ denotes concatenation. Here and throughout the
architecture, $\bm{z}$ denotes a hidden state of the network.

When node $i$'s type is a bus, a branch (AC line or transformer), or a cycle, it
is associated with a positional encoding $\bm{e}_i^{\rm pos}$ that captures its position
relative to the other nodes. It is computed through a diffusion operation on a
Laplacian defined from the topology of the input grid. This style
of encoding is not new \cite{dwivedi2022graph}. 

We begin by defining the Laplacians of the encoding. Let
$\bm{A}_{\mathrm{ac}}$ collect the bus--line signs, so that
$(\bm{A}_{\mathrm{ac}})_{ij} = +1$ if bus $i$ is the from-bus of line $j$, $-1$
if it is the to-bus, and $0$ if the two are not incident, and let
$\bm{A}_{\mathrm{tr}}$ collect the bus--transformer signs,
$\bm{A}^{\mathrm{cycle}}_{\mathrm{ac}}$ the line--cycle signs, and
$\bm{A}^{\mathrm{cycle}}_{\mathrm{tr}}$ the transformer--cycle signs in the same
way. These give one Laplacian per encoded node type,
\begin{subequations}
\label{eq:laplacians}
\begin{align}
    \bm{L}_{\mathrm{bus}}  =&\,
        \bm{A}_{\mathrm{ac}}\bm{A}_{\mathrm{ac}}^{\!\top}
        + \bm{A}_{\mathrm{tr}}\bm{A}_{\mathrm{tr}}^{\!\top} \,, \\
    \bm{L}_{\mathrm{ac}}   =&\,
        \bm{A}_{\mathrm{ac}}^{\!\top}\bm{A}_{\mathrm{ac}}
        + \bm{A}^{\mathrm{cycle}}_{\mathrm{ac}}
          (\bm{A}^{\mathrm{cycle}}_{\mathrm{ac}})^{\!\top} \,, \\
    \bm{L}_{\mathrm{tr}}   =&\,
        \bm{A}_{\mathrm{tr}}^{\!\top}\bm{A}_{\mathrm{tr}}
        + \bm{A}^{\mathrm{cycle}}_{\mathrm{tr}}
          (\bm{A}^{\mathrm{cycle}}_{\mathrm{tr}})^{\!\top} \,, \\
    \bm{L}_{\mathrm{cycle}} =&\,
        (\bm{A}^{\mathrm{cycle}}_{\mathrm{ac}})^{\!\top}
         \bm{A}^{\mathrm{cycle}}_{\mathrm{ac}}
        + (\bm{A}^{\mathrm{cycle}}_{\mathrm{tr}})^{\!\top}
          \bm{A}^{\mathrm{cycle}}_{\mathrm{tr}} \,,
\end{align}
\end{subequations}
so that $\bm{L}_\sigma$ is defined for every type
$\sigma \in \{\mathrm{bus}, \mathrm{ac}, \mathrm{tr}, \mathrm{cycle}\}$ that
carries an encoding. Only $\bm{L}_{\mathrm{ac}}$ and $\bm{L}_{\mathrm{tr}}$
couple a branch both to those it shares a bus with and to those it shares a
cycle with.

The positional encoding is then obtained by diffusing along $\bm{L}_\sigma$.
Define an encoding initialization as 
\begin{equation}
\label{eq:pe-init}
    \bm{e}^{(0)}_i = \bm{W}^{\mathrm{pe}}_{\sigma(i)}\bm{f}_i
      + \bm{b}^{\mathrm{pe}}_{\sigma(i)} \,,
\end{equation}
where $\bm{W}^{\mathrm{pe}}_{\sigma(i)}$ and $\bm{b}^{\mathrm{pe}}_{\sigma(i)}$
are learnable for each type, and let
$\bm{E}^{(t)}_\sigma := \operatorname{col}\{(\bm{e}^{(t)}_i)^{\!\top}\}_{i \in \mathcal{V}_\sigma}$
be the stack of the positional encodings of all nodes $\mathcal{V}_\sigma$ of
type $\sigma$. Then, define the diffusion for $T$ steps as
\begin{equation}
\label{eq:diffusion-pe}
    \bm{E}^{(t+1)}_\sigma = \bm{E}^{(t)}_\sigma
        - \bm{L}_\sigma \bm{E}^{(t)}_\sigma
          \operatorname{diag}\bigl(\bm{\alpha}^{(t)}_\sigma\bigr),
    \qquad t = 0,\dots,T-1,
\end{equation}
where $\bm{\alpha}^{(t)}_\sigma$ is a learnable per-channel step size.
Then, the positional encoding for each node is a learnable function of
the $T+1$ states it visited,
\begin{equation}
\label{eq:pe-readout}
        \bm{e}_i^{\rm pos} = \operatorname{LayerNorm}_{\sigma(i)}\Bigl(
      \mathrm{MLP}_{\sigma(i)}\bigl([\,\bm{e}^{(0)}_i ; \dots ;
      \bm{e}^{(T)}_i\,]\bigr)\Bigr) \,,
\end{equation}
where $\mathrm{MLP}_{\sigma(i)}$ is a two-layer network with GELU activations
mapping the $T+1$ concatenated states back to the corresponding position encoding channel dimension. Only
$\bm{L}_\sigma$ is fixed by the topology.  $\bm{W}^{\mathrm{pe}}_{\sigma(i)}$,
$\bm{b}^{\mathrm{pe}}_{\sigma(i)}$, $\bm{\alpha}^{(t)}_\sigma$, and
$\mathrm{MLP}_{\sigma(i)}$ are trained jointly with the rest of the network, so
the encoding is learned rather than precomputed.

\subsection{Neural Network Design}
\label{sec:nn-design-appendix}

After encoding, GridSFM consists of eight identical blocks combining linear
self-attention, signed message passing, and a position-wise feed-forward
network (see Figure~\ref{fig:arch-diagram}). Each block applies the three in
sequence in pre-normalized residual form:
\begin{subequations}
\label{eq:gridblock-appendix}
\begin{alignat}{2}
    \bm{z}^{\mathrm{attn}}_i &= \bm{z}_i
      + \operatorname{Attn}_i\bigl(\operatorname{LayerNorm}(\bm{z})\bigr)\,,
      &\hspace{1em}&  \\
    \bm{z}^{\mathrm{mp}}_i &= \bm{z}^{\mathrm{attn}}_i
      + \operatorname{MP}_i\bigl(
        \operatorname{LayerNorm}(\bm{z}^{\mathrm{attn}})\bigr)\,,\\
    \bm{z}^{\mathrm{ffn}}_i &= \bm{z}^{\mathrm{mp}}_i
      + \operatorname{FFN}\bigl(
        \operatorname{LayerNorm}(\bm{z}^{\mathrm{mp}}_i)\bigr)\,,
\end{alignat}
\end{subequations}
where $\bm{z}$ without a node index denotes the states of all nodes, and the
three layer normalizations are distinct, each carrying its own parameters for
each node type. Note that $\operatorname{Attn}_i$ and $\operatorname{MP}_i$
depend on the states of other nodes, while $\operatorname{FFN}$ acts on node
$i$ alone. Eight such blocks are stacked, the output $\bm{z}^{\mathrm{ffn}}_i$
of each becoming the input $\bm{z}_i$ of the next, with $\bm{z}^{(0)}_i$ from
\eqref{eq:lift} entering the first and all states of common hidden dimension
$d$. Every parameter below is learned separately for each node type
and shared by all nodes of that type. We now define the three operators,
writing $\tilde{\bm{z}}$ for whichever normalized state
\eqref{eq:gridblock-appendix} passes in, each of dimension $d$.

We begin with the normalized linear self-attention design
\cite{katharopoulos20a}. The layer acts within a node type: node $i$ attends to
the nodes of its own type $\sigma(i)$ and to no others. Let $\varphi$ be the
element-wise feature map of \cite{katharopoulos20a}, that is, $\varphi(x)=x+1$ for $x>0$ and
$\varphi(x)=e^x$ for $x \le 0$. The multi-head extension of the single-head
layer of Section~\ref{sec:nn-design} is
\begin{subequations}
\label{eq:attention-layer-design}
\begin{align}
    \operatorname{Attn}_i(\tilde{\bm{z}}) =&\, \mathrlap{\bm{W}^{o}_{\sigma(i)}
        \underbrace{\operatorname{Concat}\bigl(\bm{o}_i^{(1)}, \dots,
        \bm{o}_i^{(H)}\bigr)}_{\text{multi-head attention}} \,,} \\
    \bm{o}_i^{(m)} =&\,
        \frac{(\bm{S}_{\sigma(i)}^{(m)})^{\!\top}\varphi(\bm{q}_i^{(m)})}
             {\varphi(\bm{q}_i^{(m)})^{\!\top}\bm{\kappa}_{\sigma(i)}^{(m)}} \,,
        && \text{single head} \\
    \bm{S}_{\tau}^{(m)} =&\, \sum_{j \in \mathcal{V}_\tau}
        \varphi(\bm{k}_j^{(m)})\,(\bm{v}_j^{(m)})^{\!\top} \,,
        && \text{key--value product} \\
    \bm{\kappa}_{\tau}^{(m)} =&\, \sum_{j \in \mathcal{V}_\tau}
        \varphi(\bm{k}_j^{(m)}) \,,
        && \text{key normalization} \\
    \bm{q}_i^{(m)} =&\, \bm{W}^{q,(m)}_{\sigma(i)} \tilde{\bm{z}}_i \,,
        && \text{query embedding} \\
    \bm{k}_j^{(m)} =&\, \bm{W}^{k,(m)}_{\sigma(j)} \tilde{\bm{z}}_j \,,
        && \text{key embedding} \\
    \bm{v}_j^{(m)} =&\, \bm{W}^{v,(m)}_{\sigma(j)} \tilde{\bm{z}}_j \,,
        && \text{value embedding}
\end{align}
\end{subequations}
where $m = 1,\dots,H$ indexes the $H = 4$ attention heads, each projection
$\bm{W}^{q,(m)}_{\sigma(i)}$, $\bm{W}^{k,(m)}_{\sigma(j)}$,
$\bm{W}^{v,(m)}_{\sigma(j)}$ maps into $\mathbb{R}^{d/H}$, and
$\bm{W}^{o}_{\sigma(i)}$ is a learnable output projection applied to the
concatenated head outputs; with $H=1$ this reduces to the expression of
Section~\ref{sec:nn-design}.

We turn next to the signed message passing. Let $\mathcal{N}(i)$ denote the
nodes adjacent to $i$, and let $s_{ij}$ be the sign of the edge between $i$ and
$j$, equal to $\pm 1$ on the bus--branch and branch--cycle edges of
Section~\ref{sec:graph-encoding} and to $+1$ on the unsigned edges attaching
generators, loads, and shunts. The message-passing operator aggregates the
neighbors of each type separately,
\begin{align}
\operatorname{MP}_i(\tilde{\bm{z}})
= \sum_{\tau} \Bigl[&
  \underbrace{
    \operatorname*{mean}_{j \in \mathcal{N}(i) \cap \mathcal{V}_\tau}
    \bigl(s_{ij}\, \bm{W}_{\tau,\,\sigma(i)}\,\tilde{\bm{z}}_j\bigr)
  }_{\text{neighbor contributions}}
  \notag \\
&\qquad + \underbrace{
    \bm{W}'_{\tau,\,\sigma(i)}\,\tilde{\bm{z}}_i
  }_{\text{node contribution}}
  + \bm{b}_{\tau,\,\sigma(i)} \Bigr]\,,
\label{eq:mp-update-appendix}
\end{align}
where the sum runs over the node types $\tau$ adjacent to $\sigma(i)$ and the
bracketed term is taken as zero when $\mathcal{N}(i) \cap \mathcal{V}_\tau$ is
empty; composed with the residual of \eqref{eq:gridblock-appendix}, this is
exactly \eqref{eq:mp-update}. Each learnable parameter
$\bm{W}_{\tau,\,\sigma(i)}$, $\bm{W}'_{\tau,\,\sigma(i)}$, and
$\bm{b}_{\tau,\,\sigma(i)}$ is learned once for each ordered pair of neighbor
type $\tau$ and receiving type $\sigma(i)$, and used across all relations of
those types.

Finally, each block closes with the position-wise feed-forward network
$\operatorname{FFN}$, whose weights
$\bm{W}^{\mathrm{ffn}}_{\sigma(i)}, \bm{b}^{\mathrm{ffn}}_{\sigma(i)},
\bm{W}'^{\,\mathrm{ffn}}_{\sigma(i)}, \bm{b}'^{\,\mathrm{ffn}}_{\sigma(i)}$
are carried per node type in each block, shared by every node of that type.

\subsection{GridSFM Output Heads}
\label{sec:output-appendix}

To complete GridSFM, the output heads are composed of a fusion stage that
gathers each node's neighborhood and the grid as a whole, followed by a
per-quantity head whose output is projected onto the operational limits. For a bus $i \in \mathcal N_0$, write
\begin{equation}
\label{eq:fusion-agg}
    \bm{a}^\tau_i := \operatorname*{mean}_{j \in \mathcal{N}(i) \cap \mathcal{V}_\tau}
      s_{ij}\, \bm{z}_j \,, \qquad
    \tau \in \{\mathrm{ac}, \mathrm{tr}, \mathrm{gen}, \mathrm{load},
    \mathrm{shunt}\} \,,
\end{equation}
for the signed mean over its neighbors of each adjacent type, recalling that
$s_{ij} = +1$ on the device edges. Cycles are not adjacent to buses, so we
reach them through the branches they contain by defining
\begin{equation}
\label{eq:fusion-cycle}
    \bm{a}^{\mathrm{cycle}}_i := \sum_{\tau \in \{\mathrm{ac}, \mathrm{tr}\}}
      \operatorname*{mean}_{j \in \mathcal{N}(i) \cap \mathcal{V}_\tau} s_{ij}
      \Bigl( \operatorname*{mean}_{j' \in \mathcal{N}(j) \cap
      \mathcal{V}_{\mathrm{cycle}}} s_{jj'}\, \bm{z}_{j'} \Bigr) \,.
\end{equation}
The fusion then produces a bus readout, a generator readout, and a single
global summary,
\begin{subequations}
\label{eq:fusion}
\begin{align}
    \bm{z}^{\mathrm{bus}}_i =&\, \operatorname{LayerNorm}_{\mathrm{bus}}\Bigl(
      \sum_{\tau} \bm{W}_{\mathrm{bus},\tau}^{\rm fuse}\, \bm{a}^\tau_i
      + \bm{W}'_{\mathrm{bus}} \bm{z}_i \Bigr) \,,   \label{eq:fusion-a}\\
    \bm{z}^{\mathrm{gen}}_i =&\, \operatorname{LayerNorm}_{\mathrm{gen}}\bigl(
      \bm{W}_{\mathrm{gen}}^{\rm fuse} \bm{z}^{\mathrm{bus}}_j
      + \bm{W}'_{\mathrm{gen}} \bm{z}_i \bigr) \,,  \label{eq:fusion-b} \\
    \bm{z}^{\mathrm{glob}} =&\, \operatorname{LayerNorm}_{\mathrm{glob}}\Bigl(
      \bm{W}_{\mathrm{glob}}^{\rm fuse} \bm{z}^{\mathrm{pool}} \Bigr) \,, \label{eq:fusion-c}\\
    \bm{z}^{\mathrm{pool}} =&\, \operatorname{col}\Bigl(
      \operatorname*{mean}_{k \in \mathcal N_0} \bm{z}^{\mathrm{bus}}_k \,,
      \operatorname*{max}_{k \in \mathcal N_0} \bm{z}^{\mathrm{bus}}_k \,, \nonumber \\
      &\, \qquad \quad
      \bigl\{ \operatorname*{mean}_{k \in \mathcal{V}_\tau} \bm{z}_k\,,
      \operatorname*{max}_{k \in \mathcal{V}_\tau} \bm{z}_k
      \bigr\}_{\tau \in \Sigma_{\mathrm{p}}} \Bigr) \,,\label{eq:fusion-d}
\end{align}
\end{subequations}
where in \eqref{eq:fusion-a} the sum runs over the six source types of
\eqref{eq:fusion-agg} and \eqref{eq:fusion-cycle}, each with its own learnable
parameters $\bm{W}_{\mathrm{bus},\tau}$. In \eqref{eq:fusion-b}, $i$ is a
generator ($i \in \mathcal{V}_{\mathrm{gen}}$) and $j$ is its unique adjacent
bus, $\mathcal{N}(i) = \{j\}$. In \eqref{eq:fusion-d}, $\Sigma_{\mathrm{p}}$
runs over the types carrying a positional encoding, that is,
$\Sigma_{\mathrm{p}} := \{\mathrm{bus}, \mathrm{line}, \mathrm{transformer},
\mathrm{cycle}\}$.

Each control is then produced by a two-layer network with GELU activation,
\begin{subequations}
\label{eq:heads}
\begin{align}
    \widehat p^{\mathrm g}_i =&\, \Pi_{[\underline p^{\mathrm g}_i,
      \overline p^{\mathrm g}_i]}\bigl(\check p^{\mathrm g}_i
      + \tfrac{1}{2}(\underline p^{\mathrm g}_i
      + \overline p^{\mathrm g}_i)\bigr) \,,          \label{eq:heads-a}\\
    \widehat v_i =&\, \Pi_{[\underline v_i, \overline v_i]}
      \bigl(\check v_i\bigr) \,,                      \label{eq:heads-b}\\
    \check p^{\mathrm g}_i =&\, (\bm{w}^{\mathrm p}_2)^{\!\top}
      \operatorname{GELU}\bigl(\bm{W}^{\mathrm p}_1 \bm{\eta}^{\mathrm p}_i
      + \bm{b}^{\mathrm p}_1\bigr) + b^{\mathrm p}_2 \,,  \label{eq:heads-c}\\
    \check v_i =&\, (\bm{w}^{\mathrm v}_2)^{\!\top}
      \operatorname{GELU}\bigl(\bm{W}^{\mathrm v}_1 \bm{\eta}^{\mathrm v}_i
      + \bm{b}^{\mathrm v}_1\bigr) + b^{\mathrm v}_2 \,,  \label{eq:heads-d}\\
    \bm{\eta}^{\mathrm p}_i =&\, [\,\bm{z}^{\mathrm{gen}}_i ;
      \bm{z}^{\mathrm{glob}} ; \bm{\chi}\,] \,,           \label{eq:heads-e}\\
    \bm{\eta}^{\mathrm v}_i =&\, [\,\bm{z}^{\mathrm{bus}}_i ;
      \bm{z}^{\mathrm{glob}} ; \bm{\chi} ; v^{\mathrm{set}}_i\,] \,,
                                                          \label{eq:heads-f}
\end{align}
\end{subequations}
where \eqref{eq:heads-a} and \eqref{eq:heads-c} are evaluated at each
in-service generator $i$ and \eqref{eq:heads-b} and \eqref{eq:heads-d} at each voltage-controlled bus. $\Pi_{[a, b]}$ denotes the projection onto the
interval $[a,b]$, $\bm{W}^{\mathrm p}_1, \bm{W}^{\mathrm v}_1$ are learnable
weight matrices, $\bm{w}^{\mathrm p}_2, \bm{w}^{\mathrm v}_2$ learnable
vectors, and $b^{\mathrm p}_2, b^{\mathrm v}_2$ learnable scalars.
$\bm{\chi} \in \mathbb{R}^{7}$ is a fixed per-instance descriptor of grid scale
and loading (bus and generator counts, active and reactive load as fractions of
total generation capability, the mean voltage band, and branch rating and
admittance magnitudes) and $v^{\mathrm{set}}_i$ is the mean voltage setpoint of
the generators at bus $i$, taken as zero where there are none.

\section{\raggedright Additional Experiments}\label{app:experiments} %

\subsection{Hyperparameters: Dataset Generation}\label{appendix:dataset-design}

The pretrained dataset is built by considering base topologies from OPFData \cite{lovett2024opfdatalargescaledatasetsac}, PGLib \cite{pglib}, Texas A\&M Synthetic Grids \cite{7725528} and Microsoft Research Synthetic topologies \cite{britto2026buildingpowergridmodels} and then applying the following perturbations:
\begin{enumerate}[label=(\roman*)]
\item \emph{Demand.} Scale each load as
\[
P^d_\ell \leftarrow \sigma \varepsilon^P_\ell P^d_\ell,
\qquad
Q^d_\ell \leftarrow \sigma \varepsilon^Q_\ell Q^d_\ell,
\qquad \ell \in \mathcal{L},
\]
where $\sigma \sim \mathcal{U}[\sigma_{\mathrm{lo}},\sigma_{\mathrm{hi}}]$ is system-wide and
$\varepsilon^P_\ell,\varepsilon^Q_\ell \sim \mathcal{U}[0.9,1.1]$ are line independent.
We require
\[
\sigma_{\mathrm{hi}} < \bar{\sigma}
= \frac{\sum_i P_i^{\max}}{\sum_\ell P^d_\ell},
\]
since demand above $\bar{\sigma}$ is infeasible even without losses.

\item \emph{Generation cost.} Randomly select $40\%$ of the in-service generators and pertmute their quadratic cost coefficients. Namely, if $c_i=(c_{i2},c_{i1},c_{i0})$, where the cost of unit $i$ is $c_{i2}p_i^2+c_{i1}p_i+c_{i0}$, then permute $c_{i2}$ across all $i$ between the $40\%$ subset of available coefficients for each term. Repeat the same for $c_{i1}$ and $c_{i0}$.

\item \emph{Congestion.} Randomly select $10\%$ of rated branches, ($28\%$ for \texttt{case500\_goc}), and scale their thermal limits as
\[
\mathrm{rate}_{a,b,c} \leftarrow \phi\,\mathrm{rate}_{a,b,c},
\qquad
\phi \sim \mathcal{U}[\phi_{\mathrm{lo}},\phi_{\mathrm{hi}}],
\]
with $\phi$ drawn independently for each branch.

\item \emph{Voltage bands.} Randomly select $10\%$ of buses and narrow their voltage bands:
\[
V_i^{\min} \leftarrow V_i^{\min}+\delta_i^-,
\qquad
V_i^{\max} \leftarrow V_i^{\max}-\delta_i^+,
\]
where $\delta_i^\pm \sim \mathcal{U}[0,0.01]$ p.u. Revert any change that inverts the band.

\item \emph{Outages.} Consider in-service generators with
$P_i^{\max}>10^{-2}$ p.u. Draw the number of outages ($(3, 2, 1)$ or $(6, 4, 2)$ for \texttt{ACTIVSg10k}) with probabilities $(0.7,0.2,0.1)$, then trip that many generators uniformly at random. Limit the number of outages so that at least two generators remain in service. 
\end{enumerate}
where the demand and generation cost perturbation are always applied and the congestion, voltage band, and outage are applied with uniform probabilities $0.2, 0.15,$ and $0.3$ respectively. 
These perturbations are designed to force GridSFM to learn not only the least-cost generation dispatch, but also how to operate across a broad range of active constraint regimes that stress different components of the network. We provide the full hyperparameter list for dataset generation in Table \ref{tab:hyperparams-dataset-gen}.

\begin{table}[H]
\centering
\caption{\vspace{1em} \parbox{\columnwidth}{\normalfont\footnotesize {Scenario generation hyperparameters using the perturbations as in Section \ref{subsec:data-gen}. The modes are drawn independently and compose, so one
scenario carries two to four of them at once: demand and cost are applied to every
scenario, while congestion, voltage bands and outages fire with probability $0.20$,
$0.15$ and $0.30$ respectively. }}}
\label{tab:hyperparams-dataset-gen}
\scriptsize
\setlength{\tabcolsep}{3pt}
\renewcommand{\arraystretch}{1.15}
\setlength{\arrayrulewidth}{0.6pt}
\setlength{\heavyrulewidth}{0.9pt}
\setlength{\lightrulewidth}{0.55pt}
\arrayrulecolor{gray!60}
\begin{tabular}{@{}l;{1.5pt/1.2pt}cccc@{}}
\arrayrulecolor{black}
\toprule
& & & & \\
\noalign{\vskip -8.6pt}
Hyperparameter & \texttt{case500\_goc} & \texttt{Texas2k} & \texttt{case6470\_rte} & \texttt{ACTIVSg10k} \\
\noalign{\vskip 0.8pt}
\arrayrulecolor{gray!50}\cdashline{1-5}[2.4pt/2.2pt]\arrayrulecolor{black}
\noalign{\vskip 0.8pt}
Demand $\sigma$ & $[0.8,1.2]$ & $[0.8,1.3]$ & $[0.8,1.2]$ & $[0.8,1.1]$ \\
Per-load & \multicolumn{4}{c}{$\varepsilon^P_\ell,\varepsilon^Q_\ell\sim\mathcal U[0.9,1.1]$} \\
Cost $c_i$ & \multicolumn{4}{c}{$40\%$ of in-service units} \\
Congestion & \multicolumn{4}{c}{$10\%$ of rated branches, $\phi\sim\mathcal U[0.70,0.95]$} \\
Voltage bands & \multicolumn{4}{c}{$10\%$ of buses, $\delta_i^{\pm}\sim\mathcal U[0,0.01]$ p.u.} \\
Outages $k$ & $(1,2,3)$ & $(1,2,3)$ & $(1,2,3)$ & $(2,4,6)$ \\
$k$ probabilities & \multicolumn{4}{c}{$(0.7,0.2,0.1)$} \\
\bottomrule
\end{tabular}
\setlength{\arrayrulewidth}{0.4pt}
\end{table}

\subsection{Hyperparameters: Training and Fine-tuning}\label{appendix:hyperparams}
\begin{table}[H]
\centering
\caption{\vspace{1em} \parbox{\columnwidth}{\normalfont\footnotesize Pretraining and fine-tuning loss parameters. For \texttt{case6470\_rte}, the control weight is geometrically annealed from $13.05$ to $85$ over the first $100$ epochs to ensure convergence.}}
\label{tab:hyperparams-finetuning}
\scriptsize
\setlength{\tabcolsep}{2pt}
\renewcommand{\arraystretch}{1.15}
\setlength{\arrayrulewidth}{0.6pt}
\setlength{\heavyrulewidth}{0.9pt}
\setlength{\lightrulewidth}{0.55pt}
\arrayrulecolor{gray!60}
\begin{tabular}{@{}l;{1.5pt/1.2pt}cccc@{}}
\arrayrulecolor{black}
\toprule
& & & & \\
\noalign{\vskip -8.6pt}
Hyperparameter & \texttt{case500\_goc} & \texttt{Texas2k} & \texttt{case6470\_rte} & \texttt{ACTIVSg10k} \\
\noalign{\vskip 0.8pt}
\arrayrulecolor{gray!50}\cdashline{1-5}[2.4pt/2.2pt]\arrayrulecolor{black}
\noalign{\vskip 0.8pt}
\makecell[l]{Generation cost\\$w_c$ ($\times10^{-6}$)} & $1.16$ & $2.18$ & $63.6$ & $0.46$ \\
Equality violation $w_g$ & $10$ & $10$ & $10.4$ & $10$ \\
Inequality violation $w_h$ & $10$ & $10$ & $14.5$ & $10$ \\
Power flow failure $w_r$ & $170$ & $144$ & $200$ & $2000$ \\
Control weight $w_u$ & $32$ & $1600$ & $13.05\!\to\!85$ & $300$ \\
Learning rate & $10^{-4}$ & $10^{-4}$ & $10^{-4}$ & $3\!\times\!10^{-4}$ \\
Batch size & $16$ & $4$ & $4$ & $4$ \\
\noalign{\vskip 0.8pt}
\arrayrulecolor{gray!50}\cdashline{1-5}[2.4pt/2.2pt]\arrayrulecolor{black}
\noalign{\vskip 0.8pt}
Dispatch weight $w_p$ & \multicolumn{4}{c}{$1.0$} \\
Voltage weight $w_v$ & \multicolumn{4}{c}{$1.0$} \\
Mean weight $w_\mu$ & \multicolumn{4}{c}{$1.0$} \\
Top-$k$ weight $w_\tau$ & \multicolumn{4}{c}{$1.0$} \\
Max weight $w_\infty$ & \multicolumn{4}{c}{$0.3$} \\
Top-$k$ generators & \multicolumn{4}{c}{$k/|\mathcal A| = 0.0435$} \\
Top-$k$ buses & \multicolumn{4}{c}{$k/|\mathcal N_v| = 0.0254$} \\
Dispatch tolerance $\varepsilon_p$ & \multicolumn{4}{c}{$10^{-2}$ p.u.} \\
Voltage tolerance $\varepsilon_v$ & \multicolumn{4}{c}{$10^{-3}$ p.u.} \\
\bottomrule
\end{tabular}
\setlength{\arrayrulewidth}{0.4pt}
\end{table}

We provide the full hyperparameter list for the training and fine-tuning losses in Table \ref{tab:hyperparams-finetuning}. Moreover, the code to replicate the results and models are publicly available (see \cite{bhan_gridsfm}). 

\subsection{Hyperparameters: Baseline Implementations}\label{sec:hyperparams-baselines}

\begin{table}[H]
\centering
\caption{\vspace{1em} \parbox{\columnwidth}{\normalfont Baseline hyperparameters including both network size and the optimizer settings.}}
\label{tab:hyperparams-baselines}
\footnotesize
\setlength{\tabcolsep}{6pt}
\renewcommand{\arraystretch}{1.15}
\setlength{\arrayrulewidth}{0.6pt}
\setlength{\heavyrulewidth}{0.9pt}
\setlength{\lightrulewidth}{0.55pt}
\arrayrulecolor{gray!60}
\begin{tabular}{@{}p{13.5em};{1.5pt/1.2pt}c@{}}
\arrayrulecolor{black}
\toprule
\arrayrulecolor{gray!60}
& \\
\noalign{\vskip -10.4pt}
Hyperparameter & Value \\
\noalign{\vskip 0.8pt}
\arrayrulecolor{gray!50}\cdashline{1-2}[2.4pt/2.2pt]\arrayrulecolor{black}
\noalign{\vskip 0.8pt}
\arrayrulecolor{gray!60}
\multicolumn{1}{@{}l;{1.5pt/1.2pt}}{\itshape MLP (DeepOPF \cite{pan20})} & \\[1pt]
\arrayrulecolor{black}
Layers          & $4$ \\
Width           & $256$ \\
Activation      & GELU \\
Parameters      & \makecell{grid dependent\\(see MLP, Appendix~\ref{sec:hyperparams-baselines})} \\
\noalign{\vskip 0.8pt}
\arrayrulecolor{gray!50}\cdashline{1-2}[2.4pt/2.2pt]\arrayrulecolor{black}
\noalign{\vskip 0.8pt}
\arrayrulecolor{gray!60}
\multicolumn{1}{@{}l;{1.5pt/1.2pt}}{\itshape GNN} & \\[1pt]
\arrayrulecolor{black}
Layers          & $4$ GINEConv \cite{hu2020strategies} \\
Width           & $256$ \\
Activation      & GELU \\
Edge features   & $10$ \\
Parameters      & $0.68$M \\
\noalign{\vskip 0.8pt}
\arrayrulecolor{gray!50}\cdashline{1-2}[2.4pt/2.2pt]\arrayrulecolor{black}
\noalign{\vskip 0.8pt}
\arrayrulecolor{gray!60}
\multicolumn{1}{@{}l;{1.5pt/1.2pt}}{\itshape GridSFM without pretraining} & \\[1pt]
\arrayrulecolor{black}
GridBlocks      & $4$ \\
Width           & $64$ \\
Parameters      & $1.96$M \\
\noalign{\vskip 0.8pt}
\arrayrulecolor{gray!50}\cdashline{1-2}[2.4pt/2.2pt]\arrayrulecolor{black}
\noalign{\vskip 0.8pt}
\arrayrulecolor{gray!60}
\multicolumn{1}{@{}l;{1.5pt/1.2pt}}{\itshape Optimization} & \\[1pt]
\arrayrulecolor{black}
Optimizer       & AdamW \\
Learning rate   & $10^{-3}$ \\
Weight decay    & $10^{-4}$ \\
Schedule        & cosine to $0.01\times$ lr \\
\bottomrule
\end{tabular}
\setlength{\arrayrulewidth}{0.4pt}
\end{table}

We being by briefly explaining the design and training methodology used for each baseline. Table \ref{tab:hyperparams-baselines} provides a summary of all the hyperparameters and training configurations for each baseline.

    \textbf{Multi-Layer Perceptron (MLP):} We train an individual MLP per grid in Section \ref{sec:numerical} using a DeepOPF-style~\cite{pan20} framework. For each grid's scenario varying quantities (per-bus demand, voltage bands and
    shunt/generator aggregates, generator unit availability flag, capability limits, setpoint and cost coefficients), we flatten the scenario conditions into a single input vector. This is then passed through $4$ hidden layers of width $256$ with GELU activations to $P^{\mathrm g}$ at every generator slot and $V_m$ at every
    control bus. Input widths are grid-specific, so this is the network size is grid dependent: $1.07$M parameters on \texttt{case500\_goc}, $3.30$M on
    \texttt{Texas2k}, $6.44$M on \texttt{case6470\_rte} and $10.76$M on
    \texttt{ACTIVSg10k}. All models are trained on $1$k unseen problems per grid generated following the same design as Section \ref{subsec:data-gen}.

    \textbf{GNN:} The same width, depth and activation as the MLP, but with the
    four dense layers replaced by four GINEConv message-passing layers \cite{hu2020strategies} over the bus graph, each branch carrying its ten leading raw columns (angle limits, series impedance, charging susceptance and thermal ratings) as
    unsigned edge features. Every quantity is per-cell rather than per-grid, so the same $0.68$M parameters
    ($678{,}658$) serve all four grids.

    \textbf{Linux Foundation:} We use the harness of 
    \texttt{gridfm-graphkit}~\cite{Pue26} ($12$ layers, hidden width $48$, $8$ attention
    heads, $20.07$M parameters), trained with their code, their data pipeline, with their default hyperparameters and loss on
    $200$k \texttt{case500\_goc} scenarios for $200$ epochs. As it is trained on
    \texttt{case500\_goc} alone, its \texttt{Texas2k}, \texttt{case6470\_rte} and
    \texttt{ACTIVSg10k} rows are zero-shot transfers to grids of a size and topology it
    has never seen.
    
    \textbf{GridSFM (no pretraining):} 
    The GridSFM architecture of
    Section~\ref{sec:foundation-model-design} initialized randomly and trained per grid on the same
    $1$k instances as the MLP and GNN models with the same pretraining loss~\eqref{eq:pretrain-loss}. We use $4$
    GridBlocks of width $64$ ($1.96$M parameters) rather than the released backbone's
    $8\times128$ ($15.15$M).

\subsection{Additional fine-tuning results}\label{sec:finetuning-appendix}
\subsubsection{Additional training metrics}
We additionally present full metrics for training on all grids with all baselines in Table \ref{tab:violation-supp} including both the fine-tuning performance with various amounts data as well as the inequality violations.

\subsubsection{GPU (MadNLP) solver results}
Since we report iterations, it is worth additionally presenting the model warm-start points for solvers beyond Ipopt. To do this, we considered the recently popular MadNLP solver \cite{shin2024accelerating} which is optimized to take advantage of the performance of GPUs. In Table \ref{tab:warmstart-gpu}, we find qualitatively, there is very little difference between the GPU and CPU solver as projection takes less iterations then full AC-OPF solves and all the finetuned model warm start points perform similarly. Lastly, we provide full wall-clock time breakdowns of the solver time and GridSFM inference in Table \ref{tab:compute-cost}.

\subsubsection{Ablation study on fine-tuning design}\label{appendix:ablation}

We present two ablation studies for the fine-tuning design. They aim to answer the questions:
\begin{enumerate}
    \item Why not just finetune with the same pretraining loss?
    \item Does the merit function really contribute significantly to the gradient in such a multi-term loss function?
\end{enumerate}

\begin{figure}[H]
    \centering
    \includegraphics[width=\linewidth]{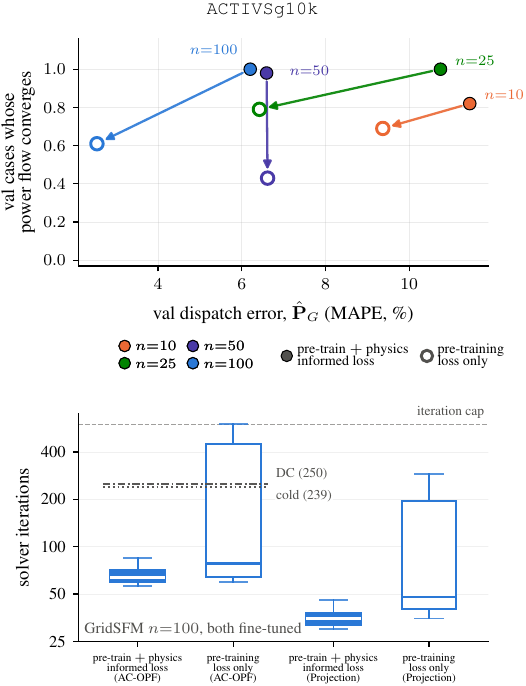}
    \caption{Performance comparison when fine-tuning with the pretrained loss in \eqref{eq:pretrain-loss} and the physics informed fine tune loss in \eqref{eq:finetune-loss}. The top row indicates power flow convergence vs the dispatch error on the validation set. The bottom showcases the performance of the resulting models control variables as both a warm start for AC-OPF and feasible projection over $100$ evaluation cases. $n$ refers to the amount of data used for fine-tuning.}
    \label{fig:ablation-pretrain}
\end{figure}

Both questions can be answered with the extrapolated grids \texttt{ACTIVSg10k} and \texttt{case6470\_rte}. In particular, in Figure \ref{fig:ablation-pretrain}, we showcase that the performance of fine-tuning with just the pretraining loss does not encourage points that satisfy the inequality constranits. Hence, when using the model as a warm start for both AC-OPF and projection, one either (a) does not complete the power flow or (b) completes the power flow, but at a worse point the the models finetuned on the full inequality penalized elastic loss. 

Furthermore, to highlight the affect of the power flow penalty term in \eqref{eq:finetune-loss}, we train finetuned models on \texttt{case6470\_rte}, a grid that is intentionally challenging to complete powerflow as indicative by the fact that a flat start outperforms a DC start for an AC-OPF solver. Figure \ref{fig:merit-v-no-merit}  confirms this, showcasing that without the power flow residual term, the evaluation cases almost never complete and hence actively hurt the AC-OPF performance compared to that of the models finetuned with the merit term.

\begin{figure}[H]
    \centering
    \includegraphics[width=\linewidth]{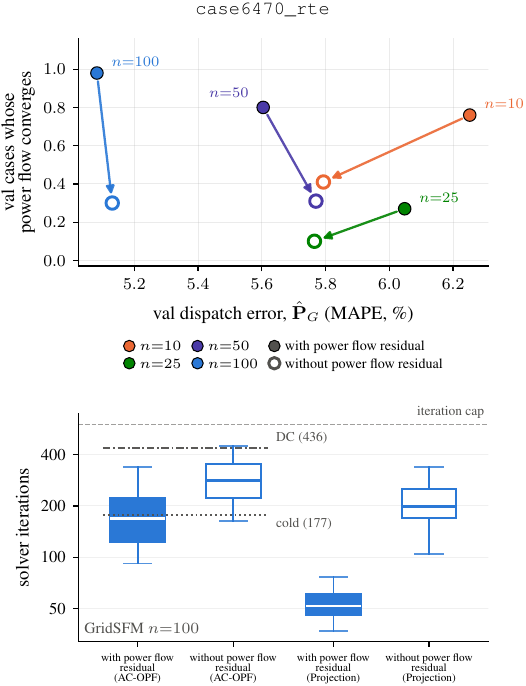}
    \caption{Ablation study on the affect of the power flow residual in the fine-tuning loss. The top row indicates power flow convergence vs the dispatch error on the validation set compared to between models finetuned with and without the power-flow residual term in the loss \eqref{eq:finetune-loss}. The bottom figure highlights the performance of resulting models as a warm start for both AC-OPF and feasible projection with and without the power flow residual term in the fine-tuning loss over $100$ evaluation cases. $n$ refers to the amount of data used for fine-tuning.}
    \label{fig:merit-v-no-merit}
\end{figure}

\newpage

\begin{table*}[!t]
\centering
\caption{\vspace{1em} \parbox{\textwidth}{\normalfont  Closed-state constraint violations, on the \emph{same} 100 held-out scenarios per grid that the warm-start tables solve. Each row is the operating point obtained by closing that model's predicted controls $(\widehat{\bm {p}}^g,\widehat{\bm v}_{\mathcal{G}})$ with a Newton power flow. Each violation column reports how far that point exceeds one class of inequality constraint, divided by the physical scale of that class so all four read as percentages. $\Delta Q$: reactive power outside $[\underline Q,\overline Q]$ at a control bus, as a fraction of that bus's reactive range. $\Delta V$: voltage magnitude outside $[\underline V_i,\overline V_i]$, as a fraction of nominal ($1$~p.u.). $\Delta P$: generator active power outside $[\underline P_{G_i},\overline P_{G_i}]$, as a fraction of the unit's range, and $\Delta S$:, apparent branch flow above its rating, as a fraction of $\overline S_\ell$. Each is averaged over \emph{all} elements of its class, violating or not, so the column is intensive and comparable across grids of very different size. The cost gap is $100(\bar c-\bar c^{\star})/\bar c^{\star}$, where $\bar c$ is the mean generation cost of the closed state over the scored scenarios and $\bar c^{\star}$ the mean cost of the reference AC-OPF optimum on those same scenarios. It is signed, so a negative value is a cheaper but, its an infeasible operating point rather than a better one. $\dagger$ indicates not every control for that model converged in the power flow, hence violation calculation is only over converged points. }}
\label{tab:violation-supp}
\setlength{\tabcolsep}{5pt}
\renewcommand{\arraystretch}{1.05}
\footnotesize
\begin{tabular}{@{}lrrrrrrrrr@{}}
\toprule
 & \multicolumn{5}{c}{violation (\% of capability)} & \multicolumn{4}{c}{accuracy} \\
\cmidrule(lr){2-6}\cmidrule(lr){7-10}
Operating point & mean & $\Delta Q$ & $\Delta S$ & $\Delta V$ & $\Delta P$ & $P_G$ MAPE & $V$ MAE & cost gap & power flow conv. \\
\midrule
\multicolumn{10}{@{}l}{\itshape \texttt{case500\_goc}} \\
\quad MLP ($n{=}1000$) & 8.452 & 32.784 & 0.0241 & 0.0041 & 0.9978 & 19.41 & 0.00699 & +6.60 & 100 \\
\quad GNN ($n{=}1000$) & 9.394 & 37.106 & 0.0104 & 0.0014 & 0.4581 & 4.89 & 0.00588 & -0.37 & 100 \\
\quad GridSFM, no pretrain ($n{=}1000$) & 12.013 & 48.043 & 0.0061 & 0.0015 & 0.0000 & \textbf{2.13} & \textbf{0.00479} & -0.04 & 100 \\
\quad Linux Foundation & 13.055 & 52.140 & 0.0784 & 0.0001 & 0.0000 & 7.79 & 0.01234 & -7.92 & 100 \\
\noalign{\vskip 0.8pt}
\arrayrulecolor{gray!50}\cdashline{1-10}[2.4pt/2.2pt]\arrayrulecolor{black}
\noalign{\vskip 0.8pt}
\quad GridSFM (zero-shot) & 20.208 & 80.822 & 0.0101 & 0.0013 & 0.0000 & 3.83 & 0.00648 & -0.17 & 100 \\
\quad GridSFM ($n{=}10$) & 19.940 & 79.745 & 0.0120 & 0.0032 & 0.0000 & 4.25 & 0.00726 & -1.57 & 100 \\
\quad GridSFM ($n{=}25$) & 8.658 & 34.617 & 0.0124 & 0.0013 & 0.0000 & 3.96 & 0.00701 & -1.48 & 100 \\
\quad GridSFM ($n{=}50$) & 7.402 & 29.595 & 0.0115 & 0.0012 & 0.0000 & 3.73 & 0.00676 & -0.32 & 100 \\
\quad GridSFM ($n{=}100$) & \textbf{6.475} & 25.885 & 0.0135 & 0.0012 & 0.0000 & 3.41 & 0.00638 & -0.94 & 100 \\
\midrule
\multicolumn{10}{@{}l}{\itshape \texttt{Texas2k}} \\
\quad MLP ($n{=}1000$) & 32.832 & 130.324 & 0.7092 & 0.0019 & 0.2941 & 3.38 & 0.00689 & +4.06 & 100 \\
\quad GNN ($n{=}1000$) & 51.401 & 204.421 & 0.7413 & 0.0003 & 0.4425 & 2.60 & 0.01085 & -0.39 & 100 \\
\quad GridSFM, no pretrain ($n{=}1000$) & 10.655 & 41.986 & 0.6318 & 0.0006 & 0.0000 & \textbf{0.65} & \textbf{0.00454} & -0.27 & 100 \\
\quad Linux Foundation$^\dagger$ & 34.868 & 138.943 & 0.4850 & 0.0447 & 0.0000 & 23.71 & 0.03686 & +62.54 & 75 \\
\noalign{\vskip 0.8pt}
\arrayrulecolor{gray!50}\cdashline{1-10}[2.4pt/2.2pt]\arrayrulecolor{black}
\noalign{\vskip 0.8pt}
\quad GridSFM (zero-shot) & 14.854 & 58.900 & 0.5130 & 0.0009 & 0.0000 & 2.11 & 0.00993 & -0.53 & 100 \\
\quad GridSFM ($n{=}10$) & 12.616 & 50.288 & 0.1738 & 0.0008 & 0.0000 & 2.43 & 0.01339 & +0.66 & 100 \\
\quad GridSFM ($n{=}25$) & 11.634 & 46.359 & 0.1753 & 0.0015 & 0.0000 & 2.61 & 0.01144 & -1.69 & 100 \\
\quad GridSFM ($n{=}50$) & 6.369 & 25.405 & 0.0705 & 0.0009 & 0.0000 & 2.02 & 0.01025 & +0.64 & 100 \\
\quad GridSFM ($n{=}100$) & \textbf{3.830} & 15.301 & 0.0195 & 0.0004 & 0.0000 & 1.37 & 0.00768 & -0.04 & 100 \\
\midrule
\multicolumn{10}{@{}l}{\itshape \texttt{case6470\_rte}} \\
\quad MLP ($n{=}1000$) & 52.647 & 209.505 & 0.3075 & 0.0021 & 0.7729 & 30.45 & 0.01099 & +18.99 & 100 \\
\quad GNN ($n{=}1000$) & 64.552 & 256.860 & 0.1028 & 0.0025 & 1.2433 & 10.65 & 0.01156 & +1.28 & 100 \\
\quad GridSFM, no pretrain ($n{=}1000$)$^\dagger$ & -- & -- & -- & -- & -- & \textbf{3.83} & \textbf{0.01080} & -- & 0 \\
\quad Linux Foundation$^\dagger$ & -- & -- & -- & -- & -- & 48.71 & 0.03960 & -- & 0 \\
\noalign{\vskip 0.8pt}
\arrayrulecolor{gray!50}\cdashline{1-10}[2.4pt/2.2pt]\arrayrulecolor{black}
\noalign{\vskip 0.8pt}
\quad GridSFM (zero-shot)$^\dagger$ & 3476.622 & 1.39e+04 & 10.5862 & 2.0160 & 0.0000 & 5.65 & 0.02035 & -1.71 & 3 \\
\quad GridSFM ($n{=}10$)$^\dagger$ & 169.412 & 677.253 & 0.2648 & 0.1297 & 0.0000 & 5.92 & 0.01648 & -1.61 & 62 \\
\quad GridSFM ($n{=}25$)$^\dagger$ & 867.159 & 3461.940 & 5.8910 & 0.8048 & 0.0000 & 5.64 & 0.01493 & +2.12 & 28 \\
\quad GridSFM ($n{=}50$)$^\dagger$ & 252.589 & 1008.865 & 1.3890 & 0.1015 & 0.0000 & 5.31 & 0.01495 & -0.43 & 74 \\
\quad GridSFM ($n{=}100$)$^\dagger$ & \textbf{38.747} & 154.923 & 0.0433 & 0.0218 & 0.0000 & 4.78 & 0.01332 & +0.12 & 98 \\
\midrule
\multicolumn{10}{@{}l}{\itshape \texttt{ACTIVSg10k}} \\
\quad MLP ($n{=}1000$)$^\dagger$ & 21.950 & 87.276 & 0.2079 & 0.0038 & 0.3115 & 14.82 & 0.00350 & +2.36 & 94 \\
\quad GNN ($n{=}1000$)$^\dagger$ & 26.182 & 103.117 & 0.1986 & 0.0014 & 1.4104 & 6.77 & 0.00563 & +0.44 & 98 \\
\quad GridSFM, no pretrain ($n{=}1000$)$^\dagger$ & 154.513 & 607.376 & 9.1390 & 1.5362 & 0.0000 & \textbf{0.76} & \textbf{0.00312} & +0.06 & 33 \\
\quad Linux Foundation$^\dagger$ & 100.411 & 397.778 & 3.4935 & 0.3722 & 0.0000 & 26.61 & 0.02107 & -14.71 & 26 \\
\noalign{\vskip 0.8pt}
\arrayrulecolor{gray!50}\cdashline{1-10}[2.4pt/2.2pt]\arrayrulecolor{black}
\noalign{\vskip 0.8pt}
\quad GridSFM (zero-shot)$^\dagger$ & -- & -- & -- & -- & -- & 33.39 & 0.03391 & -- & 0 \\
\quad GridSFM ($n{=}10$)$^\dagger$ & 24.942 & 99.001 & 0.7395 & 0.0279 & 0.0000 & 11.04 & 0.01015 & +0.60 & 77 \\
\quad GridSFM ($n{=}25$) & 5.427 & 21.703 & 0.0042 & 0.0004 & 0.0000 & 10.52 & 0.00792 & +2.96 & 100 \\
\quad GridSFM ($n{=}50$) & \textbf{1.067} & 4.265 & 0.0023 & 0.0007 & 0.0000 & 6.59 & 0.00809 & +1.25 & 100 \\
\quad GridSFM ($n{=}100$) & 2.327 & 9.303 & 0.0024 & 0.0004 & 0.0000 & 6.11 & 0.00528 & +1.18 & 100 \\
\bottomrule
\end{tabular}
\end{table*}

\begin{table*}[!t]
\centering
\caption{\parbox{\textwidth}{\normalfont\footnotesize \vspace{1em}Performance of various warm-starting points on both AC-OPF \eqref{eq:AC-OPF} and the feasible projection \eqref{eq:AC-OPF-projection} using the MadNLP (GPU) solver for $100$ evaluation problems across various grid sizes. This is the GPU counterpart of Table~\ref{tab:warmstart}: identical $100$ test cases solved with MadNLP on the GPU (NVIDIA RTX PRO 6000 Blackwell).}}
\label{tab:warmstart-gpu}
\setlength{\tabcolsep}{3pt}
\renewcommand{\arraystretch}{1.15}
\footnotesize
\setlength{\arrayrulewidth}{0.6pt}
\arrayrulecolor{gray!60}
\begin{tabular}{@{}l;{1.5pt/1.2pt}cccc;{1.5pt/1.2pt}cccc;{1.5pt/1.2pt}cccc;{1.5pt/1.2pt}cccc@{}}
\arrayrulecolor{black}
\toprule
Warm start & \multicolumn{4}{c}{\texttt{case500\_goc}} & \multicolumn{4}{c}{\texttt{Texas2k}} & \multicolumn{4}{c}{\texttt{case6470\_rte}} & \multicolumn{4}{c}{\texttt{ACTIVSg10k}} \\
\cmidrule(lr){2-5}\cmidrule(lr){6-9}\cmidrule(lr){10-13}\cmidrule(lr){14-17}
 & \makecell{90\%\\done} & \makecell{Speed-\\up} & \makecell{Conv.\\(\%)} & \makecell{Cost gap\\(\%)} & \makecell{90\%\\done} & \makecell{Speed-\\up} & \makecell{Conv.\\(\%)} & \makecell{Cost gap\\(\%)} & \makecell{90\%\\done} & \makecell{Speed-\\up} & \makecell{Conv.\\(\%)} & \makecell{Cost gap\\(\%)} & \makecell{90\%\\done} & \makecell{Speed-\\up} & \makecell{Conv.\\(\%)} & \makecell{Cost gap\\(\%)} \\
\midrule
\multicolumn{17}{@{}l}{\itshape AC-OPF warm start}\\[1pt]
Cold start (flat) & 146 & 1.00 & 100 & 0.00 & -- & 1.00 & 87 & 0.00 & -- & 1.00 & 83 & 0.00 & 465 & 1.00 & 94 & 0.00 \\
DC warm start & 68 & 3.10 & 100 & 0.00 & 79 & 2.25 & 100 & 0.00 & 341 & 2.63 & 95 & 0.00 & 157 & 2.72 & 98 & 0.00 \\
\noalign{\vskip 0.8pt}
\arrayrulecolor{gray!50}\cdashline{1-17}[2.4pt/2.2pt]\arrayrulecolor{black}
\noalign{\vskip 0.8pt}
MLP (DeepOPF \cite{pan20}) & 71 & 2.45 & 100 & 0.00 & 86 & 2.10 & 100 & 0.00 & 314 & 2.42 & 98 & 0.00 & -- & \textcolor{gray!60}{1.43\,$^{\dagger}$} & 8 & 0.00 \\
GNN & 67 & 3.30 & 100 & 0.00 & 85 & 2.09 & 99 & 0.00 & 319 & 2.79 & 98 & 0.00 & -- & \textcolor{gray!60}{1.22\,$^{\dagger}$} & 4 & 0.00 \\
Linux Foundation \cite{Pue26} & 68 & 2.79 & 100 & 0.00 & 114 & 1.51 & 99 & 0.00 & -- & \textcolor{gray!60}{1.52\,$^{\dagger}$} & 50 & 0.00 & -- & 1.29 & 87 & 0.00 \\
\noalign{\vskip 0.8pt}
\arrayrulecolor{gray!50}\cdashline{1-17}[2.4pt/2.2pt]\arrayrulecolor{black}
\noalign{\vskip 0.8pt}
GridSFM (no pretrain) & 65 & \textbf{3.59} & 100 & 0.00 & 71 & 2.56 & 100 & 0.00 & -- & \textcolor{gray!60}{0.80\,$^{\dagger}$} & 34 & 0.00 & -- & 1.37 & 71 & 0.00 \\
GridSFM (zero-shot) & 65 & 3.12 & 100 & 0.00 & 69 & 2.59 & 100 & 0.00 & -- & 1.08 & 56 & 0.00 & -- & -- & 0 & -- \\
GridSFM (fine-tuned) & \textbf{64} & 3.30 & 100 & 0.00 & \textbf{66} & \textbf{2.83} & 100 & 0.00 & \textbf{303} & \textbf{2.99} & 95 & 0.00 & \textbf{73} & \textbf{7.00} & 98 & 0.00 \\
\midrule
\multicolumn{17}{@{}l}{\itshape AC-OPF feasible projection}\\[1pt]
MLP (DeepOPF \cite{pan20}) & \textbf{29} & 4.93 & 100 & 9.41 & \textbf{38} & \textbf{4.73} & 100 & 2.55 & \textbf{43} & \textbf{11.89} & 100 & 32.87 & -- & \textcolor{gray!60}{1.88\,$^{\dagger}$} & 5 & 2.61 \\
GNN & 33 & 4.81 & 100 & 0.87 & 40 & 4.49 & 100 & 1.96 & 55 & 10.59 & 99 & 4.00 & -- & \textcolor{gray!60}{7.07\,$^{\dagger}$} & 1 & 0.41 \\
Linux Foundation \cite{Pue26} & 48 & 3.97 & 100 & 1.69 & 94 & 2.14 & 100 & 65.81 & -- & \textcolor{gray!60}{0.70\,$^{\dagger}$} & 18 & 41.61 & -- & \textcolor{gray!60}{0.95\,$^{\dagger}$} & 50 & 3.46 \\
\noalign{\vskip 0.8pt}
\arrayrulecolor{gray!50}\cdashline{1-17}[2.4pt/2.2pt]\arrayrulecolor{black}
\noalign{\vskip 0.8pt}
GridSFM (no pretrain) & 31 & \textbf{5.32} & 100 & 0.31 & 40 & 4.33 & 100 & 0.29 & -- & \textcolor{gray!60}{1.04\,$^{\dagger}$} & 27 & 1.51 & -- & 1.11 & 58 & 0.07 \\
GridSFM (zero-shot) & 34 & 5.11 & 100 & 0.59 & 43 & 4.16 & 100 & 0.87 & -- & \textcolor{gray!60}{1.03\,$^{\dagger}$} & 49 & 2.36 & -- & -- & 0 & -- \\
GridSFM (fine-tuned) & 35 & 4.84 & 100 & 0.55 & 41 & 4.17 & 100 & 0.66 & 46 & 10.73 & 98 & 1.96 & \textbf{45} & \textbf{8.39} & 100 & 0.58 \\
\bottomrule
\end{tabular}
\setlength{\arrayrulewidth}{0.4pt}
\end{table*}

\begin{table*}[!t]
\centering
\caption{\vspace{1em} \parbox{\linewidth}{\normalfont\footnotesize Computational cost of a solve on each device (AMD EPYC 9535 (CPU) and Nvidia RTX Pro 6000 Blackwell (GPU)), measured on $10$ evaluation problems per grid. Every computation is measured uncontested and serially throughout. Model inference time is assuming an unbatched size of $1$ and can be significantly sped up with batching. \mbox{--} marks a quantity that does not apply or was not measured. The GridSFM forward pass is the median of $30$ single-case forwards on an otherwise idle GPU after $10$ warm-up passes, synchronised around each call.
}}
\label{tab:compute-cost}
\setlength{\tabcolsep}{4pt}
\renewcommand{\arraystretch}{1.15}
\footnotesize
\setlength{\arrayrulewidth}{0.6pt}
\arrayrulecolor{gray!60}
\begin{tabular}{@{}l;{1.5pt/1.2pt}rr;{1.5pt/1.2pt}rr;{1.5pt/1.2pt}rr;{1.5pt/1.2pt}rr@{}}
\arrayrulecolor{black}
\toprule
 & \multicolumn{2}{c}{\texttt{case500\_goc}} & \multicolumn{2}{c}{\texttt{Texas2k}} & \multicolumn{2}{c}{\texttt{case6470\_rte}} & \multicolumn{2}{c}{\texttt{ACTIVSg10k}} \\
\cmidrule(lr){2-3}\cmidrule(lr){4-5}\cmidrule(lr){6-7}\cmidrule(lr){8-9}
 & ms/iter & s/prob & ms/iter & s/prob & ms/iter & s/prob & ms/iter & s/prob \\
\midrule
\multicolumn{9}{@{}l}{\itshape Solve, CPU (Ipopt)}\\[1pt]
\quad Cold start & 34.2 & 1.80 & 220.6 & 61.32 & 321.2 & 56.38 & 560.6 & 138.61 \\
\quad DC warm start & 23.0 & 1.25 & 255.2 & 14.56 & 380.1 & 146.05 & 503.1 & 125.34 \\
\quad GridSFM (fine-tuned) & 23.1 & 1.03 & 272.4 & 10.85 & 339.4 & 59.42 & 558.3 & 38.39 \\
\quad GridSFM $+$ projection & 24.2 & \textbf{0.66} & 241.5 & \textbf{6.74} & 472.9 & \textbf{25.32} & 699.8 & \textbf{24.52} \\
\noalign{\vskip 0.8pt}
\arrayrulecolor{gray!50}\cdashline{1-9}[2.4pt/2.2pt]\arrayrulecolor{black}
\noalign{\vskip 0.8pt}
\multicolumn{9}{@{}l}{\itshape Solve, GPU (MadNLP)}\\[1pt]
\quad Cold start & 4.4 & 0.55 & 16.6 & 1.85 & 11.4 & 4.78 & 19.2 & 4.98 \\
\quad DC warm start & 4.8 & 0.21 & 15.4 & 0.64 & 11.1 & 1.91 & 15.1 & 1.67 \\
\quad GridSFM (fine-tuned) & 4.7 & \textbf{0.18} & 16.1 & \textbf{0.56} & 9.8 & 1.55 & 19.0 & \textbf{0.96} \\
\quad GridSFM $+$ projection & 20.7 & 0.60 & 20.3 & 0.88 & 26.9 & \textbf{1.02} & 31.1 & 1.13 \\
\quad GPU data construction & -- & 0.02 & -- & 0.07 & -- & 0.23 & -- & 0.34 \\
\noalign{\vskip 0.8pt}
\arrayrulecolor{gray!50}\cdashline{1-9}[2.4pt/2.2pt]\arrayrulecolor{black}
\noalign{\vskip 0.8pt}
\multicolumn{9}{@{}l}{\itshape Warm-start production}\\[1pt]
\quad DC-OPF solve & -- & 1.43 & -- & 1.66 & -- & 3.54 & -- & 9.30 \\
\quad GridSFM forward pass & -- & 0.035 & -- & 0.035 & -- & 0.044 & -- & 0.053 \\
\quad GridSFM power-flow closure & -- & 0.009 & -- & 0.045 & -- & 0.132 & -- & 0.140 \\
\bottomrule
\end{tabular}
\setlength{\arrayrulewidth}{0.4pt}
\begin{minipage}{\linewidth}\vspace{3pt}\footnotesize
\end{minipage}
\end{table*}

\makeatletter
\setlength{\@fptop}{0pt}
\makeatother
\clearpage            
\makeatletter
\setlength{\@fptop}{0pt plus 1fil}
\makeatother

\end{document}